\documentclass[onecolumn, a4paper, unpublished]{quantumarticle}
\pdfoutput=1
\usepackage[utf8]{inputenc}
\usepackage{amsmath}
\usepackage{amssymb}
\usepackage{braket}
\usepackage{graphicx}
\usepackage{float}
\usepackage{braket}
\usepackage{bbold}
\usepackage[margin=1in]{geometry}
\usepackage{amsthm}
\usepackage[dvipsnames]{xcolor}
\usepackage{comment}
\usepackage{mathtools}
\usepackage{hyperref}
\usepackage{thm-restate}
\usepackage{thmtools}
\usepackage[capitalise]{cleveref}
\usepackage{tikz}
\usepackage{standalone}
\usepackage{textcomp}
\usepackage{cancel}
\usepackage{stmaryrd}
\usepackage{subcaption}
\usepackage[sorting=none]{biblatex}
\usepackage{float}
\bibliography{references}

\def\dbra#1{\mathinner{\langle\!\langle{#1}|}}
\def\dket#1{\mathinner{|{#1}\rangle\!\rangle}}
    \newcommand{\inprod}[2]{\langle #1 | #2 \rangle}
	\newcommand{\proj}[2]{\ket{#1}\bra{#2}}

\DeclareMathOperator{\tr}{tr}
\newcommand{\C}{\mathcal}

\newtheorem{defi}{Definition}
\newtheorem{lemma}{Lemma}

\newtheorem{example}{Example}
\newtheorem{assumption}{Assumption}
\newtheorem{remark}{Remark}

\Crefname{thm}{Theorem}{Theorems}
\crefname{thm}{Thm.}{Thms.}
\Crefname{prop}{Proposition}{Propositions}
\crefname{prop}{Prop.}{Props.}
\Crefname{defi}{Definition}{Definitions}
\crefname{defi}{Def.}{Defs.}
\Crefname{prop}{Proposition}{Propositions}
\crefname{lemma}{Lemma}{Lemmas}
\crefname{coro}{Cor.}{Cors.}
\Crefname{coro}{Corollary}{Corollaries}
\crefname{exampe}{Ex.}{Exs.}
\Crefname{exampe}{Example}{Examples}
\crefname{remark}{Remark}{Remarks}
\crefname{conj}{Conj.}{Conjs.}
\Crefname{conj}{Conjecture}{Conjectures}
\newcommand{\ms}[1]{\textcolor{black!60!green}{ [#1]}}
\newcommand{\vv}[1]{\textcolor{blue}{ [#1]}}

\hypersetup{
    colorlinks=true,
    linkcolor=blue,
    filecolor=magenta,      
    urlcolor=blue,
    citecolor=green,
    }

\DeclareFieldFormat{title}{\mkbibquote{#1}}

\title{Mapping indefinite causal order processes to composable quantum protocols in a spacetime}

\author{Matthias Salzger}
\affiliation{International Centre for Theory of Quantum Technologies, University of Gda\'{n}sk, 80-309 Gda\'{n}sk, Poland}
\affiliation{Institute for Theoretical Physics, ETH Zurich, 8093 Z\"{u}rich, Switzerland}
\email{matthias.salzger@phdstud.ug.edu.pl}

\author{V. Vilasini}
\affiliation{Institute for Theoretical Physics, ETH Zurich, 8093 Z\"{u}rich, Switzerland}
\email{vilasini@phys.ethz.ch}

\date{April 2024}

\begin{document}

\maketitle

\begin{abstract}


Formalisms for higher order quantum processes provide a theoretical formalisation of quantum processes where the order of agents' operations need not be definite and acyclic, but may be subject to quantum superpositions. This has led to the concept of indefinite causal structures (ICS) which have garnered much interest. However, the interface between these information-theoretic approaches and spatiotemporal notions of causality is less understood, and questions relating to the physical realisability of ICS in a spatiotemporal context persist despite progress in their information-theoretic characterisation.  Further, previous work suggests that composition of processes is not so straightforward in ICS frameworks, which raises the question of how this connects with the observed composability of physical experiments in spacetime. To address these points, we compare the formalism of quantum circuits with quantum control of causal order (QC-QC), which models an interesting class of ICS processes, with that of causal boxes, which models composable quantum information protocols in spacetime. We incorporate the set-up assumptions of the QC-QC framework into the spatiotemporal perspective and show that every QC-QC can be mapped to a causal box that satisfies these set up assumptions and acts on a Fock space while reproducing the QC-QC's behaviour in a relevant subspace defined by the assumptions. Using a recently introduced concept of fine-graining, we show that the causal box corresponds to a fine-graining of the QC-QC, which unravels the original ICS of the QC-QC into a set of quantum operations with a well-defined and acyclic causal order, compatible with the spacetime structure. Our results also clarify how the composability of physical experiments is recovered, while highlighting the essential role of relativistic causality and the Fock space structure.

\end{abstract}

\pagenumbering{arabic}

\tableofcontents

\newpage

\section{Introduction}
\label{sec:intro}

In recent years, significant progress has been made in understanding causal influence in quantum theory, particularly regarding its deviation from classical intuitions due to phenomena like superpositions and entanglement. Various approaches to quantum causality have emerged, such as frameworks for quantum causal models \cite{Henson_2014, pienaar2020quantum, pienaar2015graph, Barrett2019, Barrett_2021, costa2016quantum}, that provide causal explanations for practical quantum experiments such as Bell scenarios, as well as theoretical formalisms for higher order quantum processes which lead to more exotic notions such as indefinite causal structures \cite{Hardy2005,Oreshkov_2012,Chiribella2013}. While the former entails quantum operations of agents occurring in a definite and acyclic order, the latter involves more general, abstract quantum protocols where the order of operations is no longer fixed and definite.\footnote{We do not necessarily endorse the terminology "indefinite" causal structures for such protocols in general but adhere to it in this paper for consistency with the literature.}

These indefinite causal structures have been extensively studied for their intriguing theoretical possibilities which extend beyond the standard quantum circuit paradigm, and the potential applications they may offer for information processing \cite{Chiribella_2012, Colnaghi_2012, Ara_jo_2014}. 
However, the physical realisability of these theoretical processes, as well as what constitutes a ``faithful" realisation thereof, remains a highly discussed open problem in the field \cite{Portmann_2017, Paunkovi__2020, Oreshkov_2019, Wechs_2021, Ormrod_2023}.


The framework of \emph{quantum circuits with quantum control of causal order (QC-QC)} \cite{Wechs_2021} describes a subset of processes which can be interpreted in terms of generalised quantum circuits. These allow for processes which are associated with an indefinite causal structure, and involve a quantum superposition of the order of agents' operations, which can be coherently controlled by the state of some quantum system, or where one agent's quantum/classical output may dynamically determine the order in which future agents act. QC-QCs are widely regarded as being physically realisable (in principle), but it is yet to be formalised in what precise sense. Moreover, a crucial question is whether this class encompasses \emph{all} physically realisable processes. 

A related question concerns the composability of physical processes, due to indications that composing even simple processes in typical frameworks for indefinite causality is not so straightforward \cite{Gu_rin_2019, Jia_2018}. Composability is central to our understanding of physics, as the composition of two physical experiments is another physical experiment. How does this potential difficulty with compositions in the process framework reconcile with the observed composability of real-world experiments? Previous approaches \cite{Jia_2018, Kissinger_2019} have proposed consistent rules for composing (higher order) quantum processes. However, the implications of these abstract rules for the composability of quantum experiments in space and time have not been previously considered.

The aforementioned approaches operate within an information-theoretic understanding of causality, based on the flow of information between systems. This is a priori distinct from relativistic  notions of causality related to spacetime \cite{VilasiniColbeck_2022, VilasiniColbeck_2022_loops, Vilasini_2022}. To formalise physical realisability and address related questions, it is necessary to link such information-theoretic structures to space and time and account for relativistic principles of causality. In \cite{Vilasini_2022}, a top-down framework was proposed for making precise the link between these notions and formalising relativistic principles for general quantum protocols in a \emph{fixed and acyclic background spacetime}. This led to a formalisation of the concept of realisation of a quantum process in such a spacetime, allowing for scenarios where quantum messages exchanged between agents\footnote{The agents and their labs are regarded as classical, although they can perform quantum operations on quantum systems entering their labs. This distinguishes the approach of \cite{Vilasini_2022} from frameworks for quantum reference frames where the reference frame/observer is itself regarded as a quantum system.} can take superpositions of trajectories in the spacetime. There are two important insights derived from this approach which are relevant for the physicality question: the first concerns the concept of \emph{fine-graining} and the second concerns the relation to the previously known framework of \emph{causal boxes} \cite{Portmann_2017}. We describe these two aspects in turn.

The concept of fine-graining of causal structures introduced in \cite{Vilasini_2022} demonstrates that an abstract information-theoretic causal structure, which is not acyclic, can unravel into an acyclic causal structure at a fine-grained level once realised in space and time, without violating relativistic causality principles. A simple and intuitive example of this phenomenon is the following: if the demand $D$ and price $P$ of a commodity causally influence each other, we have a cyclic causal structure between $D$ and $P$, although the fine-grained description of the physical process is an acyclic one where demand $D_1$ at time $t_1$ influences price $P_2$ at time $t_2>t_1$ which influences demand $D_3$ at time $t_3>t_2$ and so on. This enables information-theoretic and spatiotemporal causality notions to be consistently reconciled for quantum experiments in spacetime, both at a formal and a conceptual level (see \cite{Vilasini_2022} and \cref{sec:finegraining} for details).

The causal box framework \cite{Portmann_2017} describes composable information processing protocols within fixed acyclic spacetimes, allowing for scenarios where quantum states may be sent or received at a superposition of different spacetime locations. Originally developed for studying security notions in relativistic quantum cryptography which remain stable under composition of protocols \cite{Vilasini_2019}, the formalism guarantees that composition of two or more causal boxes results in yet another causal box. It has been shown through the top-down approach of \cite{Vilasini_2022} that the most general protocols that can be realised without violating relativistic causality in a background spacetime are those that can be described as causal boxes.\footnote{This statement applies to all protocols involving finite-dimensional quantum systems and a finite number of information processing steps, not to the infinite-dimensional case in its current form.} 

Here we focus on the question of physical realisability of processes in a fixed background spacetime--- arguably the regime in which current-day experiments operate. The above-mentioned results of \cite{Vilasini_2022} enable the question of physical realisability of processes in a fixed spacetime to be reduced to asking which subset of processes can be modelled as causal boxes.\footnote{Explicitly, since causal boxes describe the most general protocols in a fixed spacetime, a necessary condition for a process to be realisable in a fixed spacetime is that it admits a causal box model.} This motivates us to compare and map between the causal boxes and QC-QC frameworks. However, these two frameworks, while sharing some common features, are quite different. For example, causal boxes are composable and allow multiple rounds of information processing and superpositions of different numbers of messages while processes/QC-QCs are not immediately composable as suggested by \cite{Gu_rin_2019}, and only consider agents in closed labs acting once on a single message. Thus, mapping between the formalisms requires a careful analysis of the underlying objects, state spaces, and spacetime information. We undertake such an analysis here, and our results (which we summarise below) formally ground this important class of processes within a spatiotemporal and operational perspective, and shed light on their fine-grained causal structure as well as the composability of physical experiments. 


{\bf Summary of contributions} Starting with a review of higher-order quantum processes (including QC-QCs) in \cref{sec:pm} and the causal box formalism in \cref{sec:cb}, \cref{sec:statespace} delves into a detailed analysis of state spaces and operations within the QC-QC and causal box frameworks. Incorporating the set-up assumptions of the process formalism (such as the fact that each party can act exactly once) in terms of the spacetime picture, we formally address the question: what does it mean for a causal box to model/behave like a QC-QC under these assumptions? We term such a causal box an \textit{extension} of the associated QC-QC. 
In \cref{sec:qcqctopb}, we demonstrate the existence of a causal box extension for each QC-QC and construct multiple extensions. 
In \cref{sec:fine}, we show the causal box extension of a QC-QC is a fine-graining of the QC-QC according to the concept introduced in \cite{Vilasini_2022}. Importantly, although the QC-QC may be associated with an indefinite causal structure, we show that the resulting causal box always exhibits a definite and acyclic fine-grained causal structure. This holds even as the causal box reproduces the QC-QC's action on relevant states, permitting each party to act once on a physical system. In \cref{sec:composability}, we discuss how one can reconcile the suggested difficulties in composition within the process framework \cite{Gu_rin_2019} with the observation that physical experiments in spacetime are composable, emphasizing the interplay of relativistic causality principles, the Fock space structure and the set-up assumptions of the process/QC-QC frameworks in this regard. This work sets the foundation for a follow-up paper, where we explore the reverse mapping, from general causal boxes to QC-QCs, to provide a tighter characterisation of processes realisable in a fixed background spacetime.

\section{Review of higher order quantum processes}\label{sec:pm}
\subsection{General higher order quantum processes}\label{sec:generalpm}

The standard quantum circuit paradigm describes protocols with a well-defined acyclic ordering between the different operations (or quantum gates). This is consistent with a clear arrow of time. Recently, frameworks have been proposed which go beyond this to define abstract information-processing protocols without assuming the existence of a definite acyclic order between the different operations, or the existence of a background spacetime structure \cite{Barrett_2021, Oreshkov_2012, Chiribella2013, Wechs_2021}.


In general, we consider agents which apply quantum instruments $\C{M}_A^a = \{\C{M}_A^{x|a}\}_x$ where $\C{M}_A^{x|a}: \C{L}(\C{H}^{A^I}) \rightarrow \C{L}(\C{H}^{A^O})$ \cite{davies1970operational} where $\C{H}^{A^I}$ is the agent's input Hilbert space while $\C{H}^{A^O}$ is their output Hilbert space, $a$ is the measurement setting and $x$ the measurement outcome. However, for our purposes, it will actually be enough to just consider generic CP maps $\C{M}_A$ as we will be interested in generic outcome probabilities. We will refer to such operations as local operations. 

\begin{figure}
    \centering
    \includegraphics{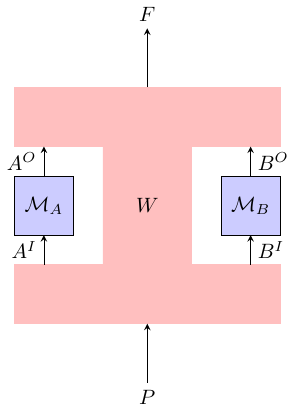}
    \caption{A higher order process (pink) taking two local operations (blue) as input. The result is a map from some global past $P$ to some global future $F$.}
    \label{fig:higherorder}
\end{figure}

We can then view a general quantum process as a higher order map $\C{W}$ which acts on $N$ local operations and maps them to another operation (CP map) $\C{M}: \C{L}(\C{H}^{P}) \rightarrow \C{L}(\C{H}^{F})$, where $\C{H}^{P/F}$ correspond to the Hilbert spaces associated with the global past/future (relative to the remaining operations)

\begin{equation}
  \C{W}:  \{\C{M}_{A_1},...,\C{M}_{A_N}\} \mapsto \C{M}.
\end{equation}

This is depicted diagrammatically in \cref{fig:higherorder} for the case of two agents $A$ and $B$. If $\C{H}^{P/F}$ are trivial (one-dimensional), $\C{W}$ encodes the joint probability associated with the CP maps of the local operations.

The Choi representation of such a higher order process is also known as the process matrix \cite{Oreshkov_2012}, $W\in \bigotimes_{i=1}^N\C{L}(\C{H}^{A_i^IA_i^O})\otimes \C{L}(\C{H}^{PF})$, which lives in the joint space of all the agents' in and outputs and the global past and future. The action of the supermap on the local operations can be written in terms of the process matrix. The result is the Choi representation of the map $\C{C}$ which lives in $\C{L}(\C{H}^{PF})$,


\begin{equation}\label{eq:supermap}
    W(\C{M}_{A_1},...,\C{M}_{A_N}) \coloneqq (M_{A_1} \otimes ... \otimes M_{A_N}) * W \in \C{L}(\C{H}^{PF})
\end{equation}

where $M_{A_i}$ denotes the Choi matrix of the local operation $\C{M}_{A_i}$ and $*$ is the link product. For a review of the link product we refer to \cref{sec:link}. The outcome probabilities for a given state $\rho^P \in \C{H}^P$, and local operations associated with particular outcomes and settings can be obtained via the generalised Born rule,

\begin{equation}\label{eq:pmprob}
 P(x_1,...x_N|a_1,...,a_N)=   \tr_{F} [(\rho^P \otimes M_{A_1}^{x_1|a_1} \otimes ... \otimes M_{A_N}^{x_N|a_N}) * W].
\end{equation}

For the remainder of this work we will work in the process matrix formulation. However, one could equivalently work in the supermap formulation of \cite{Chiribella2013} and as such our results apply to both equally. 

It is often much easier to work with process vectors \cite{Ara_jo_2015} instead of process matrices. While the process matrix can be viewed as the Choi matrix of the environment, the process vector is essentially the corresponding Choi vector. If the process vector is given by $\ket{w}$, then the process matrix is simply $W = \ket{w} \bra{w}$. Instead of \cref{eq:supermap}, we can then use

\begin{equation}\label{eq:pmprob_pure}
    (\dket{A_1} \otimes ... \otimes \dket{A_N}) * \ket{w}.
\end{equation}

where $\C{M}_{A_k}(\rho) = A_k \rho A_k^\dagger$ and $\dket{A_k}$ is the Choi vector of $A_k$ (cf. \cref{sec:choi}). In such cases, we will also refer to $A_k$ as the local operation.\footnote{Note that we use $A_k$ to refer to both the local agent and the single Kraus operator describing the pure operation that this agent applies. This will allow us to keep equations compact while it should be clear from context whether the agent or the Kraus operator is meant.}

\subsection{Subset of processes modelled by generalised quantum circuits}\label{sec:circuitpm}

In general, interpreting the causal structures described by process matrices is difficult. As discussed above, the framework can model very general scenarios as it does not assume a background spacetime and a long-standing open question is to understand which process matrices can be physically realised and under what assumptions and physical regimes. In particular, there are 
so-called non-causal processes that produce correlations which violate bounds known as causal inequalities\footnote{These bounds on correlations are set by the causal processes, i.e. those that are compatible with some definite acyclic causal order or a convex combination of several such orders.}. Such processes can provide an advantage over more conventional processes (e.g., those with a definite, acyclic order) in certain information-theoretic games \cite{Oreshkov_2012}.


The framework of quantum circuits with quantum control of causal order (QC-QC) \cite{Wechs_2021}, adopts a bottom-up approach to this problem and defines a broad class of process matrices that can be represented in terms of generalised quantum circuits. As we will see, they can be viewed as circuits in the sense that the local operations can be ``plugged in" with their causal order being quantum coherently or classically controlled. They are, however, more general than standard quantum circuits as they include processes which are regarded as having an indefinite causal structure. Nevertheless, QC-QCs have been shown to not violate any causal inequalities (analogous to entangled states that do not violate Bell inequalities).

A framework with similar results was developed in \cite{Purves_2021}, we will, however, focus on the QC-QC framework in this work. In this section, we will only briefly summarise the main features of the QC-QC framework in words. This should be enough to understand the core concepts and arguments presented in the main text. A more mathematical summary can be found in \cref{sec:qcqcdetails}, whereas for the full description we refer to the original paper \cite{Wechs_2021}.

\begin{figure}[t!]
    \centering
    \begin{subfigure}{1\textwidth}
    \centering
    \includegraphics[width=1\textwidth]{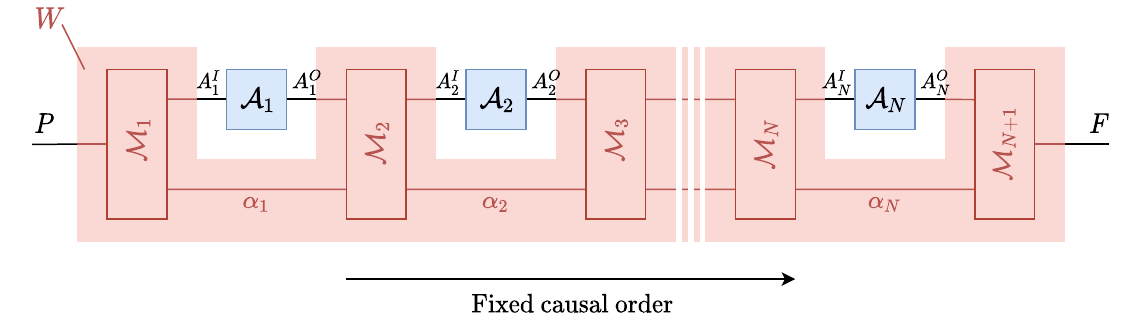}
    \caption{Quantum circuits with fixed order (QC-FO) are the standard quantum circuits where the order of operations is fixed and determined.}
    \label{fig:qcfo}
    \end{subfigure}
    \begin{subfigure}{1\textwidth}
    \centering
    \includegraphics[width=1\textwidth]{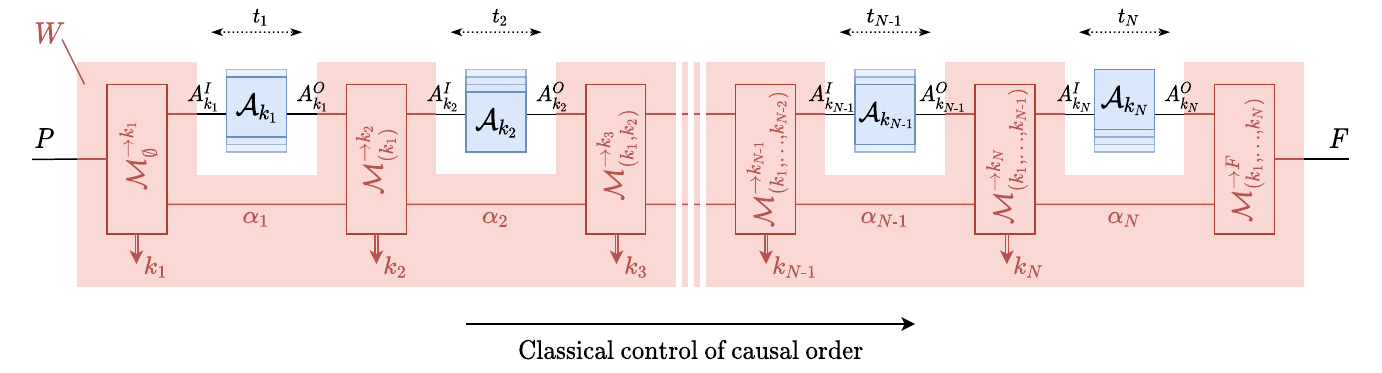}
    \caption{Quantum circuits with classical control of causal order (QC-CC) apply a measurement during each time step (to the target or to some ancilla or to both together). Each possible measurement outcome corresponds to an agent that has not acted so far (therefore, if there are $N$ agents, in the $n$-th time step there will be up to $N-n+1$ possible outcomes).}
    \label{fig:qccc}
    \end{subfigure}
    \begin{subfigure}{1\textwidth}
    \centering
    \includegraphics[width=1\textwidth]{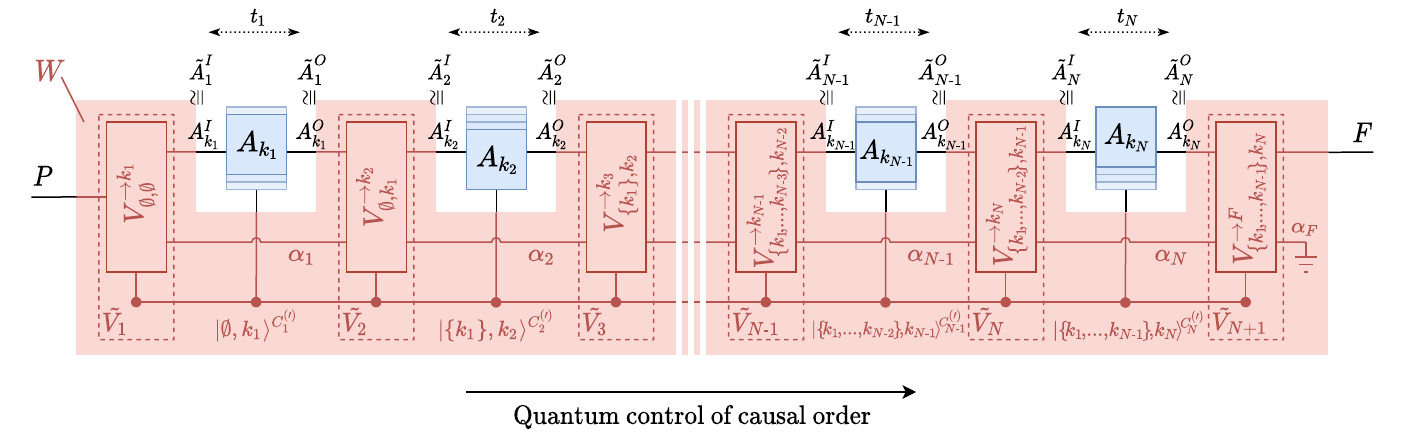}
    \caption{Instead of applying a measurement as is the case for QC-CCs, a quantum circuit with quantum control of causal order (QC-QC) sends the target system in coherent superposition to (a subset of) all agents that have not acted so far. Unlike in the case of QC-CCs, where the causal order is not predetermined but still well defined, the target and the control are thus in a superposition of different ``paths". }
    \label{fig:qcqc}
    \end{subfigure}
    \caption{The three types of classifications for objects in the QC-QC framework. The figures are taken from \cite{Wechs_2019}, where they are listed Fig. 4, Fig. 7 and Fig. 10}
\end{figure}

QC-QCs can be divided into three categories. Quantum circuits with fixed order (QC-FO) are essentially standard quantum circuits with ``slots'' where an external operation can be plugged in and the order in which operations are applied by the circuit is fixed and acyclic. These objects are also referred to as quantum combs \cite{Chiribella_2019}. Quantum circuits with classical control of causal order (QC-CC) allow for classical mixtures, including in a dynamical fashion, of definite acyclic orders. Here, the circuit applies a measurement during each time step (to the target, an ancilla or to both together) and depending  on the classical measurement outcome, sends the target system to an agent who has not yet acted so far. Thus at each time step, each possible measurement outcome corresponds to an agent that has not acted so far (therefore, if there are $N$ agents, in the $n$-th time step there will be up to $N-n+1$ possible outcomes). 
The order is thus ultimately classical, but is not predetermined and may be established dynamically during run time. Note that the applied measurement must in general depend on which agents already acted. This can be achieved by recording the agents that already acted in a control system which is carried along and updated in the circuit. QC-FOs are a subset of QC-CCs as they can be viewed as QC-CCs with a single measurement outcome in each time step.

Finally, there is general QC-QCs of which QC-CCs (and thus also QC-FOs) are a subset. The main difference is that instead of applying a measurement, the circuit sends the target in coherent superposition to the agents. This is controlled via a control system $\ket{\{k_1,...,k_{n}\}, k_{n+1}}$ where $k_1,...,k_n$ refers to the agents that acted previously (without revealing the order in which they acted) and $k_{n+1}$ refers to the agent the circuit sends the target to. We will usually write this in the more compact form $\ket{\C{K}_n, k_n}$ with $\C{K}_n = \{k_1,...,k_n\}$. Unlike in the case of QC-CCs, which have only classical uncertainty in the causal order, QC-QCs enable a quantum uncertainty in the order of the operations. 

We note that which agent receives the target at a given time step can also depend dynamically on the state of the target (in addition to depending on the state of the control system). A simple example from the physical world is given by a circuit consisting of a beam splitter: Alice sends a photon to the circuit, and depending on the polarization of Alice's photon, the beam splitter then reflects it to Bob or transmits it to Charlie. Additionally, the circuit may apply some transformation to the target before sending it to the next agent.

\section{Review of the causal box framework}\label{sec:cb}

The causal box framework \cite{Portmann_2017} models information processing protocols satisfying relativistic causality in a fixed background spacetime (such as Minkowski spacetime), where quantum messages may be exchanged in superpositions of different spacetime locations. 
The framework is, however, very different from the process matrix or QC-QC approaches as it allows for multiple rounds of information processing, is closed under arbitrary composition and does not partition protocols into local operations of agents and processes describing an inaccessible (to the agents) environment. Additionally, it explicitly models sending ``nothing" with a vacuum state $\ket{\Omega}$. The existence of such a state and the possibility of sending superpositions of ``something" and ``nothing", $\alpha \ket{\psi} + \beta \ket{\Omega}$, has physical relevance. For example, the coherently controlled application of an unknown unitary to a target system was shown to be impossible in theoretical formalisms that do not model the vacuum, but has been experimentally realised due to the physical possibility of such ``vacuum'' superpositions \cite{Friis_2014, Zhou_2011}.

\subsection{Messages and Fock spaces}\label{sec:fock}

We now review the formal aspects of the causal box framework. Causal boxes are defined on a background spacetime which is modelled as a partially ordered set $\C{T}$, capturing the light cone structure of the spacetime. However, for our purposes it will suffice to consider finite and totally ordered sets (in which case $\C{T}$ can be interpreted as a set of time stamps). For simplicity, we will therefore restrict to this case in this review as well. We refer to \cite{Portmann_2017} for the general case. A message corresponds to a Hilbert space $\C{H}^A$ together with a time stamp encoded in the sequence space $l^2(\C{T})$ of $\C{T}$ with bounded 2-norm. We can write the state of an arbitrary message as $\ket{\psi}^A\otimes  \ket{t} \in \C{H}^A \otimes l^2(\C{T})$, where the content of the message is encoded in $\ket{\psi}^A$, while $t \in \C{T}$ contains the time information of when the message is sent or received. 
We will frequently write $\ket{\psi, t}^A$ instead of $\ket{\psi}^A \otimes \ket{t}$.

For process matrices and QC-QCs, the state space of a wire is a finite-dimensional Hilbert space. As mentioned earlier, the causal box framework allows the sending of multiple messages or, in other words, there can be any number of messages on a wire. To capture this, the state space of the wire in the causal box framework is modelled as the symmetric Fock space of the single-message space 

\begin{equation}\label{eq:fockdef}
    \C{F}(\C{H}^A \otimes l^2(\C{T})) = \bigoplus_{n=0}^\infty \vee^n (\C{H}^A \otimes l^2(\C{T}))
\end{equation}

where $\vee^n (\C{H}^A \otimes l^2(\C{T}))$ is the symmetric subspace of $(\C{H}^A \otimes l^2(\C{T}))^{\otimes n}$. The one-dimensional space $(\C{H}^{A} \otimes l^2(\C{T}))^{\otimes 0}$ corresponds to the vacuum state $\ket{\Omega}$. For more details on the symmetric tensor product, see \cref{sec:fockinner}. The notation in \cref{eq:fockdef} is quite cumbersome so we will often abbreviate it by writing $\C{F}^{\C{T}}_{A} := \C{F}(\C{H}^A \otimes l^2(\C{T}))$.



\paragraph{Wire isomorphisms:} There exist two useful isomorphisms regarding the splitting of wires. For any $\C{H}^A, \C{H}^B$ and any $\C{T}' \subseteq \C{T}$, we have

\begin{gather}\label{eq:wireiso}
\begin{aligned}
    \C{F}^{\C{T}}_{C} &\cong \C{F}^{\C{T}}_{A} \otimes \C{F}^{\C{T}}_{B} \\
    \C{F}^{\C{T}}_{A} &\cong \C{F}^{\C{T}'}_{A} \otimes \C{F}^{\C{T} \backslash \C{T}'}_{A}
\end{aligned}
\end{gather}

where $\C{H}^C = \C{H}^A \oplus \C{H}^B$.

The first isomorphism implies that two wires, one carrying messages from one Hilbert space $\C{H}^B$ and the other carrying messages from another Hilbert space $\C{H}^C$, are equivalent to a single wire carrying messages from the direct sum of the two Hilbert spaces. This will allow us to formally define causal boxes with just a single input and a single output wire. 

The second isomorphism applied recursively implies that there is an equivalence between one wire carrying the messages from all times $t\in \C{T}$ and having a separate wire for each $t \in \C{T}$. Throughout this paper, we will treat messages with different time stamps are being associated with different wires, such that states of multiple messages associated with distinct time stamps can be written using the regular tensor product (rather than a symmetrised product). Furthermore, for any $\C{T}'\subseteq \C{T}$, there is a natural embedding of $\C{F}^{\C{T}'}_{A}$ in $\C{F}^{\C{T}}_{A}$ given by appending the vacuum state on $\C{T} \backslash \C{T}'$.

\begin{equation}
\label{eq:wireisostate}
   \C{F}^{\C{T}'}_{A} \cong \C{F}^{\C{T}}_{A}\otimes \ket{\Omega, \C{T} \backslash \C{T}'}. 
\end{equation}



\subsection{Definition of causal boxes}\label{sec:defcb}

We give now a simplified definition of causal boxes that captures those causal boxes defined on a finite and totally ordered $\C{T}$. 

\begin{defi}[Causal boxes \cite{Portmann_2017}]\label{def:causalbox}
A causal box is a system with an input wire $X$ and an output wire $Y$, together with a CPTP map

\begin{equation}
    \C{C}: \mathcal{L}(\C{F}^{\C{T}}_X) \rightarrow \mathcal{L}(\C{F}^{\C{T}}_Y)
\end{equation}

that fulfills for all $t \in \C{T}$

\begin{equation}\label{eq:cbreq}
    \tr_{>t} \circ \C{C} = \tr_{>t} \circ \C{C} \circ \tr_{>t} 
\end{equation}

where $\tr_{>t}$ corresponds to tracing out all messages with time stamps larger than $t$.

\end{defi}


\cref{eq:cbreq} encodes the causality condition, i.e. that we can calculate the outputs up to time $t$ from the inputs up to $t$ which is equivalent to saying that inputs after $t$ cannot influence outputs up to $t$. Note that from this definition it directly follows that a causal box restricted to some subset of timestamps less than $t$ is once again a causal box.

\subsection{Representations of causal boxes}\label{sec:repcb}

Causal boxes admit two alternative representations \cite{Portmann_2017}. One of these is the Choi representation. For infinite-dimensional Hilbert spaces, the usual Choi operator can be unbounded and an alternate definition has to be used in those cases \cite{Holevo_2011}. For the definition in this general case, see \cite{Portmann_2017}. However, we note that the space $\bigoplus_{n=0}^N \vee^n (\C{H}^A \otimes l^2 (\C{T}))$ of up to $N$ messages is finite-dimensional for all $N$. By restricting the causal box to this domain, we can thus use the normal Choi operator. This restriction will not lead to any loss of generality in our results because 
defining the causal box for every $N$ uniquely determines the causal box on the full Fock space where we have a direct sum from $N=0$ to $N=\infty$, as explained in \cref{remark:infinite}.

The other representation is the sequence representation, which is essentially a consequence of the Stinespring dilation \cite{Stinespring_1955} of a CPTP map. As causal boxes are CPTP maps, they also admit Stinespring dilations \cite{Portmann_2017}. One can use this fact to decompose a causal box into a sequence of isometries, each of which describes the behaviour of the causal box during some disjoint subset of $\C{T}$ (see \cref{fig:sequencerep}). This is called a sequence representation of the causal box. This representation can thus be viewed as a causal unraveling of the causal box. On the one hand, this is a useful way to visualise the behaviour of a causal box and on the other hand, we can also use this to construct causal boxes. A sequence of isometries $V_n$ with appropriate input and output spaces will always yield a causal box.

\begin{figure}
    \centering
    \includegraphics{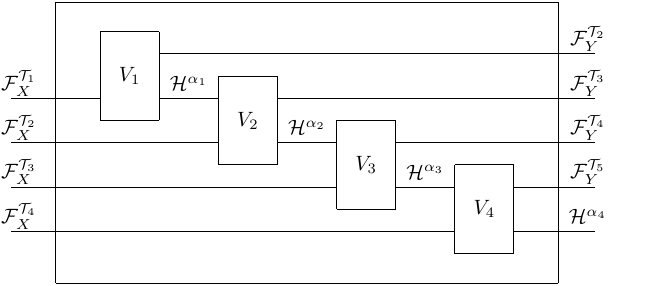}
    \caption{Graphical representation of the sequence representation of a causal box. The causal box is decomposed into a sequence of isometries $V_n$ (in this case there are four such isometries). Each isometry acts on messages during a specific time slice (named $\C{T}_n$ for $n=1,2,3,4,5$ in the figure) and outputs messages during the next time slice. A causal box can thus be understood via its action during each time slice.}
    \label{fig:sequencerep}
\end{figure}







\begin{remark}\label{remark:infinite}
Note that defining a linear map on each $n$-message subspace does not a priori mean it is well-defined on the full Fock space. This is because an infinite-dimensional Hilbert space is not simply the (finite) linear span of some orthonormal basis ${\ket{j}}_{j \in \mathbb{N}}$ but the metric completion of such a span (which contains also any elements corresponding to the limits of Cauchy sequences). A counter-example is the linear map $V\ket{j} = j\ket{j}$. We then have, for example, $V \sum_{j \in \mathbb{N}} 1/j \ket{j}= \sum_{j \in \mathbb{N}} \ket{j}$ which is divergent in its norm. We additionally need that the map is continuous. That this then suffices follows straightforwardly from the fact that the action of a continuous function $f$ on a limit $a$, $f(a)$, must be the limit of $f(a_n)$ for a sequence $\{a_n\}_{n\in\mathbb{N}}$ that converges towards $a$. In the case of causal boxes, we can always consider a purification, which is an isometry and thus continuous. Therefore, it suffices to work with maps on the $n$-message subspaces. Since Choi vectors and matrices uniquely determine the corresponding maps (once we fix a basis) the same argument justifies working with the finite-dimensional $n$-message Choi representations as well.
\end{remark}

\subsection{Composition of causal boxes}\label{sec:composition}

Causal boxes can be composed in the following ways: \emph{parallel composition} of two causal boxes is simply given by their tensor product, \emph{sequential composition} involves connecting an output of a causal box at some time to the input of a causal box at a later time, and \emph{loop composition} is an operation on a single causal box which feeds back an output of the box to an input of the same box associated with a later time. Sequential composition of two causal boxes can be expressed in terms of the other two by first taking the parallel composition of the two boxes and then performing an appropriate loop composition \cite{Portmann_2017, Vilasini_2022}. It can be shown that the set of valid causal boxes are closed under arbitrary compositions of these types. 



We give the definition of loop composition below. As for the Choi representation, we will once again only consider the finite-dimensional case. This is again sufficient to determine the behaviour on the full Fock space because of the reasoning given in \cref{remark:infinite}.

\begin{defi}[Loop composition \cite{Portmann_2017, }]\label{def:comp}
Consider a CP map $\C{M}: \mathcal{L}(\mathcal{H}^{AB})\rightarrow\mathcal{L}(\mathcal{H}^{CD})$ with input systems $A$ and $B$ and output systems $C$ and $D$ with $\C{H}^B \cong \C{H}^C$. Let $\{\ket{k}^C\}_k$ be any orthonormal basis of $\mathcal{H}^C$, and denote with $\{\ket{k}^B\}_k$ the corresponding basis of $\mathcal{H}^B$ i.e. for all $k$, $\ket{k}^C\cong\ket{k}^B$. The new system resulting from looping the output system $C$ to the input system $B$, $\C{M}^{C\hookrightarrow B}$ is given as 
\begin{equation}
    \label{eq:loopfinite}
    \C{M}^{C\hookrightarrow B} (\ket{\psi}^A \bra{\phi}^A) = \sum_{k,l} \bra{k}^C\C{M}(\ket{\psi}^A\ket{k}^B\bra{l}^B\bra{\phi}^A)\ket{l}^C.
\end{equation}
\end{defi}

The sequential composition of two CP maps $\C{M}_A: \C{L}(\C{H}^A) \rightarrow \C{L}(\C{H}^C)$ and $\C{M}_B: \C{L}(\C{H}^B) \rightarrow \C{L}(\C{H}^D)$ is then given by $(\C{M}_A \otimes \C{M}_B)^{C\hookrightarrow B}$.

\begin{remark}[Basis dependence of the loop composition] Note that the composition is basis-dependent in the same sense that the Choi isomorphism and the link product are basis-dependent. The bases we use for composition should thus be the same bases that we use to calculate Choi matrices and link products to obtain consistent results.
\end{remark}



\section{Extending QC-QCs to causal boxes}\label{sec:qcqctopb}

\subsection{Overview of the results}
\label{sec:resultsoverview}

Here, we outline the ingredients behind the main theorems of this paper, the statement of the results along with their relevant implications for causality and composability. As highlighted in the introduction, characterising the relation between QC-QCs and causal boxes is important for understanding the physicality and composability of abstract higher order quantum processes, in the context of their realisations in a background spacetime. The frameworks are rather different a priori, and it is necessary to identify the subset of protocols described by causal boxes that satisfy the set-up assumptions of the process matrix/QC-QC frameworks \cite{Vilasini_2020}. We describe these assumptions below.

\begin{assumption}[Acting once and only once]
\label{assumption: assump1}
Whenever a party in a protocol described by a QC-QC acts on a $d$-dimensional system, then in any causal box description of the protocol, the party must act on exactly one non-vacuum message, and this non-vacuum message must also be of $d$-dimensions.
\end{assumption}

\begin{assumption}[Order of local in/output events]
\label{assumption: assump2}
  Whenever a party in a protocol described by a QC-QC has non-trivial in and output spaces, then in any causal box description of the protocol, that party must receive a non-vacuum input to their lab before they send out any non-vacuum output. 
\end{assumption}

We refer to the conjunction of these two assumptions as the \emph{spatiotemporal closed labs assumption}. These assumptions and consequently the following theorem, which is a main result of this paper, will be mathematically formalised in \cref{sec:statespace}. The proof of the theorem is a consequence of \cref{prop:anotherequivalence} in \cref{sec:another}.
\begin{restatable}{theorem}{mainth}\label{mainth}
   Every protocol described by a QC-QC can be mapped to a protocol described by valid causal boxes in Minkowski spacetime, which reproduce the action of the QC-QC on a well-defined subspace and also respect the spatiotemporal closed labs assumption. 
\end{restatable}

As a consequence of this theorem, the subset of causal boxes in the image of our mapping have certain natural and operationally motivated properties which we consider necessary (making no claims about sufficiency)\footnote{Satisfying these assumptions, while necessary for faithfully realising an indefinite causal order process in spacetime, is not sufficient for regarding the realisation as a genuinely indefinite causal structure. Even when these properties are satisfied, such as in the experimental realisations of the quantum switch, there has been a long-standing debate about whether these implement or simulate indefinite causality. In fact, the results of \cite{Vilasini_2022} and \cref{sec:finegraining} show that such spacetime realisations will always, as a consequence of relativistic causality, admit an explanation in terms of a definite acyclic causal structure, lending support to the side of the debate that they are simulations of ICO. These conditions ensure that regardless of whether these realisations are regarded as simulations or implementations, they are faithful to the set up assumptions respected by the original abstract process.} for regarding the causal box as a faithful spatiotemporal realisation of the QC-QC. Moreover, the causal boxes in the image of our mapping also encode the spatiotemporal degrees of freedom in a minimal manner, and satisfy the following property that simplify their mathematical representation. We map $N$-partite QC-QCs to causal boxes where each party can act at $N$ distinct pairs of in/output times $\{t^I_1,t^O_1,...,t^I_N,t^O_N\}$ such that every non-vacuum input at $t_i^I$ yields a non-vacuum output at the corresponding output time $t_i^O>t_i^I$.

\paragraph{Fine-grained causal structure of QC-QCs}

When considering the realisation of indefinite causal order processes in a fixed spacetime, one is faced with a fundamental question: how can an indefinite information-theoretic causal structure be consistent with a definite spacetime causal structure? This question has been resolved in \cite{Vilasini_2022} where it is shown that a realisation of any (possibly indefinite causal order) process satisfying relativistic causality in a fixed spacetime will ultimately admit a fine-grained description in terms of a fixed and acyclic causal order process which is compatible with the light cone structure of the spacetime. Intuitively, the picture painted by this result is similar to the use of cyclic information-theoretic causal structures in classical statistics to describe physical situations with feedback (e.g., the demand and price of a commodity influence each other), but one is aware that this is a coarse-grained description of an acyclic fine-grained causal structure (where demand at a given time influences price at a later time and vice-versa). In this previous work, the set of processes that can indeed be realised in a spacetime in this manner, was not characterised. Here we link the causal boxes in the image of the mapping provided in \cref{mainth} to the definition of spacetime realisation of a process proposed in \cite{Vilasini_2022}, which allows us to regard the causal box description as a fine-graining of the QC-QC description, and consequently as a fixed spacetime realisation of the QC-QC in the formal sense defined in \cite{Vilasini_2022}.

\begin{restatable}{theorem}{theoremfinegraining}\label{theoremfinegraining}
 Every QC-QC can be mapped to a causal box which 
 reproduces the action of the QC-QC on a subspace, satisfies the spatiotemporal closed labs assumption, is a fine-graining of the QC-QC, where the fine-grained causal order is definite and acyclic and consistent with relativistic causality principles in spacetime.
\end{restatable}

We give an outline of the proof of this theorem in \cref{sec:fine}, while the detailed proof can be found in the appendix, \cref{sec:fineproof}. 

Our results imply that the set of processes realisable in a spacetime without violating relativistic causality, must be at least as large as QC-QCs. This also establishes that for each QC-QC, the indefinite causal structure (if present) can be explained in terms of a fine-grained acyclic causal structure (that of the causal box in the image of our mapping). Showing that a causal box is a fine-graining of a QC-QC also provides an explicit mapping back from the causal box (fine-grained picture) to the QC-QC (coarse-grained picture) through an encoding-decoding scheme, which is analogous to the mapping from physical layer operations and corresponding logical layer description of a quantum computation. 

\paragraph{Insights on composability of QC-QCs}  The notion of composition as well as the closedness of the set of physical operations under composition is absolutely central to physical protocols and experiments (and is respected by causal boxes). However, there are no-go theorems indicating apparent limitations for the composability of process matrices, even those which have a fixed causal order (i.e which are QC-FOs). Our mapping from QC-QCs to causal boxes sheds new light on how process matrices, despite this apparent composability issue, can consistently recover our standard intuitions about composability when physically realised in a spacetime. This connects to a previous resolution of the issue proposed within a purely informational and abstract approach \cite{Jia_2018, Kissinger_2019}, while offering new dimensions to the solution coming from relativistic causality, the use of Fock spaces, as well as the interplay between the sets of causal boxes that do or do not respect the set up assumptions of the process matrix and QC-QC frameworks. We discuss this in \cref{sec:composability}.

\subsection{When does a causal box model a QC-QC?}\label{sec:statespace}

As a first step towards our goal of mapping QC-QCs to causal boxes, we need to ask when can we regard a causal box as reproducing the behaviour of a QC-QC. As we have already pointed out, the two frameworks differ in a number of points. In particular, causal boxes account for spacetime information and use Fock spaces as their state spaces, neither of which are explicitly modelled in the QC-QC framework. However, when we map a QC-QC to a causal box we want the two to behave in ``the same way" at least ``in situations of interest". The aim of this section is to formalise this statement. More specifically, we want a mapping from local operations in the QC-QC framework to the causal box framework and of the QC-QC supermap to the causal box supermap such that when composing the local operations with the supermap in the respective frameworks yields the same transformation.


In the following, let us assume that $\C{T} = \{1,2,...,2N+2\}$ for some $N \in \mathbb{N}$ as this will simplify our discussion. Additionally, let us split $\C{T}$ into the set of input times $\C{T}^I = \{2, 4, ..., 2N+2\}$ and the set of output times $\C{T}^O = \{1, 3, ..., 2N+1\}$\footnote{These assumptions and others we will make in this section may seem too restrictive at a glance. One may think that they come with some loss of generality. However, in the end, for each QC-QC, we only wish to find one corresponding causal box. So, in principle, we can make arbitrary restrictions to the kind of causal boxes we consider as long as we can still achieve this goal. However, naturally, these restrictions make it impossible for us to find \textit{all} causal box extensions of a given QC-QC. One could also say that the QC-QC describes a general and abstract information-theoretic protocol whereas the causal box describes a specific implementation in spacetime of an experiment, of which there may be many.}.

We now model \cref{assumption: assump1}. The QC-QC picture involves agents $A_k$ with input and output spaces $\C{H}^{A^I_k}$ and $\C{H}^{A^O_k}$. It models superpositions of orders by identifying these spaces with the generic spaces $\C{H}^{\tilde{A}^I_n}$ and $\C{H}^{\tilde{A}^O_n}$ which are labelled by the index $n$. The label $n$ of the generic space assigned to an agent $A_k$ specifies the order in which the agent acts, and is determined (possibly coherently) by a control system (cf. \cref{fig:qcqc}). Thus, it is natural to map these generic spaces to the state space in the causal box picture at a specific time, i.e. $n \leftrightarrow t=2n$ or $n \leftrightarrow t=2n+1$ depending on whether we are looking at outputs or inputs.

\begin{figure}[t!]
    \centering
    \begin{subfigure}{1\textwidth}
    \centering
    \includegraphics[]{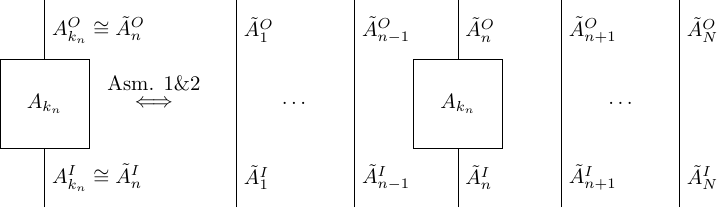}
    \caption{In the QC-QC framework each agent is initially associated with a single input and output wire in between which they apply their operation (left). The output spaces are then identified with generic spaces, which specify the order in which the agent acts, and we can think of there being one generic space for each time step, yielding the picture to the right. Superpositions of orders are modelled by assigning generic spaces to agents in a coherently controlled manner. In a given branch of the superposition, the agent acts during one time step (specified by the label $n$ of the generic space assigned to them) and is inactive during all other times.}
    \label{fig:qcqcstatespaces}
    \end{subfigure}
    \begin{subfigure}{1\textwidth}
    \centering
    \includegraphics[]{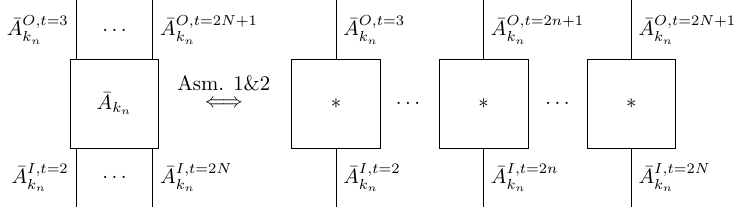}
    \caption{In the causal box framework the agent can be viewed as having an input and output wire for each time step (left) and acting during all time steps which yields the picture to the right. Superpositions of orders are achieved through superpositions of vacuum and non-vacuum states, in each branch of the superposition, the agent acts on a non-vacuum state at one time and the vacuum state at all other times (as ensured by Assumptions~\ref{assumption: assump1} and \ref{assumption: assump2}). On a non-vacuum input, the original QC-QC operation $A_{k_n}$ is applied while a vacuum state input is mapped to a vacuum output at the next time step. Thus, under these assumptions, at each time step, the agent effectively applies the same operation $*$, which stands for $A_{k_n}\otimes \ket{t+1}\bra{t} + \proj{\Omega,t+1}{\Omega, t}$. This is precisely a single tensor factor of \cref{eq:local_pb}, with $A_{k_n}\otimes \ket{t+1}\bra{t}$ taking the place of $\bar{A}^{x,t}$ for all times $t$, capturing the time-independence of the local operations as in the QC-QC picture.}
    \label{fig:pbstatespaces}
    \end{subfigure}
    \caption{Local operations in the QC-QC vs causal boxes pictures, and different ways to capture superpositions of orders. }

\end{figure}

\begin{figure}[t!]
    \centering
    \begin{subfigure}{1\textwidth}
    \centering
    \includegraphics[]{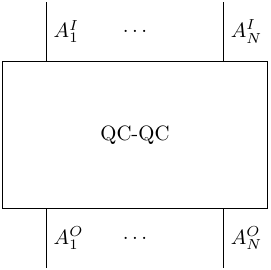}
    \caption{The QC-QC is a supermap acting on the channels corresponding to the local operations of the agents. It can be equivalently represented as a channel in itself, a simplified picture of the statespace of a QC-QC is shown here (in its channel representation) where it is a map from the agents' outputs to their inputs. There is a single input/output wire for each agent.}
    \label{fig:qcqcspace}
    \end{subfigure}
    \begin{subfigure}{1\textwidth}
    \centering
    \includegraphics[]{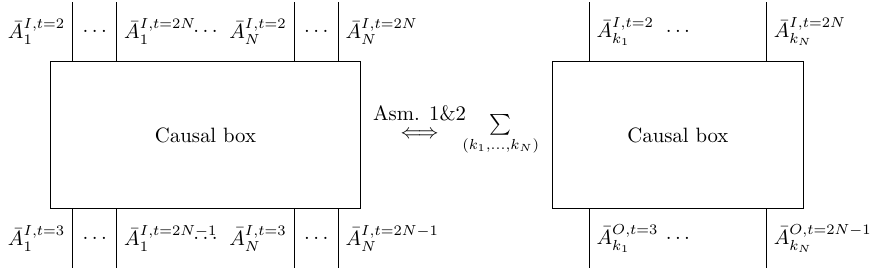}
    \caption{In the causal box framework, the causal box modelling the supermap can be seen as having wires not just for each agent but also for each time stamp. Once we impose that agents act once and only once (\cref{assumption: assump1}) and that agents only send a non-vacuum message after receiving a non-vacuum message (\cref{assumption: assump2}), we obtain the picture on the right, where we can view the causal box as having only a single wire for each agent, which is the wire along which they receive a non-vacuum state in a given ``branch'' of the superposition (hence the sum on the right).}
    \label{fig:cbspace}
    \end{subfigure}
    \caption{The supermap (represented as a channel) in the QC-QC and causal boxes pictures, which would act on the local operations of 
\cref{fig:qcqcstatespaces,fig:pbstatespaces} respectively through composition of systems with the same labels.}
\end{figure}

\Cref{fig:qcqcstatespaces,fig:pbstatespaces} illustrate the correspondence. The input space $\C{H}^{A^I_k}$ is isomorphic to all $N$ generic input spaces $\C{H}^{\tilde{A}^I_n}$, but in a given branch of the temporal superposition, the QC-QC framework only identifies it with one of them. This is because of the assumption that each agent only acts once during a run of the experiment modelled by the QC-QC. The same holds for the output spaces. On the other hand, we model the agent $\bar{A}_k$ (where from here on out we will put a bar over objects in the causal box framework to distinguish them from the corresponding objects in the QC-QC framework) in the causal box picture as another causal box with an input wire for each $t \in \C{T}^I$ and an output wire for each $t \in \C{T}^O$. \cref{assumption: assump1} then corresponds to imposing that in a given branch of the temporal superposition, only one of these wires will contain a non-vacuum state of the dimension given by the corresponding in/output system of the QC-QC\footnote{To be clear, this assumption is not satisfied when composing the agent with arbitrary causal boxes. This is thus an imposition on the whole set-up, and not just the agent alone. For example, consider the causal box with a single output wire and trivial input wire that outputs a qubit at $t=1$ and another one at $t=2$. If we compose this causal box with an agent, the agent receives multiple messages.}

\begin{gather}
\begin{aligned}
    \C{H}^{\tilde{A}^I_n} &\leftrightarrow \C{H}^{\bar{A}^I_k} \otimes \ket{t=2n} \cong \C{H}^{\bar{A}^{I, t=2n}_k}\\
    \C{H}^{\tilde{A}^O_n} &\leftrightarrow \C{H}^{\bar{A}^O_k} \otimes \ket{t=2n+1} \cong \C{H}^{\bar{A}^{O, t=2n+1}_k},
\end{aligned}
\end{gather}
where $\C{H}^{\tilde{A}^I_n}$ and $\C{H}^{\bar{A}^I_k}$ are isomorphic and similar for the output spaces, and there is an equivalence between writing timestamps in kets (middle) and writing them as system labels (right).
On the level of states, we can then make the following identification

\begin{gather}
\begin{aligned}
    \ket{\psi}^{\tilde{A}^I_n} \otimes \ket{\C{K}_{n-1}, k_n} &\leftrightarrow \ket{\psi, t=2n}^{\bar{A}^I_{k_n}} \ket{\Omega, \C{T}^I \backslash 2n}^{\bar{A}^I_{k_n}} \ket{\C{K}_{n-1}, k_n}^{C_n}  \\
    \ket{\psi}^{\tilde{A}^O_{n}} \otimes \ket{\C{K}_{n-1}, k_n} &\leftrightarrow \ket{\psi, t=2n+1}^{\bar{A}^O_{k_n}} \ket{\Omega, \C{T}^O \backslash 2n+1}^{\bar{A}^O_{k_n}} \ket{\C{K}_{n-1}, k_n}^{C_n}
\end{aligned}
\end{gather}
where we now also added the control system which allows us to make the previous correspondence one-to-one. 

Given these correspondences, let us now consider how the agents apply their local operations in each framework. In terms of these time stamps, in the QC-QC framework, the agent $A_{k_n}$ applies their local operation on an input system only at the time step $t=2n$ (and outputting the result at $t=2n+1$), associated with the generic space labelled by $n$, which can be determined coherently by a control. The agent remains fully inactive during all other time steps, as depicted in \cref{fig:qcqcstatespaces}. In the causal box picture on the other hand, each agent's operation acts between every pair of input/output time stamps but it may act on a vacuum or non-vacuum state depending coherently on a control system. This gives us the picture in \cref{fig:pbstatespaces} for causal boxes.

\cref{assumption: assump2} is then modelled by imposing that the agents' operations leave vacuum states invariant. To formalise this, we consider the Kraus operators associated with the local operations. Let $A_i^x: \C{H}^{A^I_i} \rightarrow \C{H}^{A^O_i}$ denote the Kraus operator associated with an outcome $x$ obtained by the $i$-th agent in the QC-QC picture, where we take $A^I_i$ and $A^O_i$ to be $d$-dimensional spaces.\footnote{Here, $d$ is arbitrary and if the two spaces are not of the same dimension the smaller one can be trivially enlarged so that the dimensions match, as is also done in the QC-QC framework \cite{Wechs_2021}.} As noted above, in the causal box picture the corresponding operation acts at each time step, and we consider the operation of the agent associated with obtaining the outcome $x$ when acting on a $d$-dimensional non-vacuum input at time $t$, $\bar{A}_i^{x,t}: \C{H}^{\bar{A}^I_i}\otimes \ket{t} \rightarrow \C{H}^{\bar{A}^O_i}\otimes \ket{t+1}$. As the local operation $A_i^x$ in the QC-QC picture is defined independently of the order in which it is applied by the QC-QC, it is natural to replicate this in the causal box picture by requiring $\bar{A}_i^{x,t}:=A_i^x \otimes \proj{t+1}{t}$ for all $t\in \C{T}^I$. 

For vacuum states, we then need to impose that $\ket{\Omega,t}$ maps to $\ket{\Omega,t+1}$ under the local operation in the causal box picture. This ensures that a non-vacuum output at time $t+1$ must be preceded by a non-vacuum input at $t$, thereby imposing the local order condition of \cref{assumption: assump2}. If the agent obtains the outcome $x$ by measuring a non-vacuum message at time $t$, then \cref{assumption: assump1} and \cref{assumption: assump2} impose that at all other times, they receive and send vacuum states. Hence, we have the map

\begin{equation}\label{eq:kraustx}
    \bar{A}^{x,t}_i \otimes \proj{\Omega, (t+1)^c}{\Omega, t^c} 
\end{equation}
where $\proj{\Omega, (t+1)^c}{\Omega, t^c} = \bigotimes_{t' \neq t} \proj{\Omega, t'+1}{\Omega, t'}$, and $\bar{A}^{x,t}_i \ket{\Omega, t} = 0$. 

Generally, the agent may receive a non-vacuum state at any time $t\in \C{T}_i^I$ (possibly in a superposition of different arrival times). Thus, the operation associated with obtaining $x$ (at any $t$) is the sum over all $t \in \C{T}^I$ of \cref{eq:kraustx},

\begin{equation}\label{eq:notime}
   \bar{A}^{x}_i\coloneqq  \sum_{t \in \C{T}^I_i} \bar{A}^{x,t}_i \otimes \proj{\Omega, (t+1)^c}{\Omega, t^c}.
\end{equation}

This is a complete set of Kraus operators on the one-message subspace whenever the operators $\bar{A}^{x,t}_i$ form a complete set for each time $t$, and this is the case for us as we construct these operators from the QC-QC operators $A^x_i$ for each $t$. We can explicitly check that $\sum_x  (\bar{A}^{x}_i)^{\dagger}  \bar{A}^{x}_i$ is equal to the identity on the one message space,

\begin{equation}
\sum_x    \Bigg(\sum_{t \in \C{T}^I_i} (\bar{A}^{x,t}_i)^\dagger \otimes \proj{\Omega, (t+1)^c}{\Omega, t^c}\Bigg)\Bigg(\sum_{t' \in \C{T}^I_i} \bar{A}^{x,t'}_i \otimes \proj{\Omega, (t'+1)^c}{\Omega, t'^c}\Bigg)
\end{equation}
Then denoting the identity operator on the $m$-message space at time $t$, up to shifting $t$ to $t+1$, as $\mathbb{1}_m^t$, we see that $\proj{\Omega, (t+1)}{\Omega, t}=\mathbb{1}_0^t$. We will use the natural short form $\proj{\Omega, (t+1)^c}{\Omega, t^c}=\mathbb{1}_0^{t^c}=\bigotimes_{t'\neq t}\mathbb{1}_0^{t}$. Then the above simplifies as follows, when we also note that $\bar{A}^{x,t}\ket{\Omega,t}=0$ and $\inprod{t}{t'}=\delta_{t,t'}$,

\begin{equation}
\sum_{t \in \C{T}^I_i}   \sum_x  (\bar{A}^{x,t}_i)^\dagger \bar{A}^{x,t}_i \otimes \mathbb{1}_0^{t^c}=\sum_{t \in \C{T}^I_i}  \mathbb{1}_1^t \otimes \mathbb{1}_0^{t^c},
\end{equation}
which is precisely the identity operator on the overall one-message space (over all times). Here we used the normalisation of the one-message Kraus operators $\bar{A}^{x,t}$ at each time.

Further, \cref{eq:notime} preserves the coherence of temporal superpositions, in particular, if we have $A^x_i\ket{x}=\ket{x}$ for the original QC-QC operators, then an eigenstate $\alpha\ket{x,t_1}+\beta\ket{x,t_2}$ of an outcome $x$ arriving at a superposition of distinct times $t_1$ and $t_2$ (for any non-zero amplitudes $\alpha$ and $\beta$) will remain unaltered by Kraus operators constructed above for the causal box picture. 

We can also obtain an alternative representation of the local map in the causal box picture, where we do not start by assuming \cref{assumption: assump1}, but when we do, we recover the above operators. The map $\bar{A}_i^{x,t}$ acts on $d$-dimensional non-vacuum messages arriving at time $t$. We can extend this to a map on the $d+1$-dimensional space, with trivial action on vacuum (as required by \cref{assumption: assump2}), by considering $\bar{A}_i^{x,t}+\ket{\Omega,t+1}\bra{\Omega,t}$. As this map acts at each time step, we have the following overall map in the causal box picture,

\begin{equation}\label{eq:local_pb}
    \bigotimes_{t \in \C{T}^I_i} (\bar{A}^{x,t}_i + \proj{\Omega,t+1}{\Omega, t}).
\end{equation}

We can see that the maps defined in \cref{eq:notime} and \cref{eq:local_pb} are equivalent on the one-message state space of each agent (i.e. when imposing \cref{assumption: assump1}) by direct calculation,

\begin{gather}
\begin{aligned}
    \Bigg(\sum_{t \in \C{T}^I_i}& \bar{A}^{x,t}_i \otimes \proj{\Omega, (t+1)^c}{\Omega, t^c}\Bigg) \ket{\psi, t_i}^{\bar{A}^I_i} \ket{\Omega, t_i^c}^{\bar{A}^I_i}  \\
    =& \bar{A}^{x,t_i}_i \ket{\psi, t_i}^{A^I_i} \otimes \proj{\Omega, (t_i+1)^c}{\Omega, t_i^c} \ket{\Omega, t_i^c}^{\bar{A}^I_i} \\
    =& (\bar{A}^{x,t_i}_i + \proj{\Omega, t_i+1}{\Omega, t_i}) \ket{\psi, t_i}^{\bar{A}^I_i} (\bar{A}^{t_i^c, x}_i + \proj{\Omega, (t_i+1)^c}{\Omega, t_i^c}) \ket{\Omega, t_i^c}^{\bar{A}^I_i} \\
    =& \bigotimes_{t\in \C{T}^I_i}  (\bar{A}^{x,t}_i + \proj{\Omega, t+1}{\Omega, t}) \ket{\psi, t_i}^{\bar{A}^I_i} \ket{\Omega, t_i^c}^{\bar{A}^I_i}
\end{aligned}
\end{gather}
where we used that $\braket{\psi|\Omega} = 0$ and that $\bar{A}^{x,t}_i \ket{\Omega, t} = 0$ for all $t \in \C{T}^I_i$. By linearity of the two maps, it is clear that this equivalence extends to superpositions and mixtures of states living in the one-message space.

Finally, we note that it suffices to consider agents with a single non-vacuum outcome and thus a single Kraus operator $A_i$ (or $\bar{A}_i$ depending on the framework). Any statement that holds for such agents also holds for the general case via linearity. Based on the above, we define an extension of a local operation in the QC-QC framework to one in the causal box framework as follows.

\begin{defi}[Extension of local operations]\label{def:local_equivalence}
Let $A_i$ be a local agent in the QC-QC framework who applies a local operation $A_i: \C{H}^{A^I_i} \rightarrow \C{H}^{A^O_i}$. We define the extension of this local agent $\bar{A}_i$ in the causal box framework as a causal box with input/output Fock space $\C{F}(\C{H}^{\bar{A}_i^{I/O}} \otimes l^2(\C{T}_i^{I/O}))$. We then call the local operation $\bar{A}_i$ in the causal box framework the extension of the local operation $A_i$ in the QC-QC framework if $\bar{A}_i = \bigotimes_{t \in \C{T}^I_i} (A_i \otimes \proj{t+1}{t} + \proj{\Omega, t+1}{\Omega, t})$ on zero and one message states where the action of $A_i$ on $\C{H}^{\bar{A}^I_i}$ is defined via the trivial isomorphism.
\end{defi}

A slightly modified version of this definition could also be applied to the global past. We restrict the possible output times to $t=1$ so that the global past is indeed in the past of all other agents. Since it has a trivial input, i.e. corresponds to a state $\ket{\psi}^P$, the corresponding version in the causal box picture is then simply the state $\ket{\psi}^P \otimes \ket{t=1}$.

\begin{remark}[A full set of Kraus operators]\label{rem:fullkraus}
As previously mentioned, the Kraus operators \cref{eq:kraustx} are only a complete set (i.e., normalised) on the one-message subspace. As this is the case relevant for modelling QC-QCs, it is not necessary to define the action of the agents on multi-message states. However, a rather trivial extension to the full Fock space can be achieved by adding (1) a single additional Kraus operator $A^{m>1}_i$ (for each agent $A_i$) which is the projection onto the space $\bigoplus_{m>1} \vee^m (\C{H}^{A_i^I} \otimes l^2(\C{T}))$, i.e. the space of states consisting of more than one message (up to relabelling the input system $A_i^I$ to output system $A_i^O$ and shifting the time-stamp $t$ to $t+1$), and (2) the all-time vacuum projector $A^{m=0}_i=\bigotimes_{t\in \C{T}_i^I}\ket{\Omega,t+1}\bra{\Omega, t}$ for the zero-message space. $A^{m>1}_i$ and $A^{m=0}_i$ act as the identity on the $m>1$ and $m=0$ message subspaces respectively, and together with the complete set we have constructed above for $m=1$, these yield a complete set of Kraus operators for the full space. 
\end{remark}

\begin{remark}[Vacuum extensions]
We note that our definition of Kraus operators can be viewed as a vacuum extension \cite{Chiribella_2019}. In the framework of \cite{Chiribella_2019}, the vacuum extension of a channel $A^x$ defined on a non-vacuum space, takes the form $\bar{A}^x = A^x \oplus \gamma^x \proj{\Omega}{\Omega}$ with $\sum_{x} |\gamma^x|^2 = 1$. \Cref{eq:kraustx} (or rather the operators $\bar{A}^{x,t}_i$ which are related to \cref{eq:kraustx} via the vacuum isomorphism \cref{eq:wireisostate}) together with vaccuum projector $\bar{A}^{\Omega,t}_i:=\ket{\Omega,t+1}\bra{\Omega,t}$ thus yield a vacuum extension, with $\gamma^x=0$ for all $x \neq \Omega$ and $\gamma^{\Omega}=1$. Physically, this corresponds to the idea that if agents were to implement a measurement with Kraus operators of this form, they can always accurately distinguish whether there are messages on the wire ($x\neq \Omega$) or not ($x=\Omega$). Note however that when this measurement is not performed, any initial superposition in the temporal order in which non-vacuum states arrive to a party is preserved coherently (as explained before). Moreover, other choices of Kraus operators in the vacuum-extended picture are possible which will correspond to different causal boxes. Indeed, \cite{Chiribella_2019} provides a different vacuum extension of the quantum SWITCH, which can also be considered a causal box. 
Furthermore, while \cref{eq:local_pb} may look like a vacuum extension (or a tensor product of vacuum extensions), it is only an effective version of \cref{eq:kraustx} on the one-message subspace. The normalization condition $\sum_{x} |\gamma^x|^2 = 1$ is not fulfilled as $\gamma^x = 1$ for all $x$. Nevertheless, these operators are normalised on the one-message subspace which follows from the fact that they are equivalent to the original definition on this subspace. Note however that the equivalence breaks when we include states from the zero-message (vacuum at all times) subspace or $n>1$ message subspace.

\end{remark}

We can now formulate what it means for a causal box to be an extension of a QC-QC. This will be easier to do if we first show that sequential composition in the causal box framework is the same as the link product. 

\begin{restatable}[Composition is equivalent to the link product]{lemma}{compislink}
\label{lemma:compislink}
Let $\C{M}_A: \C{L}(\C{H}^A) \rightarrow \C{L}(\C{H}^C)$ and $\C{M}_B: \C{L}(\C{H}^B) \rightarrow \C{L}(\C{H}^D)$ be two CP maps and denote their parallel composition as $\C{M} = \C{M}_A \otimes \C{M}_B$. Then, the Choi matrix of their sequential composition $\C{M}^{C \hookrightarrow B} = (\C{M}_A \otimes \C{M}_B)^{C \hookrightarrow B}$ can be expressed with the link product as

\begin{equation}
    M^{C \hookrightarrow B} = M_A * M_B
\end{equation}
if one takes $\C{H}^B = \C{H}^C$ for the purposes of the link product.
\end{restatable}

This allows us to write the composition of a causal box $\C{C}$ with local agents as

\begin{equation}
    (M_{\bar{A}_1} \otimes ... \otimes M_{\bar{A}_N}) * C
\end{equation}

or in terms of the Choi vector

\begin{equation}
    (\bar{A}_1 \otimes ... \otimes \bar{A}_N) * \dket{\bar{V}},
\end{equation}

which are very reminiscent of the general Born rules \cref{eq:pmprob,eq:pmprob_pure}. Finally, we define what it means for a causal box to be an extension of a QC-QC.

\begin{defi}[Extensions of QC-QCs to causal boxes]\label{def:pbqceq}
Let $\ket{w}$ be the process vector of an $N$-partite QC-QC and $\dket{\bar{V}_w}$ the Choi vector of an $N$-partite causal box with a totally ordered set of positions (of the form depicted on the left side of \cref{fig:cbspace}). We say that the causal box is an extension of a QC-QC if for any set of local operations of the QC-QC $\{A_1,...,A_N\}$ and any set of local operations of the causal box $\{\bar{A}_1,...,\bar{A}_N\}$ such that $\bar{A}_n$ is an extension of $A_n$ (\cref{def:local_equivalence}) for all $n \in \{1,...,N\}$ and any $\ket{\psi}^P \in \C{H}^P$, it holds that
\begin{equation}
    (\dket{A_1} \otimes .... \otimes \dket{A_N} \otimes \ket{\psi}^P) * \ket{w} \otimes \ket{t=2N+2} = (\dket{\bar{A}_1} \otimes ... \otimes \dket{\bar{A}_N} \otimes \ket{\psi, t=1}^P) * \dket{\bar{V}_w}.
\end{equation}
\end{defi}


For simplicity, we just consider pure processes in the above definition. This is not a problem as all QC-QCs and all causal boxes are purifiable \cite{Portmann_2017, Wechs_2021}. 

\begin{remark}[Extensions vs. realisations]
In \cite{Vilasini_2022}, a formal definition of what it means to realise an abstract QC-QC protocol in spacetime was proposed, using the concept of fine-graining introduced there. In \cref{sec:fine}, we connect the mathematical notion of extensions defined here to the more physical concepts of spacetime realisations and fine-graining, showing that the causal box extensions that we construct for QC-QCs can be understood as fine-grained descriptions of their spacetime realisations. 
\end{remark}

\begin{remark}[Hilbert spaces vs. Fock spaces]
The in/output state spaces of QC-QCs/process matrices are finite-dimensional Hilbert spaces of a fixed dimension, while that of causal boxes is potentially an infinite-dimensional Fock space. Mathematically speaking, Fock spaces are a specific type of Hilbert spaces. However, each finite-dimensional Hilbert space of dimension $d$ can be seen as the one-message subspace of an infinite-dimensional Fock space of $d$-dimensional messages. This is the viewpoint we take in our mapping.
\end{remark}


\subsection{Photonic causal box realisation of the Grenoble process}\label{sec:switchtopb}

As a warm up to the general result, in this section, we consider the specific example of the Grenoble process \cite{Wechs_2021} and provide a causal box description for it. Unlike the frequently discussed quantum switch, the Grenoble process features dynamical control of causal order. This means that the causal order depends not just on some fixed control bit in the global past but on the outputs of the agents produced during the run time of the protocol. Additionally, the Grenoble process is not causally separable. 

The Grenoble process is a process between three agents. In a first time step, a qubit is sent in superposition to all three agents. Dynamical control of the causal order happens when the agents send the qubit back to the circuit. In this second time step, the qubit acts as a quantum control for which agent receives it next (so for example, if the circuit receives the qubit from $A_1$, it sends it to $A_2$ if the state of the qubit is $\ket{0}$ and to $A_3$ if the state is $\ket{1}$). In the final step, the circuit sends the qubit to the last remaining agent, attaching an additional ancillary state and applying a CNOT gate to the target-ancilla pair.

A proposed experimental realisation of this QC-QC using photons \cite{Wechs_2021} is depicted in \cref{fig:new_QCQC}. This realisation uses well-known components like coherent COPY and CNOT gates and polarizing beam splitters (PBS). The COPY gates copy the photon's state onto its polarization which the PBS then use to achieve the dynamical control of the causal order.

\begin{figure}
\centering
\includegraphics[width=0.8\textwidth]{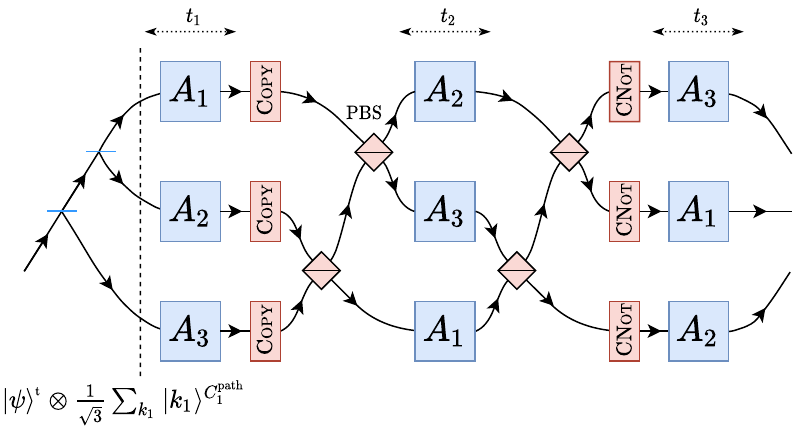}
\caption{Schematic of an experimental implementation of the Grenoble process using photons. The figure is taken from \cite{Wechs_2021}. Initially, the photon is sent to all agents in superposition with the path acting as the control system. In the next step, the state of the target is copied onto the polarization via the COPY gates. This then allows the polarizing beam splitters (PBS) to guide the photon to the correct agent. Finally, PBS along with CNOT gates acting jointly on the target and the polarization implement the final internal operation.}
\label{fig:new_QCQC}
\end{figure}

Let us now consider how we obtain a causal box from this experimental set-up. For this, start with the action of the gates on non-vacuum, qubit states. The COPY gate in \cref{fig:new_QCQC} implements the following operation on an incoming qubit state $V_{\text{Copy}} = \sum_{i=0,1} \ket{i}^{O} \ket{i}^{\alpha} \bra{i}^I$ while the CNOT gates apply the following operation on the two input qubits, $V_{\text{CNOT}} = \sum_{ij} \ket{i}^O \ket{i \oplus j}^\alpha \bra{i}^I \bra{j}^\alpha$. It can be easily checked that both of these are isometries. A beam splitter consists of two input wires $I^0$ and $I^1$ and two output wires $O^0$ and $O^1$, each of which can carry two-dimensional messages. The action of the beam splitter is that it reflects one type of polarization (say photons in state $\ket{0}$) and transmits those of the orthogonal polarization ($\ket{1}$). We can thus say that a polarizing beam splitter sends a photon in state $\ket{i}$ arriving on wire $I^j$ to the output wire $O^{i \oplus j}$. If there is only one photon, it is clear that the beam splitter acts as an isometry.



For a causal box representation, we also need to consider the possibility of multiple (qubit) messages as well as no messages on each of the wires. There are many ways in which one can extend the action of these gates to the Fock space. One natural extension involves having it act on each photon independently essentially. To do this we define orthonormal occupation number bases for qubits and pairs of qubits (to model the joint system of target and ancilla). A basis for qubits is given by $\ket{m, n}$ for $m,n \in \mathbb{N}$, which is the state with $m$ zero-qubits and $n$ one-qubits. For the ququarts we similarly have $\ket{m,n,k,l}$ for $m,n,k,l \in \mathbb{N}$, which is the state with $m$ qubit pairs in the state $00$, $n$ $01$-pairs, $k$ $10$-pairs and $l$ $11$-pairs. The action of the aforementioned gates is then

\begin{gather}\label{eq:grenoble_extension}
\begin{aligned}
    V_{\text{Copy}}(\ket{m,n}^I) &= \ket{m,0,0,n}^{O \alpha} \\
    V_{\text{CNot}}(\ket{m,n,k,l}^{I \alpha}) &= \ket{m,n,l,k}^{O \alpha} \\
    V_{\text{PBS}}(\ket{m,n}^{I_0} \ket{k,l}^{I_1}) &= \ket{m,l}^{O_0} \ket{k,n}^{O_1} 
\end{aligned}
\end{gather}

The extended gates are then also isometries since they map an orthonormal basis to an orthonormal basis in an injective manner. They also reproduce the correct behaviour in the one-message subspace and are thus indeed Fock space extensions of the corresponding gates after which they are named. The one-message subspace is spanned by the states $\ket{1,0} \equiv \ket{0}$ and $\ket{0,1} \equiv \ket{1}$ for qubits and $\ket{1,0,0,0} \equiv \ket{00}, \ket{0,1,0,0} \equiv \ket{01}, \ket{0,0,1,0} \equiv \ket{10}, \ket{0,0,0,1} \equiv \ket{11}$ for pairs of qubits. Using these equivalences we see that $V_{\text{Copy}}$ in \cref{eq:grenoble_extension} correctly maps $\ket{1,0}^I \equiv \ket{0}^I$ to $\ket{1,0,0,0}^{O\alpha} \equiv \ket{0}^O \ket{0}^\alpha$ and $\ket{0,1}^I \equiv \ket{1}^I$ to $\ket{0,0,0,1}^{O \alpha} \equiv \ket{1}^O \ket{1}^\alpha$. Correct behavior for the CNOT gate and the polarising beam splitter can be shown similarly. 

Finally, we need to add the time stamps. We can do so after composing the gates and beam splitters into the various isometries that take the outputs of the local agents to their inputs. If $\bar{V}$ is such an isometry we can replace it with $\bar{V} \otimes \C{O}$ where $\C{O}$ simply increases the time by 1, $\C{O}\ket{t} = \ket{t+1}$. The operator $\C{O}$ is an isometry and therefore $\bar{V} \otimes \C{O}$ is an isometry as well. 

In conclusion, arbitrary compositions and tensor products of COPY and CNOT gates and polarizing beam splitters are isometries. Any sequence representation made up of these components is thus a causal box. Further, if these components are composed in some way to form a QC-QC, then the causal box obtained via the same composition is an extension. This is because when we restrict the components to the zero- and one-message space (at each time), they must behave the same (in the sense of \cref{def:pbqceq}) as in the original QC-QC picture by construction. At the same time, the components conserve the number of messages which means that once we compose the causal box with local agents as defined in \cref{def:local_equivalence}, which do the same, the number of messages in the circuit never changes. We can thus make the restriction to the zero- and one-message space because we start out with a single photon (this is necessarily the case because otherwise the composition would not define a QC-QC). This models \cref{assumption: assump1} as discussed in \cref{sec:statespace}.

In particular, this shows that the resultant causal box satisfies \cref{assumption: assump1} and \cref{assumption: assump2} (as formalised in \cref{sec:statespace}) and we have found a causal box extension of the Grenoble process which respects the requirements of \cref{mainth}.

\begin{figure}
    \centering
    \includegraphics{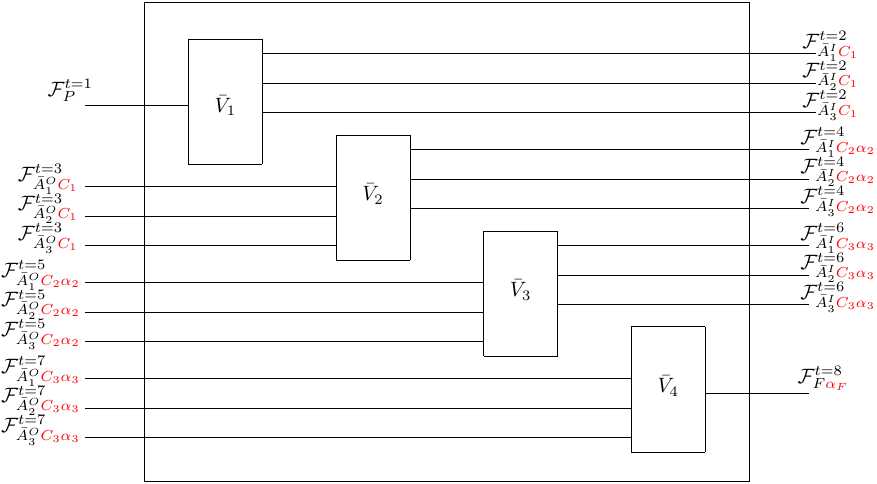}
    \caption{Schematic depiction of the sequence representation causal box extension of the Grenoble process. The internal operations $\bar{V}_n$ are themselves compositions of other operations, namely the COPY and CNOT gates and polarizing beam splitters. The control and ancilla are part of the target system and as such are sent to the agents together with the target. However, we assume that the agents do not act on them. Once this causal box is composed is composed with agents' local operations only the single-message sector is relevant, yielding behaviour that can be seen as equivalent to that of the QC-QC.}
    \label{fig:grenobleextensionl}
\end{figure}

\subsection{Isometric extension}\label{sec:another}

We now construct a general scheme to extend QC-QCs to causal boxes. However, this will not be a direct generalization of the extension provided for the Grenoble process in the previous section. One reason is that the latter was based on a photonic experimental realisation proposed for the Grenoble process while no such explicit proposal exists for general abstract QC-QCs (nevertheless, we give an alternative general extension inspired by such photonic experiments in \cref{sec:controlintarget}, which uses more abstract operations of the QC-QC framework instead of COPY and CNOT gates and polarizing beam splitters). 
A key difference arising between the photonic realisation and the abstract QC-QC description is that in the latter, the ancilla and the control are ``inside" the QC-QC, whereas for the former, these systems and the target were simply different degrees of freedom of the same photon and thus travel ``outside'' the QC-QC and into the local operations (although the devices implementing the local operations only act on the target degree of freedom). Thus, in the photonic case, a new control and ancilla are automatically introduced, whenever a new photon is introduced in the experiment. On the other hand, when keeping these systems internal to the QC-QC, there will only be one control and ancilla.

A concrete issue that arises when trying to formulate a causal box extension where the control is internal is that we can have states with a mismatch between the control and target degrees of freedom, which by construction never arises in the QC-QC picture. That is the control system is in a state $\ket{\C{K}_n, k_{n+1}}$ (which indicates that agent $k_{n+1}$ was the last to perform a non-trivial action), but there is a vacuum state on the wire of the agent $k_{n+1}$ and a non-vacuum state on the wire of a different party. 
Consider a causal box that during the first time step simply forwards the target system in a superposition to every agent. The control system is then $\sum_{k_1 \in \C{N}} \ket{\emptyset, k_1}$. If we then input a non-vacuum state $\ket{\psi}$ on a single wire (say wire 1, for example) in the next time step, the isometry $\bar{V}_2$ will receive $\sum_{k_1 \in \C{N}} \ket{\psi, t=3}^{\bar{A}^O_1} \otimes \ket{\emptyset, k_1}$. Only the term corresponding to $k_1 = 1$ is valid. The way we will account for this is by noting that ``correct" messages and ``mismatched" messages lie in different, orthogonal Hilbert spaces. We will thus define the action of the causal box on each subspace separately, and as long as each defines an isometry on its own domain, the full map will also be an isometry. Moreover, on the correct one-message subspace, it will reproduce the action of the QC-QC.

\begin{figure}
    \centering
    \includegraphics{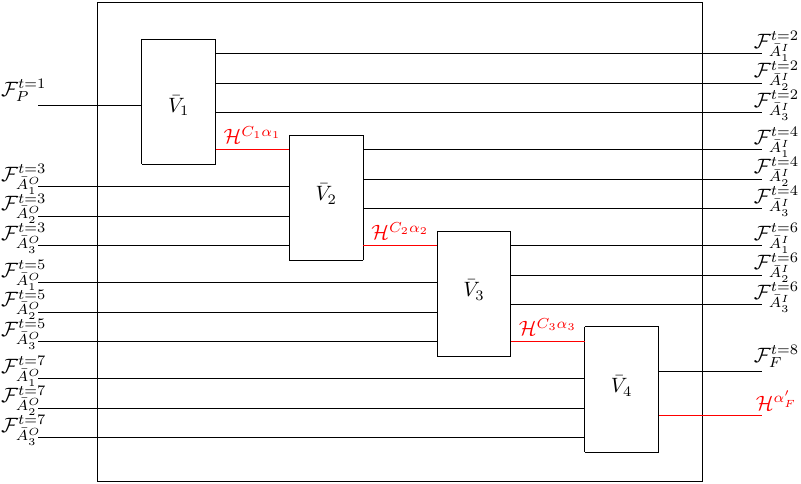}
    \caption{Schematic depiction of the sequence representation causal box extension of a QC-QC, here for three agents. The internal operations $\bar{V}_n$ conserve the number of message coming from the agents. The control and ancilla are internal wires and are marked in red in the figure. Once this causal box is composed is composed with agents' local operations only the single-message sector is relevant, yielding behaviour that can be seen as equivalent to that of the QC-QC.}
    \label{fig:another}
\end{figure}

A similar procedure can be used to deal with multiple message. We define the internal operation of the causal box as a direct sum

\begin{equation}\label{eq:directsum}
    \bar{V}_{n+1} = \bigoplus_{m=0}^{\infty} \bar{V}_{n+1}^m
\end{equation}

where $\bar{V}_{n+1}^m$ acts on states with a total of $m$ messages on the agents' output wires at time $t=2n+1$ and outputs to an $m$-messages state over the input wires of the agents at time $t=2n+2$\footnote{Note that previously when we talked about an $m$-message state, in the context of local operations, we were referring to the number of messages at all times but on the wire of a single agent whose operation was being considered (\cref{fig:pbstatespaces}). Here, we are considering the internal operations of the causal box modelling the supermap, where each operation acts on $m$ messages at a fixed time but spread over the wires of all agents (\cref{fig:another}). Note that \cref{assumption: assump1} and \cref{assumption: assump2} enforce that in the former case $m=1$ but not in the latter case. However, by construction in the QC-QC framework (and w.l.o.g), at each (time) step, only one message is received by its internal operation, which is why we recover the original QC-QC action on the one-message space here as well.  } , i.e.

\begin{gather}
\begin{aligned}
    \bar{V}_1^m: &\vee^m (\C{H}^P \otimes \ket{t=1}) \rightarrow \Big(\bigoplus_{\substack{m_1,...,m_N: \\ \sum_i m_i = m}} \bigotimes_{k_1} \vee^{m_{k_1}}(\C{H}^{\bar{A}^I_{k_1}} \otimes \ket{t=2})\Big) \otimes \C{H}^{C_1} \otimes \C{H}^{\alpha_1} \\
    \bar{V}_{n+1}^m: &\Big(\bigoplus_{\substack{m_1,...,m_N: \\ \sum_i m_i = m}} \bigotimes_{k_n} \vee^{m_{k_n}} (\C{H}^{\bar{A}^O_{k_n}} \otimes \ket{t=2n+1})\Big) \otimes \C{H}^{C_n} \otimes \C{H}^{\alpha_n} \\
    &\rightarrow \Big(\bigoplus_{\substack{m_1,...,m_N: \\ \sum_i m_i = m}} \bigotimes_{k_{n+1}} \vee^{m_{k_{n+1}}}(\C{H}^{\bar{A}^I_{k_{n+1}}} \otimes \ket{t=2n+2})\Big) \otimes \C{H}^{C_{n+1}} \otimes \C{H}^{\alpha_{n+1}} \\
    \bar{V}_{N+1}^m: &\Big(\bigoplus_{\substack{m_1,...,m_N: \\ \sum_i m_i = m}} \bigotimes_{k_N} \vee^{m_{k_N}} (\C{H}^{\bar{A}^O_{k_N}} \otimes \ket{t=2N+1})\Big) \otimes \C{H}^{C_N} \otimes \C{H}^{\alpha_N} \rightarrow \vee^m (\C{H}^F \otimes \ket{t=2N+2}) \otimes \C{H}^{\alpha'_F}
\end{aligned}
\end{gather}

where $\alpha'_F$ is an ancilla such that the corresponding ancilla $\alpha_F$ from the QC-QC is a subspace. Note that \cref{remark:infinite} ensures that defining the operations for each $m$ as above is sufficient to have well-defined causal box on the full Fock space (assuming, of course, that we are actually dealing with isometries, which we prove below).


Let us now deal with $\bar{V}_{n+1}^1$, which acts on the one-message subspace, by formalising what we outlined earlier. We define for all $n \in \C{N}=\{1,...,N\}$ two subspaces,

\begin{equation}
    \C{H}^{O_n}_{\text{corr}} = \text{span}\{\C{H}^{\bar{A}^O_{k_n}} \otimes \ket{t=2n+1} \otimes \ket{\C{K}_{n-1}, k_n} \otimes \C{H}^{\alpha_n}\}_{k_n, \C{K}_{n-1}}
\end{equation}

and its orthogonal complement (relative to the one-message subspace)
\begin{equation}
    \C{H}^{O_n}_{\text{mis}} = (\C{H}^{O_n}_{\text{corr}})^\perp = \text{span}\{\C{H}^{\bar{A}^O_{k_n}} \otimes \ket{t=2n+1}) \otimes \ket{\C{K}_{n-1}, k'_n} \otimes \C{H}^{\alpha_n}\}_{k_n \neq k'_n, \C{K}_{n-1}}.
\end{equation}

For the edge case of the global past, we simply have $\C{H}^{O_0}_{\text{corr}} = \C{H}^P \otimes \ket{t=1}$ and thus $\C{H}^{O_0}_{\text{mis}}$ is trivial.

For the agents' input spaces, we similarly define

\begin{gather}
\begin{aligned}
    \C{H}^{I_{n+1}}_{\text{corr}} &= \text{span}\{\C{H}^{\bar{A}^I_{k_{n+1}}} \otimes \ket{t=2n+2}) \otimes \ket{\C{K}_n, k_{n+1}} \otimes \C{H}^{\alpha_n}\}_{k_{n+1}, \C{K}_n} \\
    \C{H}^{I_{n+1}}_{\text{mis}} &= \text{span}\{\C{H}^{\bar{A}^I_{k_{n+1}}} \otimes \ket{t=2n+2}) \otimes \ket{\C{K}_n, k'_{n+1}} \otimes \C{H}^{\alpha_n}\}_{k_{n+1} \neq k'_{n+1}, \C{K}_n}.
\end{aligned}
\end{gather}

For the edge case of the global future, we define

\begin{gather}
\begin{aligned}
    \C{H}^{I_{N+1}}_{\text{corr}} &= \C{H}^F \otimes \ket{t=2N+2} \otimes \C{H}^{\alpha_F} \\
    \C{H}^{I_{N+1}}_{\text{mis}} &=  \C{H}^F \otimes \ket{t=2N+2} \otimes (\C{H}^{\alpha_F})^\perp.
\end{aligned}
\end{gather}

The spaces with the subscript corr can be seen as messages where the control system is in the correct state, whereas the spaces with the subscript mis correspond to states where the control is mismatched. In order to reproduce the action of the internal operations $V_{n+1}$ of the QC-QC, the action on $\C{H}^{O_n}_{\text{corr}}$ must then necessarily be

\begin{equation}\label{eq:anotherone}
    \bar{V}_{n+1}^1 \ket{\psi, t=2n+1}^{\bar{A}^O_{k_n}} \ket{\C{K}_{n-1}, k_n} \ket{\alpha}^{\alpha_n} = V_{n+1} (\ket{\psi}^{\bar{A}^O_{k_n}} \ket{\C{K}_{n-1}, k_n} \ket{\alpha}^{\alpha_n}) \ket{t=2n+2} \in \C{H}^{I_{n+1}}_{\text{corr}}
\end{equation}

where $\ket{\psi} \neq \ket{\Omega}$ and we interpret $V_{n+1} (\ket{\psi}^{\bar{A}^O_{k_n}} \ket{\C{K}_{n-1}, k_n} \ket{\alpha}^{\alpha_n})$ via the obvious isomorphism between the spaces $\C{H}^{A^{I/O}_{k_n}}$ and $\C{H}^{\bar{A}^{I/O}_{k_n}}$. The restrictions to the respective ``correct" spaces of the edge cases $\bar{V}_1^1, \bar{V}_{N+1}^1$ are defined analogously.

On the other hand, we simply demand that the restriction of $\bar{V}_{n+1}^1$ to $\C{H}^{O_n}_{\text{mis}}$ is an isometry and that it is a map of the form

\begin{equation}
    \bar{V}_{n+1}^1|_{\C{H}^{O_n}_{\text{mis}}}: \C{H}^{O_n}_{\text{mis}} \rightarrow \C{H}^{I_{n+1}}_{\text{mis}}
\end{equation}

For $m \neq 1$, it similarly suffices for $\bar{V}_{n+1}^m$ to be an arbitrary isometry. 

\begin{restatable}[Sequence representation for QC-QCs]{prop}{anotheriso}\label{prop:anotheriso}
There exist maps $\bar{V}_1,...,\bar{V}_{N+1}$ which are isometries and fulfill all the conditions of \cref{sec:another}.
\end{restatable}

From this it immediately follows that we have constructed a sequence representation of a causal box. We can then show that this causal box is also an extension of the original QC-QC.

\begin{restatable}[Causal box description of QC-QCs]{prop}{anotherequivalence}\label{prop:anotherequivalence}
The causal box with sequence representation given by any set of isometries $\bar{V}_1,..., \bar{V}_{N+1}$ satisfying the conditions in \cref{sec:another} is an extension of the QC-QC described by the process vector $$\ket{w} = \sum_{(k_1,..., k_N)} \dket{V^{\rightarrow k_1}_{\emptyset, \emptyset}} * ... * \dket{V^{\rightarrow k_N}_{\{k_1,...,k_{N-2}\}, k_{N-1}}} * \dket{V^{\rightarrow F}_{\C{N}\backslash k_N, k_N}}.$$
\end{restatable}

\subsection{Projective extension}\label{sec:internalcontrol}

The extension that we just constructed allows for multiple messages to be sent to the causal box at the same time. However, since the extension is equivalent to the QC-QC, once it is actually composed with local agents, there is at any time at most one non-vacuum state on the input wires of the causal box. An experimentalist interested in simulating just the usual action of the QC-QC might therefore not care about the multi-message space and build their device in such a way that it simply aborts the procedure if an unexpected (meaning unobtainable in the QC-QC framework) input is encountered. We can capture this idea by applying a projective measurement $\{\C{P}^{I_n}_{accept}, 1 - \C{P}^{I_n}_{accept}\}$ on the inputs and $\{\C{P}^{O_n}_{accept}, 1 - \C{P}^{O_n}_{accept}\}$ on the outputs, each with outcomes $\{accept, abort\}$, before and after every internal operation. Here, $\C{P}^{I_n/O_n}_{accept}$ projects on $\C{H}^{I_n/O_n}_{\text{corr}}$ defined in the previous section. If the outcome is $abort$, the circuit outputs the state $\ket{abort}$. If the outcome is $accept$, the circuit applies the isometry $\bar{V}_{n+1}$. 

When we compose this with local agents satisfying \cref{assumption: assump1} and \cref{assumption: assump2}, this extension never aborts.

\begin{restatable}[No aborts]{lemma}{accept}\label{lemma:accept}
When composing the projective extension with local agents satisfying \cref{assumption: assump1} and \cref{assumption: assump2}, the projective measurements yield the outcome $accept$ with unit probability.
\end{restatable}

Further, note that a projective measurement which always yields the same outcome on some space acts as the identity on that space. Therefore, this causal box acts equivalently to the one from \cref{sec:another} when composed with local operations which are extensions of the QC-QC operations and acts as a trace preserving map in this context. However, on the full Fock space, this causal box extension can decrease the trace unlike the extension of \cref{sec:another} which remains an isometry (and hence trace preserving) on the full space.  

\begin{figure}
    \centering
    \includegraphics{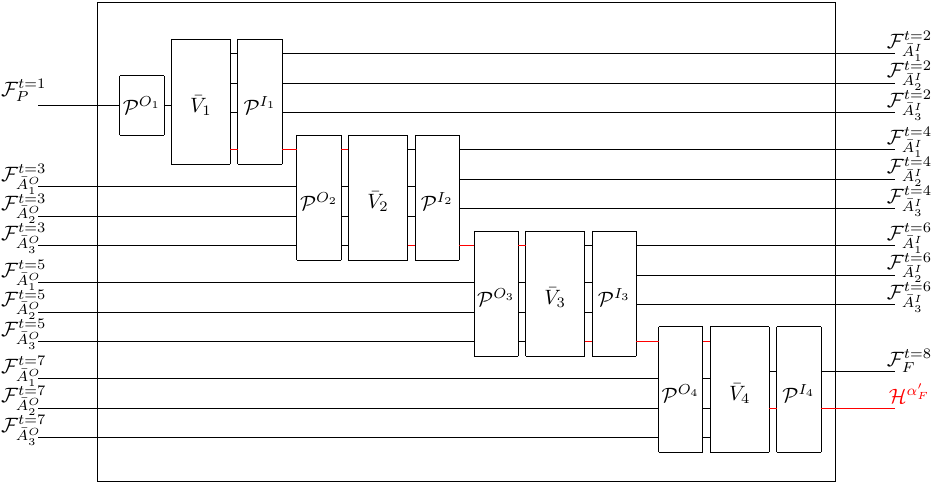}
    \caption{Schematic depiction of the sequence representation causal box extension of a QC-QC, here for three agents, which uses projective measurements to check that there is only a single message on the wires and the control matches the wire. The wires carrying the ancilla and control are marked in red. Once this causal box is composed with agents' local operations the projective measurements always yield the outcome $accept$.}
    \label{fig:another2}
\end{figure}

An additional consequence is that whenever we consider causal boxes which fulfill the requirements laid out in \cref{sec:statespace} and only care about their composition with local agents, we can replace the Choi vector with an effective Choi vector \cite{Vilasini_2020},

\begin{equation}
    \bigotimes_{n=0}^N \dket{\C{P}^{I_n}_{accept}} * \dket{\C{P}^{O_n}_{accept}} * \dket{\bar{V}_w}.
\end{equation}

Naturally, an analogous statement can be made for Choi matrices.

\section{Insights on physical realisability in a spacetime}
\label{sec:finegraining}

The previous sections focused on \cref{mainth}, which is about the mapping from QC-QCs to causal boxes. Here, we describe the relevant concepts behind \cref{theoremfinegraining} and provide a proof sketch of the same, which enables us to regard the causal box in the image of our mapping as a fine-grained description of the QC-QC and consequently as a spacetime realisation of the given QC-QC.

In \cite{Vilasini_2022}, a formal definition of a spacetime realisation of an abstract quantum network, and a relativistic causality condition necessary for ensuring the absence of superluminal causal influences were formalised. These quantum networks include in particular, all protocols described by process matrices and the definition applies to all globally hyperbolic spacetimes i.e. those without closed timelike curves. This definition places minimal but essential conditions that relate the abstract information-theoretic description of the protocol to the description of its spacetime realisation. It shares analogies with how we define the physical realisation of a quantum computation on $n$ logical qubits in terms of $m\geq n$ physical qubits. The concept of fine-graining of causal structures was introduced to formalise this definition, and allows one to interpret the spacetime realisation as a fine-grained description of the abstract process matrix description, as the former can in general involve a larger number of degrees of freedom. We review the core aspects of this definition here, before moving on to the proof sketch of \cref{theoremfinegraining}. In the appendix we provide the full proof of the theorem.

\subsection{Fine-graining and spacetime realisations}

\paragraph{Fine-graining of quantum protocols} Suppose we have a quantum channel (a CPTP map) $\C{M}: \C{L}(\C{H}^I)\rightarrow \C{L}(\C{H}^O)$ which we wish to ``implement''  by some physical means where the input $I$ is encoded in a set $\{I_1,...,I_n\}$ of physical input systems and the output $O$ is decoded from a set $\{O_1,...,O_m\}$ of physical output systems. For example, in quantum computing, $\C{M}$ can be a quantum gate such as the bit flip on a single logical qubit and it may be implemented on a physical hardware where each logical qubit corresponds to $n$ physical qubits (in this example, $n=m$). The physical layer map from inputs $\{I_1,...,I_n\}$ to outputs $\{O_1,...,O_m\}$ is denoted as $\C{M}^f$.

To regard $\C{M}^f$ as an ``implementation'' of $\C{M}$, or in the language we will use, a \emph{fine-graining} of $\C{M}$, it is necessary that there exists an encoding and decoding scheme between the coarse in/outputs $I$ and $O$ and fine in/outputs $\{I_1,...,I_n\}$ and $\{O_1,...,O_m\}$ which can recover $\C{M}$ from $\C{M}^f$. Formally, we require that there exists a pair of maps $(\mathrm{Enc}, \mathrm{Dec})$, where $\mathrm{Enc}: \C{L}(\C{H}^I) \rightarrow \bigotimes_{i=1}^n\C{L}(\C{H}^{I_i})$ is CPTP and $\mathrm{Dec}: \bigotimes_{j=1}^m\C{L}(\C{H}^{O_j}) \rightarrow \C{L}(\C{H}^{O})$ is CPTP on the image of $\C{M}^f\circ\mathrm{Enc}$ such that
    \begin{equation}
    \label{eq: finegraining}
        \C{M}=\mathrm{Dec}\circ \C{M}^f\circ\mathrm{Enc}
    \end{equation}

\begin{figure}
    \centering
    \includegraphics{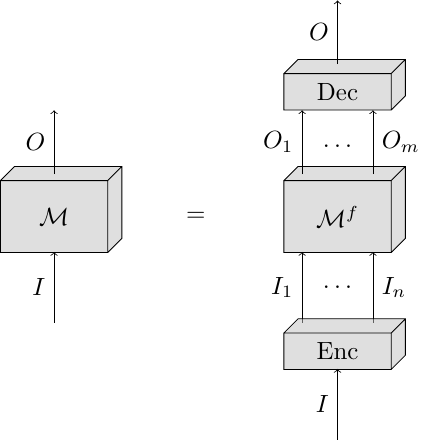}
    \caption{Illustration of \cref{eq: finegraining}. Figure taken from \cite{Vilasini_2022}.}
    \label{fig:encoder-decoder}
\end{figure}

This is illustrated in \cref{fig:encoder-decoder}. A second condition we will require is that the causal properties of $\C{M}$ are preserved in $\C{M}^f$, where the relevant property we care about are the signalling relations: the ability to communicate information between in/output systems of the channel. Formally, we say that $\{I\}$ signals to $\{O\}$ in $\C{M}$ if there exists a map $\C{N}^I: \C{L}(\C{H}^I)\rightarrow \C{L}(\C{H}^I)$ such that $\C{M}\neq \C{M}\circ \C{N}^I$, which captures that an agent with access to $I$ can communicate information to an agent with access to $O$ by a suitable local operation $\C{N}^I$ on $I$. More generally, one can define signalling between arbitrary subsets of in and output systems in CPTP maps with multiple in and output systems (hence the set symbol for $I$ and $O$ here) \cite{Vilasini_2022} (see also \cite{Ormrod_2023} for equivalent definitions of signaling). Then we will require that: $\{I\}$ signals to $\{O\}$ in $\C{M}$, then $\{I_i\}_{i=1}^n$ signals to $\{O_j\}_{j=1}^m$ in $\C{M}^f$. 

More generally, given a channel $\C{M}$ with any number of inputs and outputs, another map $\C{M}^f$ is called a fine-graining of $\C{M}$ if each input $I$ of $\C{M}$ can be identified with a distinct set $\bar{I}$ of inputs of $\C{M}^f$, and likewise for outputs such that (1) there is an encoding-decoding scheme relating these in/output systems which satisfies \cref{eq: finegraining} and (2) the signalling relations of $\C{M}$ are preserved in $\C{M}^f$. The concept has also been extended to quantum networks formed by the composition of multiple CPTP maps, which covers a large class of abstract quantum information protocols, even those with a cyclic causal structure. For defining fine-graining for such general networks of CPTP maps, one applies the above definition for each CPTP map involved in the network. We refer to \cite{Vilasini_2022} for further details of this definition.



\paragraph{Spacetime realisation of abstract quantum protocols} How does this concept of fine-graining relate to the spatiotemporal picture? When we realise a CPTP map such as $\C{M}$ above (or more generally a network of CPTP maps) in a spacetime, we are assigning a spacetime region to each in/output quantum system of the map (or network). In the most minimal, yet operational sense, we can model an acyclic spacetime (such as Minkowski spacetime) as a partially ordered set $\C{T}$ and regard any subset of points in the partially ordered set as defining a spacetime region.\footnote{Note that generally, if the subsets of $\C{T}$ defining the regions are not individual elements of $\C{T}$, then the possibility for causal influence between regions need not be transitive. However, in the case where we choose regions that coincide with individual elements of $\C{T}$ (or are sufficiently localised around single points), we do have transitivity as $\C{T}$ is a poset.} Suppose that we have a spacetime embedding $\C{E}$ which assigns to the in and output systems $I$ and $O$, the spacetime regions $\C{E}(I)=\{P_1,...,P_n\}$ and $\C{E}(O)=\{Q_1,...,Q_m\}$ consisting of some spacetime locations $P_i, Q_j\in \C{T}$. Then, a spacetime realisation of $\C{M}$ relative to this embedding would be a fine-graining $\C{M}^f: \bigotimes_{i=1}^n\C{L}(\C{H}^{I_i})\rightarrow \bigotimes_{j=1}^m\C{L}(\C{H}^{O_j})$ of $\C{M}: \C{L}(\C{H}^{I})\rightarrow \C{L}(\C{H}^{O})$ where each input $I_i$ of $\C{M}^f$ is associated with a spacetime location $P_i$ and each output $O_j$ is associated with a corresponding spacetime location $Q_j$.

Generally, the $P_i$ and $Q_j$ can themselves be subregions rather than individual spacetime events. The ability to partition a larger region into subregions translates operationally into the ability of agents to perform more fine-grained interventions, where they may perform a different operation within different spacetime subregions of a larger protocol.

\paragraph{Relativistic causality} Since signalling in $\C{M}$ is preserved in $\C{M}^f$, relativistic causality implies that whenever $\{I\}$ signals to $\{O\}$ in $\C{M}$, then the embedding must be such that at least one $Q_j\in \C{E}(O)$ must be in the future light cone of some $P_i\in \C{E}(I)$, denoted $P_i\prec Q_j$. There will in general be additional relativistic causality conditions at the fine-grained level, for instance if $\{I_i\}$ signals to $\{O_j,O_k\}$ in fine-graining $\C{M}^f$, then the spacetime point $P_i$ associated with $I_i$ must be in the past light cone of at least one of the spacetime points $Q_j$ or $Q_k$ associated with the corresponding outputs. This captures that even under the more fine-grained interventions, agents cannot signal outside the future light cone. Again, most generally, we can extend this operational definition of relativistic causality for spacetime realisations of arbitrary networks of CPTP maps \cite{Vilasini_2022}.

\begin{figure}[ht]
    \centering
\subfloat[]{\includegraphics[scale=0.5]{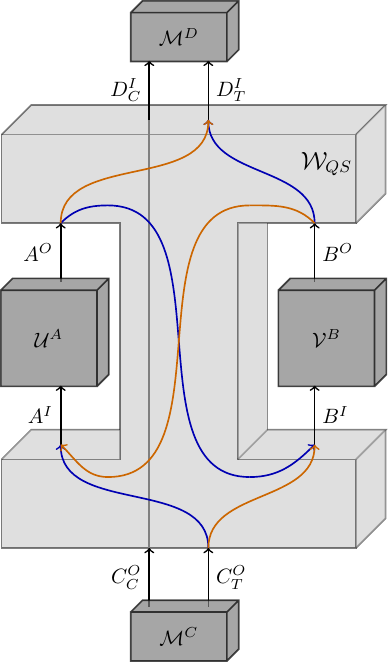}}\qquad\qquad\qquad\subfloat[]{    \includegraphics[scale=0.5]{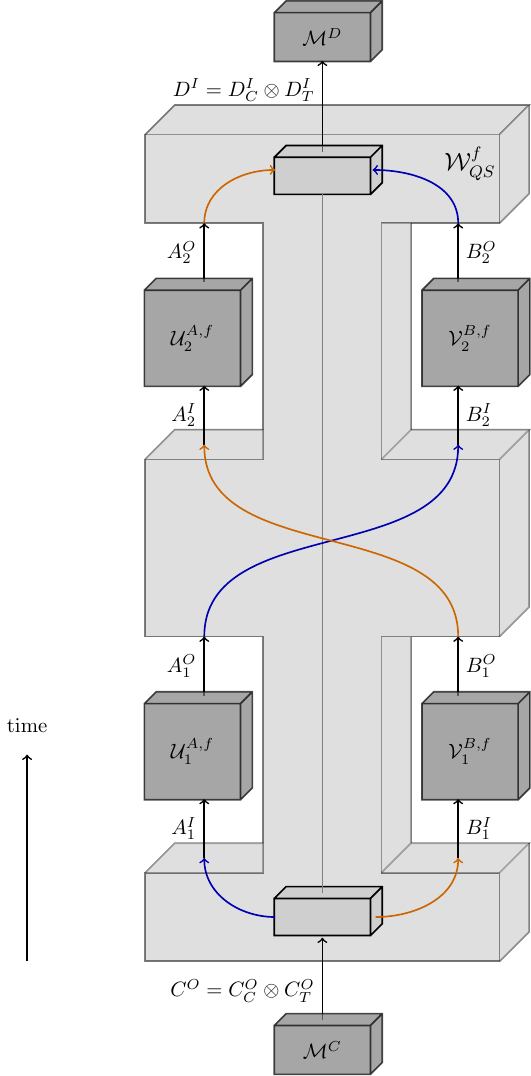}}
    \caption{(a) The coarse-grained process matrix/QC-QC of the quantum switch, which admits an indefinite causal order. Here the target takes the blue path (Alice before Bob) when the control is $\ket{0}$ and the orange path (Bob before Alice) when the control is $\ket{1}$ and a superposition of the two in general. (b) A causal box model of the quantum switch, which corresponds to a fine-grained description that is associated with a definite and acyclic causal structure. The non-vacuum system takes the blue path (and the vacuum takes the orange path) when the control is $\ket{0}$ and vice versa when the control is $\ket{1}$. This figure is reproduced from \cite{Vilasini_2022}. In this work, we have shown that every QC-QC admits a causal box extension with such an acyclic fine-grained causal structure, providing a similar physical interpretation to all QC-QCs in the context of a background space-time.}
    \label{fig:QS}
\end{figure}

\subsection{All QC-QCs admit an acyclic fine-grained causal structure}\label{sec:fine}
We sketch the proof of \cref{theoremfinegraining}. The full proof can be found in \cref{sec:fineproof}.

\begin{enumerate}
    \item \emph{QC-QC to causal box mapping} We consider the image of our mapping from QC-QCs to causal boxes, given by the isometric extension (\cref{sec:another}) and recall that every protocol described by a QC-QC acting on some local operations can be described as a network of CPTP maps formed through parallel or loop composition, as defined and proven in \cite{Vilasini_2022}.
    \item \emph{Encoding-decoding scheme} We then construct an explicit encoding-decoding scheme which relates the causal box in the image of the above mapping to the original QC-QC as per \cref{eq: finegraining} where the causal box plays the role of the fine-grained map $\C{M}^f$ and the QC-QC plays the role of the coarse-grained map $\C{M}$.
    \item \emph{Preserving signalling relations} We show that the signalling relations of the QC-QC are preserved in the causal box, as a consequence of the properties that a causal box must satisfy to be regarded as an extension of the QC-QC satisfying the spatiotemporal closed labs assumption (according to \cref{mainth}).
    \item \emph{Extending the argument to the full network} 
    We also take into account the local operations and show that the whole network formed by composing the causal box modelling the QC-QC with corresponding local operations is a fine-graining of the original QC-QC network. 
    \item \emph{Acyclic fine-grained causal structure} The fact that the fine-grained network is completely described by valid causal boxes with sequence representation in terms of a totally ordered set $\C{T}$ guarantees that its information-theoretic causal structure is acyclic and compatible with the time ordering given by $\C{T}$.

\end{enumerate}
 
We have shown that the causal box realises the action of the QC-QC in a spatiotemporal setting while respecting the spatiotemporal closed labs assumptions and relativistic causality. The causal box preserves the signalling structure of the QC-QC as well as other relevant information captured by \cref{eq: finegraining}. However, \cref{theoremfinegraining} highlights that this does not imply that we have genuinely implemented the indefinite causal structure of the QC-QC. Indefinite causal structures in the sense of the process matrix framework can be viewed as special cases of cyclic information-theoretic causal structures and cyclic causal structures can become acyclic under fine-graining \cite{Vilasini_2022}. This accords with the physical interpretation of cyclic causal structures describing feedback, such as mutual causal influence between demand and price, where the fine-grained description accounting for spacetime degrees of freedom is acyclic. The causal boxes in the image of our mapping from QC-QCs provide an analogously clear and physical interpretation to the causal structure of the QC-QC, at least in the context of their realisations in a background spacetime. We already discussed the physical relevance of this result in \cref{sec:resultsoverview}, and the intuition is illustrated in \cref{fig:QS} for the quantum switch process which is the simplest example of a QC-QC.

\section{Insights on composability}\label{sec:composability}

Composability is central to our understanding of the physical world, as it helps us comprehend how complex structures are built up from constituent parts. Moreover, we expect the composition of any two physical experiments to be yet another physical experiment, i.e. we expect physical processes to be closed under composition. This is indeed the case for causal boxes. They describe information processing protocols in spacetime such that we can compose two protocols described by causal boxes to obtain a new protocol that can also be described as a causal box. 

On the other hand, previous work  \cite{Gu_rin_2019} suggests that composition in the process matrix framework is not so straightforward, with no-go results pointing to difficulties in composing even fixed order processes. In light of this, further works have proposed consistent rules for composition of higher order processes \cite{Kissinger_2019} as well as necessary and sufficient conditions for identifying when the composition of two processes is valid \cite{Jia_2018}. This prior literature focuses on abstract information-theoretic approaches, and it is not immediate how this discussion relates to the observed composability of physical experiments in spacetime, where both the information and spatiotemporal aspects play a role. Here, we address this relation in light of our results.



We start by discussing an example from \cite{Gu_rin_2019, Jia_2018}, which was used to highlight the difficulties in consistently composing process matrices. Consider two bipartite processes, $W_1$ over the parties $A_1$ and $B_1$, and $W_2$ over the parties $A_2$ and $B_2$. Suppose $W_1$ is a fixed order process with the order $A_1\prec B_1$ and $W_2$ is a fixed order process with the opposite order, $B_2\prec A_2$. The parallel composition of the two processes, namely $W_1\otimes W_2$, is no longer a valid bipartite process over the parties $A_1A_2$ and $B_1B_2$. This is because $W_1\otimes W_2$ involves a channel in both directions between the parties $A_1A_2$ and $B_1B_2$ which corresponds to a paradoxical causal loop, as illustrated and explained in \cref{fig:pmcomp}.

\begin{figure}\label{fig:nocomp}
    \centering
    \includegraphics[]{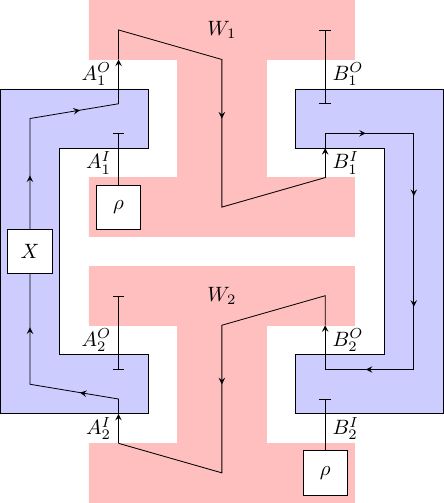}
    \caption{Composition of two process matrices $W_1$ and $W_2$. $W_1$ sends a state to Alice and then sends her output to Bob, while $W_2$ does the same but with the order of Alice and Bob reversed. The composition then includes a loop and if Alice applies the Pauli-$X$ unitary while Bob applies the identity the generalised Born rule yields a probability of 0, even though it should have yielded 1 if this composition were well defined.}
    \label{fig:pmcomp}
\end{figure}

However, notice that $W_1\otimes W_2$ is still a valid four-party process over the parties $A_1$, $A_2$, $B_1$ and $B_2$. The difference being that, when we treat it as a bipartite process, we allow for arbitrary operations $\mathcal{M}_{A_1A_2}: \C{L}(\C{H}^{A_I^1 A_I^2}) \rightarrow \C{L}(\C{H}^{A_O^1 A_O^2})$ and $\mathcal{M}_{B_1 B_2}: \C{L}(\C{H}^{B_I^1 B_I^2}) \rightarrow \C{L}(\C{H}^{B_O^1 B_O^2})$ and when we treat it as a four-party process, only product operations are allowed, $\mathcal{M}_{A_1A_2}=\mathcal{M}_{A_1}\otimes \mathcal{M}_{A_2}$ and $\mathcal{M}_{B_1 B_2}=\mathcal{M}_{B_1}\otimes \mathcal{M}_{B_2}$. Generally, composition of processes allowing only product local operations and composition allowing general local operations acting jointly on in/outputs from the two processes can be formulated as two distinct types of composition \cite{Kissinger_2019}. In particular, if we compose two $N$-party processes $W_1$ and $W_2$ always through the former type of composition (parallel composition), then $W_1\otimes W_2$ is always a valid process, but over $2N$ parties.


The solution to only use product operations aligns with the motivation of the process framework where each party is considered to be in a closed and isolated lab which can communicate to the outside world only through the in/output interface that they have with the process matrix. However, in practice, proposed realisations of process matrices entail table top experiments where the different parties do not physically correspond to closed and isolated labs, and relativistic causality does permit non-product operations in these scenarios. It remains unclear what physical principles impose the restriction to parallel composition in physical experiments in spacetime. Moreover, notice that even in the abstract framework, $W_1\otimes W_1$ and $W_2\otimes W_2$ from the previous example are both valid bipartite processes on the joint Alice and Bob systems, and are consistent under composition with arbitrary joint operations of the two Alices and two Bobs. Therefore, the restriction to parallel composition for all processes, is not necessary for consistency (see also \cite{Jia_2018}). 



We now discuss the insights provided by our work, regarding the relation to the spatiotemporal picture. Our results which map QC-QCs to causal boxes shows how composability of physical experiments is recovered, without restricting to parallel composition by fiat. We illustrate this for the above example. We can map the fixed order processes (and hence QC-QCs) $W_1$ and $W_2$ to the causal box picture (which adds a (space-)time stamp to each message) and consider general physical operations that $A_1A_2$ and $B_1B_2$ can perform. The problematic composition in \cref{fig:pmcomp} is due to the loop from $A^O_1$ to $B^I_1$ to $B^{O}_2$ to $A^{I}_2$ to $A^O_1$ again. In the causal box description, we would have messages with (space-)time labels going around the wires defining this loop. W.l.o.g. for this argument, we can start on any of the systems in the loop, say $A^O_1$, and consider a message, say some state $\ket{0}$, with a time stamp $t_1$ sent by Alice. Due to relativistic causality, each channel in the loop can only send the message forward in time, such that the message reaches $B^I_1$ only at some $t_2>t_1$, then $B^{O}_2$ at some $t_3>t_2$, $A^{I}_2$ at some $t_4>t_1$ (so far all are identity channels so the message is still $\ket{0}$) and finally back to the wire $A^O_1$ at time $t_5>t_4$ after the bit flip, which gives $\ket{1,t_5}^{A^O}$, and the cycle can continue indefinitely.\footnote{Note that if the timestamps don't satisfy the right ordering, the loop cannot be formed in the first place.} There is no contradiction because the mismatched messages on $A^O_1$ are associated with distinct times: $\ket{0,t_1}^{A^O_1}$ and $\ket{1,t_5}^{A^O_1}$ with $t_1<t_5$, and this will always be the case even if we continue the loop.

We thus see that the composition of causal boxes may also lead to cycles when we compose with non-product operations on $A_1A_2$ and $B_1B_2$, but these are well defined due to relativistic causality (indeed, such loops are only a coarse-grained representation of a fine-grained acyclic causal structure, as corroborated by \cref{sec:finegraining}). 
Generally, such compositions can lead to scenarios with superpositions of receiving two messages (e.g., $\ket{0,t_1}\ket{1,t_5}$) and one message (e.g., $\ket{0,t_1}\ket{\Omega,t_5}$). Fully understanding the composability of experimental realisations of QC-QCs therefore requires the use of
Fock spaces (or a similar space that models sending of multiple messages), which we have carefully accounted for in our results and mapping of QC-QCs to causal boxes.

Observe that if such a loop occurs in the causal box picture, the associated causal box protocol would no longer satisfy the set-up assumptions of the process framework which requires in particular that each party acts exactly once (which we referred to as the spatiotemporal closed lab assumption). Therefore, while causal boxes (like the set of physical experiments) are closed under arbitrary composition, the subset of causal boxes which model the set-up assumptions of the process framework are not closed under composition \cite{Vilasini_2020}. The image of our mapping from QC-QCs to causal boxes lies within this subset as we have shown. This highlights that whenever the composition of the causal boxes in the image of two QC-QCs yields a causal box that violates the set-up assumptions (e.g., by allowing an agent who acts on a $d$-dimensional system in the process picture to receive or send multiple $d$-dimensional messages), then the corresponding composition of the QC-QCs would be invalid in the QC-QC framework. In other words, invalid QC-QC compositions still map to physical experiments, but these will be experiments that violate the set up assumptions of the QC-QC/process matrix framework. This highlights the important role played by the interplay of relativistic causality, Fock spaces and the set up assumptions of the process framework in recovering the composability of physical experiments.

\section{Conclusions and outlook}
\label{sec:conclusion}

We compared two quantum causality frameworks— QC-QCs \cite{Wechs_2021}, an abstract information-theoretic model, and causal boxes \cite{Portmann_2017}, which involve both information and spatiotemporal structures. Our findings reveal that every QC-QC can be mapped to a causal box, adhering to the process/QC-QC framework's set-up assumptions (such as each party acting once) and operating on a Fock space while replicating the QC-QC's behaviour in the relevant subspace as defined by these assumptions. The resulting causal box exhibits a definite, acyclic fine-grained causal structure consistent with spacetime, even if the QC-QC displays an indefinite causal structure (the physical interpretation of which is discussed in \cref{sec:resultsoverview}). We established multiple such mappings, illustrating that the same abstract protocol (described by a QC-QC) may have multiple distinct physical realisations (described by causal boxes). Moreover, our results demonstrate that the failure to consistently compose two processes in the abstract framework does not lead to any paradoxes in the physical world (as we should rightfully expect), but translates to a physical protocol that violates the set up assumptions of the process framework, by producing in/outputs in the larger Fock space with multiple messages. Our results offer concrete reconciliations of the notions of causality and composability between the information-theoretic and spatiotemporal viewpoints, providing a foundation for further explorations of this interface.

There is a large volume of work (for example, \cite{Chiribella_2012, Colnaghi_2012, Ara_jo_2014}) suggesting advantages of indefinite causal order processes over standard quantum circuits for information processing tasks. However, there are important distinctions between the abstract picture and the fine-grained spatiotemporal picture as highlighted by our work. More specifically, a QC-QC is a supermap that accepts as inputs arbitrary quantum operations as ``black boxes'' while in the spacetime realisation, adherence to the set-up assumptions of acting once on a non-vacuum system, constrains the local operations to leave the vacuum state invariant, making them ``gray boxes''. This distinction has also been crucial in previous contexts where tasks that are theoretically impossible with black box unitaries (such as coherently controlling them) become experimentally feasible with ``gray box'' unitaries, where the action even on a one-dimensional subspace is known \cite{Friis_2014,Zhou_2011}. Together with previous works such as \cite{Vilasini_2022, vanrietvelde2021universal}, our results can inform a systematic investigation of whether and how proposed theoretical advantages from higher order processes translate into practical quantum experiments in spacetime, while providing more fine-grained tools for this purpose.

The original use of symmetric Fock spaces in the causal box framework aims to encapsulate the fact that all ordering information is contained in the position labels \cite{Portmann_2017}. In the causal box extension of the photonic realisation of the Grenoble process in \cref{sec:switchtopb} and another similar extension found for arbitrary QC-QCs in \cref{sec:controlintarget}, the symmetrization aligns with the perspective of viewing the QC-QC as a photonic (i.e. bosonic) circuit. There is scope for further inquiry into whether causal boxes could be realised using antisymmetric Fock spaces to describe fermionic circuits. The use of Fock spaces in the causal box framework also motivates its utility in studying quantum protocols in a quantum field theory (QFT) setting. In the context of applying the process formalism in QFT, an issue that was noted in \cite{Faleiro2020} is that some QFT processes would require a variable number of operations (e.g. radiative beta decay has a variable number of outgoing photons). As the causal box framework does not assume that each agent acts exactly once and allows superpositions of number of messages, it can be a promising candidate for addressing such questions.

This work focused on the mapping from QC-QCs to causal boxes, which lays the foundation for a follow-up work where we consider mappings in the reverse direction. Specifically, we are interested in causal box protocols which satisfy the set up assumptions of the process framework. These assumptions are operationally relevant also in the context of so-called causal inequalities which are known to be violated by certain theoretical processes, but interestingly, not by QC-QC type processes. If these assumptions are not respected in an experiment, this can open up loopholes which enable the trivial violation of causal inequalities, akin to violating Bell inequalities using classical communication \cite{Vilasini_2019}. Recalling the result of \cite{Vilasini_2022} that causal boxes describe the most general finite quantum protocols in a fixed spacetime, such a mapping will enable a characterisation of the largest set of processes realisable in a fixed spacetime in accordance with the set-up assumptions. We note that neither \cite{Vilasini_2022} nor our current or follow-up work implies that process matrices should be a subset of causal boxes. Indeed, we expect neither framework to be a subset of the other. On the one hand, there are causal boxes which violate the set-up assumptions. They thus trivially cannot correspond to an $N$-party process (although, they could be understood as an $M$-party process for some $M > N$ by splitting agents such that the set-up assumptions are respected). On the other hand, we expect that there exist at least some processes which cannot be realised in a fixed background spacetime. Consequently, these cannot be modelled by a causal box.

A related framework that accounts for vacuum states is that of routed quantum circuits. The concept of routes gives a natural and general way to model constraints on how unitaries act within orthogonal subspaces (such as vacuum vs. non-vacuum spaces), but the formalism does not explicitly consider relativistic causality in spacetime \cite{Ormrod_2023, vanrietvelde2021routed,vanrietvelde2021universal,vanrietvelde2022consistent}. An interesting open conjecture in the formalism states that processes without weak loops (which are cyclic causal influences relative to a more refined notion of causal structure for processes proposed in \cite{vanrietvelde2022consistent}) are precisely those that do not violate causal inequalities. Our results may be seen as further evidence towards this conjecture. QC-QCs which do not violate causal inequalities have an acyclic fine-grained causal structure and are thus realisable in an acyclic background spacetime (in the sense of \cite{Vilasini_2022}) while our follow-up work mentioned above, suggests a tighter relation between QC-QCs and the set of processes realisable in an acyclic spacetime. This acyclicity bears similarities to the absence of weak loops (in the sense of \cite{vanrietvelde2022consistent}). Further work is needed to make precise the link between causal boxes and routed circuits. This can provide a more comprehensive set of tools towards the characterisation of different classes of quantum causal structures, their realisability in different physical regimes/spacetimes, as well as the resource responsible for any potential information-processing advantages that one class of processes may provide over another. 

Finally, we note that there is potential for generalising our results beyond fixed background spacetimes, thanks to the general modelling of spacetime as a partially ordered set. For quantum protocols in fixed spacetimes (with classical reference frames/agents) this partial order as well as the localisation of operational events in spacetime is reference frame independent i.e. all agents will agree on whether or not a quantum message is localised in spacetime. This may no longer be the case when quantum reference frames are used, and it would be interesting to consider a generalisation of the causal box approach and our results by allowing the partial order to be reference frame or agent dependent. This can inform a better understanding of the interplay of operational and spatiotemporal notions of causality and the role of composability in quantum gravitational scenarios, where an absolute background partial order may not exist. By connecting two distinct frameworks featuring different causality notions, this work provides a step forward towards the larger goal of build a more unifying understanding of quantum and relativistic causality.

\bigskip

{\it Acknowledgements} We are thankful to Augustin Vanrietvelde for useful comments regarding the composability of higher order processes, to Renato Renner for insightful discussions on Fock spaces and to \v{C}aslav Brukner for interesting discussions on the notion of events and fine-grained causal structures. The authors also thank the anonymous referees for QPL 2022 and TQC 2023 for their reviews, which have helped to improve this work. MS acknowledges support from National Science Centre, Poland, grant number 2021/41/B/ST2/03149. MS thanks the Quantum Information Theory group at the Institute of Theoretical Physics, ETH Zurich, for financial support to present the results of this work at QPL 2022, Oxford, UK. VV's research has been supported by an ETH Postdoctoral Fellowship. VV also acknowledges support from the ETH Zurich Quantum Center, the Swiss National Science Foundation via project No.\ 200021\_188541 and the QuantERA programme via project No.\ 20QT21\_187724.

\newpage

\printbibliography

\appendix

\section{Mathematical tools}\label{sec:tools}


\subsection{Choi isomorphism}\label{sec:choi}

Frequently, it is more convenient to express channels as matrices. For this purpose, we use the Choi isomorphism. 

\begin{defi}[Choi matrix \cite{CHOI1975285}]
Given some Hilbert spaces $\C{H}^X$ with some basis $\{\ket{i}^X \}_i$ and $\C{H}^Y$ with some basis $\{\ket{i}^Y \}_i$, let $\C{M}: \C{L}(\C{H}^X) \rightarrow \C{L}(\C{H}^Y)$ be some linear map. The Choi matrix of $\C{M}$ is defined as

\begin{equation}
    M \coloneqq \C{I}^X \otimes \C{M} \dket{\mathbb{1}} \dbra{\mathbb{1}}^{XX}
\end{equation}

where $\dket{\mathbb{1}}^{XX} = \sum_i \ket{i}^X \ket{i}^X$ denotes the maximally entangled state.

\end{defi}

If $\C{M}$ is pure, i.e. there exists a linear operator $V: \C{H}^X \rightarrow \C{H}^Y$ such that $\C{M}(\rho) = V \rho V^\dagger$ for any $\rho \in \C{L}(\C{H}^X)$, we can instead work with the simpler Choi vector.

\begin{defi}[Choi vector  \cite{CHOI1975285, Royer1991}]
Given some Hilbert spaces $\C{H}^X$ with some basis $\{\ket{i}^X \}_i$ and $\C{H}^Y$ with some basis $\{\ket{i}^Y \}_i$, let $V: \C{H}^X \rightarrow \C{H}^Y$ be a linear operator. The Choi vector of $V$ is defined as 

\begin{equation}
    \dket{V} \coloneqq \mathbb{1}^X \otimes V \dket{\mathbb{1}}^{XX}.
\end{equation}
\end{defi}

Note that the maximally entangled state is the Choi vector of the identity, which justifies our use of the notation $\dket{\mathbb{1}}^{XX}$ to represent this state.

\subsection{Link product}\label{sec:link}

The link product \cite{Chiribella_2008, Chiribella_2009} is a useful tool which allows us to calculate the Choi matrix/vector of a composition of maps from the Choi matrices/vectors of the maps themselves.

\begin{defi}[Link product for vectors \cite{Wechs_2021}]\label{def:linkvec}
Let $\C{H}^X, \C{H}^Y, \C{H}^Z$ be non-overlapping Hilbert spaces. For two vectors $\ket{a} \in \C{H}^{XY}, \ket{b} \in \C{H}^{YZ}$ we define their link product

\begin{equation}
    \ket{a} * \ket{b} \coloneqq \sum_i \ket{a_i} \otimes \ket{b_i}
\end{equation}

where $\ket{a_i} = (\mathbb{1}^X \otimes \bra{i}^Y) \ket{a}$ and $\ket{b_i} = (\bra{i}^Y \otimes \mathbb{1}^Z) \ket{b}$.

\end{defi}

Let us make a few remarks about the above definition. If the shared space $\C{H}^Y$ is trivial, the link product simplifies to the tensor product $\ket{a}^X * \ket{b}^Z = \ket{a}^X \otimes \ket{b}^Z$. If the shared space $\C{H}^Y$ is the only non-trivial space, the link product simplifies to the inner product, $\ket{a}^Y * \ket{b}^Y = \braket{\bar{a}|b}$ where the entries of $\ket{\bar{a}}$ are the complex conjugated entries of $\ket{a}$ \cite{Wechs_2021}. 

\begin{defi}[Link product for matrices \cite{Chiribella_2008, Chiribella_2009}]\label{def:linkmat}
Let $\C{H}^X, \C{H}^Y, \C{H}^Z$ be non-overlapping Hilbert spaces. For two operators $A \in \C{L}(\C{H}^{XY}), B \in \C{L}(\C{H}^{YZ})$ we define their link product

\begin{gather}
\begin{aligned}
    A * B &\coloneqq \tr_Y ((A^{T_Y} \otimes \mathbb{1}^Z)(\mathbb{1}^X \otimes B)) \\
    &= \sum_{ij} A_{ij} \otimes B_{ij}
\end{aligned}
\end{gather}
where $T_Y$ denotes the partial transpose on $\C{H}^Y$ and $A_{ij} = (\mathbb{1}^X \otimes \bra{i}^Y)A(\mathbb{1}^X \otimes \ket{j}^Y), B_{ij} = (\bra{i}^Y \otimes \mathbb{1}^Z )A(\ket{j}^Y \otimes \mathbb{1}^Z)$.
\end{defi}

Once again, let us consider how this simplifies for trivial Hilbert spaces. If $\C{H}^Y$ is trivial, we obtain the tensor product $A*B = A \otimes B$. If $\C{H}^{XZ}$ is trivial, the link product becomes the trace $A*B = \tr(A^T B)$ \cite{Chiribella_2008, Chiribella_2009}. 

The link product is commutative for both vectors and operators (up to reordering of the resulting tensor product). The $n$-fold link product over vectors $\ket{a_k} \in \C{H}^{A_k}$ or operators $M_k \in \C{L}(\C{H}^{A_k})$ is also associative provided that each Hilbert space appears at most twice, i.e. $\C{H}^{A_k} \cap \C{H}^{A_l} \cap \C{H}^{A_m} = \emptyset$ for all $k \neq l \neq m \neq k$. The link product of hermitian and/or positive semidefinite maps is hermitian and/or positive semidefinite again \cite{Chiribella_2008, Chiribella_2009}.

Let us now discuss the reasons for introducing the link product. We will often want to calculate the Choi vector of a linear operator of the form $V = (\mathbb{1}^{X'} \otimes V_2)(V_1 \otimes \mathbb{1}^{Z})$, where $V_1: \C{H}^{X} \rightarrow \C{H}^{X'Y}$ and $V_2: \C{H}^{YZ} \rightarrow \C{H}^{Z'}$. It then holds that \cite{Wechs_2021}

\begin{equation}
    \dket{V} = \dket{V_1} * \dket{V_2}.
\end{equation}

A similar statement can be made for the Choi matrix of a map $\C{M} = (\C{I}^{X'} \otimes \C{M}_2)(\C{M}_1 \otimes \C{I}^{Z})$ with $\C{M}_1: \C{L}(\C{H}^X) \rightarrow \C{L}(\C{H}^{X'Y})$ and $\C{M}_2: \C{L}(\C{H}^{YZ}) \rightarrow \C{L}(\C{H}^{Z'})$. In this case we have \cite{Chiribella_2008, Chiribella_2009}

\begin{equation}
    M = M_1 * M_2.
\end{equation}

\begin{figure}
    \centering
    \includegraphics[scale=1.5]{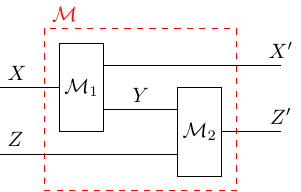}
    \caption{Diagrammatic representation of the action of the link product. The link product turns the circuits defined by the two maps $\C{M}_1, \C{M}_2$ into a new circuit $\C{M}$ from $\C{H}^{XZ}$ to $\C{H}^{X'Z'}$. The wire $\C{H}^Y$ becomes an internal wire.}
    \label{fig:link_circuit}
\end{figure}

Diagrammatically, the action of the link product can be represented by \cref{fig:link_circuit}. The circuits corresponding to the two maps are linked by connecting an output wire of one to the input wire of the other. These two wires correspond to the shared space $\C{H}^Y$.

\subsection{Symmetric tensor product and inner product in Fock spaces}\label{sec:fockinner}

An $n$-message state in a Fock space $\C{F}^\C{T}_A$ can be written as 

\begin{equation}\label{eq:symmprod}
    \bigodot_{k=1}^n \ket{\psi_k, t_k} = \frac{1}{\sqrt{n!}} \sum_{\pi \in S^n} \bigotimes_{k=1}^n \ket{\psi_{\pi(k)}, t_{\pi(k)}}
\end{equation}

(or as a linear combination of such states) where $S^n$ is the symmetric group over $n$ elements. Physically, this corresponds to $n$ messages in the states $\psi_1$,..., $\psi_n$ which arrive at times $t_1$,..., $t_n$ respectively on a wire.

The inner product of such states is given by 

\begin{gather}\label{eq:fock_inner_n}
\begin{aligned}
    \bigodot_{k=1}^n \bra{\psi_k, t_k} \bigodot_{l=1}^n \ket{\phi_l, t'_l} &= \frac{1}{\sqrt{n!}} \frac{1}{\sqrt{n!}} \sum_{\pi \in S^n} \bigotimes_{k=1}^n \bra{\psi_{\pi(k)}, t_{\pi(k)}} \sum_{\pi' \in S^n} \bigotimes_{l=1}^n \ket{\phi_{\pi'(l)}, t'_{\pi'(l)}} \\
    &= \frac{1}{n!} \sum_{\pi, \pi' \in S^n} \prod_{k=1}^n \braket{\psi_{\pi(k)}, t_{\pi(k)}|\phi_{\pi'(k)}, t'_{\pi'(k)}} \\
    &= \sum_{\pi \in S^n} \prod_{k=1}^n \braket{\psi_k, t_k|\phi_{\pi(k)}, t'_{\pi(k)}} \\
    &= \sum_{\pi \in S^n} \prod_{k=1}^n \braket{\psi_k|\phi_{\pi(k)}} \delta_{t_k, t'_{\pi(k)}}.
\end{aligned}
\end{gather}

For arbitrary states in the Fock, we have the following definition.

\begin{defi}[Inner product in Fock spaces \cite{bhatia1996matrix}]
Let $\ket{\psi} = a \ket{\Omega} + \sum_{n=1}^\infty \bigodot_{k=1}^n \ket{\psi_{n,k}, t_{n,k}}, \ket{\phi} = b \ket{\Omega} + \sum_{n=1}^\infty \bigodot_{k=1}^n \ket{\phi_{n,k}, t'_{n,k}} \in \C{F}(\C{H}^A \otimes l^2(\C{T}))$. Then their inner product is given by 

\begin{equation}\label{eq:fock_inner}
    \braket{\psi|\phi} = \overline{a} b + \sum_{n=1}^\infty \sum_{\pi \in S^n} \prod_{k=1}^n \braket{\psi_{n,k}|\phi_{n, \pi(k)}} \delta_{t_{n,k}, t'_{n,\pi(k)}}.
\end{equation}

\end{defi}

\begin{example}[Norm of a Fock space state]
Consider a wire which carries a superposition of no messages and two messages at time $t$, one in state $\ket{0}$ and the other in state $\ket{1}$, such that the probability to obtain either outcome (zero-message or two-messages) when measuring the number of messages is equal. The overall state on the wire can be written as $\ket{\psi} = \frac{1}{\sqrt{2}} \ket{\Omega} + \frac{1}{\sqrt{2}} \ket{0, t} \odot \ket{1, t} = \frac{1}{\sqrt{2}} \ket{\Omega} + \frac{1}{2} (\ket{0, t} \otimes \ket{1, t} + \ket{1, t} \otimes \ket{0, t})$. We can check that this state is normalised either by using the form of $\ket{\psi}$ where we expanded the symmetric tensor product

\begin{gather}
\begin{aligned}
    \braket{\psi|\psi} &= (\frac{1}{\sqrt{2}} \bra{\Omega} + \frac{1}{2} (\bra{0, t} \otimes \bra{1, t} + \bra{1, t} \otimes \bra{0, t}))(\frac{1}{\sqrt{2}} \ket{\Omega} + \frac{1}{2} (\ket{0, t} \otimes \ket{1, t} + \ket{1, t} \otimes \ket{0, t})) \\
    &= \frac{1}{2} \braket{\Omega|\Omega} + \frac{1}{4} (\braket{0, t|0,t} \braket{1,t|1,t} + \braket{1, t|0,t} \braket{0,t|1,t} + \braket{0, t|1,t} \braket{1,t|0,t} + \braket{1, t|1,t} \braket{0,t|0,t}) \\
    &= \frac{1}{2} + \frac{1}{4} 2 = 1
\end{aligned}
\end{gather}
or by using \cref{eq:fock_inner} with $a = b = \frac{1}{\sqrt{2}}, \ket{\psi_{1, 1}} = \ket{\phi_{1, 1}} = \frac{1}{\sqrt{2}} \ket{0}$ and $\ket{\psi_{1, 2}} = \ket{\phi_{1, 2}} = \frac{1}{\sqrt{2}} \ket{1}$

\begin{gather}
\begin{aligned}
    \braket{\psi|\psi} &= \frac{1}{2} \delta_{t, t}+ \frac{1}{2} (\braket{0|0} \braket{1|1} \delta_{t,t} + \braket{0|1} \braket{1|0} \delta_{t,t} \\
    &= \frac{1}{2}  + \frac{1}{2} = 1.
\end{aligned}
\end{gather}

One can also check that the probabilities to find the state to be the vacuum $\ket{\Omega}$ or $\ket{0,t} \odot \ket{1,t}$ are $\frac{1}{2}$ each as required by taking the square of the inner product between these states and $\ket{\psi}$.

\end{example}

\begin{example}[Fock space inner product of a multi-system state]
Consider two wires $A$ and $B$. Let the state on wire $A$ be $\ket{0, t}^A \odot \ket{1, t}^A$ and the state on wire $B$ be $\ket{0, t}^B \odot \ket{0, t'}^B$ with $t \neq t'$. When writing the overall state it does not matter whether we write the state on $A$ or $B$ first, $\ket{0, t}^A \odot \ket{1, t}^A \otimes \ket{0, t}^B \odot \ket{0, t'}^B = \ket{0, t}^B \odot \ket{0, t'}^B \otimes \ket{0, t}^A \odot \ket{1, t}^A$. We just have to make sure that when we take the inner product of such multi-wire states, we separately calculate the inner product of the states on each wire and then take the product over the wires. For example, consider the inner product of the above state with the state where the wires are exchanged

\begin{gather}
\begin{aligned}
    (\bra{0, t}^A& \odot \bra{1, t}^A \otimes \bra{0, t}^B \odot \bra{0, t'}^B)(\ket{0, t}^B \odot \ket{1, t}^B \otimes \ket{0, t}^A \odot \ket{0, t'}^A) \\
    &= (\bra{0, t}^A \odot \bra{1, t}^A)(\ket{0, t}^A \odot \ket{0, t'}^A)(\bra{0, t}^B \odot \bra{0, t'}^B)(\ket{0, t}^B \odot \ket{1, t}^B) = 0
\end{aligned}
\end{gather}
and not 1 as one might expect if one just looked at the order and not the system labels. It is thus important to be aware of the wire label of each message when working with states over multiple wires. 

\end{example}

\begin{remark}[Different conventions for Fock space states]
We note that there are different possible conventions for the symmetric tensor product $\odot$. The one we used in this section is simple to write in terms of the ordinary tensor product (cf. \cref{eq:symmprod}) but the norm of $\bigodot_{k=1}^n \ket{\psi_k, t_k}$ differs in general from the norm of $\bigotimes_{k=1}^n \ket{\psi_k, t_k}$. In particular, a state like $\ket{0} \odot \ket{0} = \frac{1}{\sqrt{2}} (\ket{0} \otimes \ket{0} + \ket{0} \otimes \ket{0}) = \sqrt{2} \ket{0} \otimes \ket{0}$ is not normalised. For applications where normalisation is important, like outcome probabilities, one needs to keep this in mind and properly normalise states. If one is in particular interested in such applications, one may prefer to choose the following convention,
\begin{equation}
    \bigodot_{k=1}^n \ket{\psi_k, t_k} = \frac{\C{N}_0}{\C{N}_s} \sum_{\pi \in S^n} \bigotimes_{k=1}^n \ket{\psi_{\pi(k)}, t_{\pi(k)}}
\end{equation}
where $\C{N}_0$ is the norm of $\bigotimes_{k=1}^n \ket{\psi_{k}, t_{k}}$ and $\C{N}_s$ is the norm of $\sum_{\pi \in S^n} \bigotimes_{k=1}^n \ket{\psi_{\pi(k)}, t_{\pi(k)}}$. With this convention, the norm of the the unsymmetrised and the symmetrised state are the same.



Note that which convention one chooses is essentially just a matter of notation. In particular, it has no effect on the Fock space and its elements, which are what is actually mathematically meaningful. In fact, the main part of this work, including the results in the main part and proofs thereof, is actually completely agnostic to the chosen convention. Only in \cref{sec:controlintarget} we assume the convention of \cref{eq:symmprod} since it makes the results easier to write down and prove. Of course, this is only to simplify equations and one could also use a different convention. 

\end{remark}

\begin{remark}[Inner product and the permanent of a matrix \cite{bhatia1996matrix}]
The inner product in \cref{eq:fock_inner_n} can also equal to the permanent of the matrix $A$ with elements $A_{ij}:=\braket{\psi_i, t_i|\phi_j, t_j}$.
\end{remark}

\subsection{Isometries}\label{sec:isometries}

We briefly review some useful properties of isometries in this section.

\begin{defi}[Isometries]
Let $W^I$ and $W^O$ be inner product spaces. Let $V: W^I \rightarrow W^O$ be a linear map. We call $V$ an isometry if for all $v, w \in W^I$, it holds that $v^\dagger V^\dagger V w = v^\dagger w$, i.e. $V$ conserves the inner product.
\end{defi}

It can be difficult to check that a map $V$ conserves the inner product. The following three equivalent conditions can be used instead.

\begin{lemma}[Characterization of isometries]
\label{lemma:isochar}
Let $V: W^I \rightarrow W^O$ be a linear map. The following four conditions are equivalent.

\begin{enumerate}
    \item $V$ is an isometry.
    \item $V$ conserves the norm, i.e. $w^\dagger V^\dagger V w = w^\dagger w, \forall w \in W^I$.
    \item If $w_1,...,w_N$ is a basis of $W^I$, then $w_i^\dagger V^\dagger V w_j = w_i^\dagger w_j, \forall i,j = 1,...,N$.
    \item If $W^I = \bigoplus_{i=1}^M W^I_i$ where $W^I_1,...,W^I_M$ are orthogonal subspaces that are mapped to orthogonal subspaces by $V$, then the restriction of $V$ to each $W^I_i$, $V|_{W^I_i}$ is an isometry.
\end{enumerate}

\end{lemma}

The fourth condition may seem a bit strange at first sight. It essentially states that if a map maps orthogonal subspaces to orthogonal subspaces, it is enough to check that it is an isometry on each of these subspaces separately. Among other things, this condition allows us to consider each space with a fixed number of messages separately in various proofs in this work.

The second condition is frequently used as the definition of an isometry in place of the definition we used. Slightly weaker versions of the third and fourth condition, stating that $V$ is an isometry iff it maps orthonormal bases to orthonormal bases or iff there exists an orthonormal basis that is mapped to an orthonormal basis, are also often included in standard textbooks (cf. \cite{axler2014linear}).

\begin{proof}

``$1 \implies 2"$: The norm (squared) is an inner product, thus if all inner products are conserved, then in particular all norms are conserved.

``$2 \implies 1$": Let $v, w \in W^I$. Then $|v+w|^2 = |V(v+w)|^2 = |Vv|^2 + |Vw|^2 + v^\dagger V^\dagger V w + w^\dagger V^\dagger V v = |v|^2 + |w|^2 + 2 \text{Re}(v^\dagger V^\dagger V w)$ and also $|v+w|^2 = |v|^2 + |w|^2 + 2 \text{Re}(v^\dagger V^\dagger V w)$ where $\text{Re}$ gives the real part of a complex number. Thus, the real part of the inner product $v^\dagger w$ is conserved and repeating the previous calculation with $v-w$ in place of $v+w$ shows that the imaginary part is also conserved. 

``$1 \implies 3$'': If $V$ conserves the inner product between all vectors then in particular it conserves the inner product between basis vectors as these are a subset of all vectors.

``$3 \implies 1$'': Condition 1 is equivalent to condition 2 so we can consider the norm of a vector $w \in W^I$. We can write $w = \sum_{i=1}^N \lambda_i w_i$ for some $\lambda_i \in \mathbb{C}$ as the vectors $w_i$ form a basis. Then $w^\dagger V^\dagger V w = \sum_{ij} \overline{\lambda}_i \lambda_j w^{\dagger}_i V^\dagger V w_j = \sum_{ij} \overline{\lambda}_i \lambda_j w^{\dagger}_i w_j = w^\dagger w$ where we used condition 3 in the second equality.

``$1 \implies 4$'': If $V$ is an isometry on the whole space $W^I$, it is in particular an isometry on any subspace.

``$4 \implies 1$'': We consider again the norm of an arbitrary vector $w \in W^I$. As $W = \bigoplus_{i=1}^M W^I_i$, we can write $w = \sum_{i=1}^M w^I_i$ for vectors $w^I_i \in W^I_i$. Then, $w^\dagger V^\dagger V w = \sum_{ij} w^{I \dagger}_i V^\dagger V w^I_j = \sum_i w^{I \dagger}_i V^\dagger V w^I_i =  \sum_i w^{I \dagger}_i w^I_i = w^\dagger w$. In the second equality, we used that $V$ conserves the orthogonality of the subspaces $W^I_i$ which causes cross terms to vanish, $w^{I \dagger}_i V^\dagger V w^I_j = 0$ for $i \neq j$, and in the third equality, we used that $V$ is an isometry when restricted to each subspace $W^I_i$.
\end{proof}

Of course, we can also mix and match these conditions. For example, if we have split our input space into orthogonal subspaces as in condition 4, we can then use condition 2 to show that our map is an isometry when restricted to each of the subspaces.

\section{More detailed review of QC-QCs}\label{sec:qcqcdetails}
In this section we will give some more details on QC-QCs. In particular, we dive into how \cite{Wechs_2021} formalises the intuitive picture we reviewed in \cref{sec:circuitpm} as well as the notation used. 

The framework considers $N$ agents, $A_1,...,A_N$ with agent $A_k$ having an input space $\C{H}^{A^I_k}$ and an output space $\C{H}^{A^O_k}$. The aforementioned control system is an element of the Hilbert spaces $\C{H}^{C_n}$, with $n=1,...,N$. The basis states of $\C{H}^{C_n}$ are

\begin{equation}
    \ket{\{k_1,...,k_{n_1}\}, k_n}
\end{equation}

or more compactly

\begin{equation}
    \ket{\C{K}_{n_1}, k_n}
\end{equation}

where $\C{K}_{n-1} = \{k_1,...,k_{n-1}\}$ will stand generically for some subset of $\C{N}$ with $n-1$ elements. The meaning of the above basis state is as follows: the control state $\ket{\C{K}_{n-1}, k_n}^{C_n}$ represents that the agent $A_{k_n}$ acted most recently or is about to act and also that the agents $A_{k_1}$,..., $A_{k_{n-1}}$ acted before $A_{k_n}$, but it does not reveal the order in which these previous agents acted (hence the set notation for $\C{K}_{n-1}$). The reason for this last point is so that different paths can interfere with each other.

The action of the circuit conditioned on the control system being in the state $\ket{\C{K}_{n-1}, k_n}^{C_n}$ and the agent $A_{k_{n+1}}$ being the one to receive the target system, is then described with a linear map

\begin{equation}
    V^{\rightarrow k_{n+1}}_{\C{K}_{n-1}, k_n}: \C{H}^{A^O_{k_n} \alpha_n} \rightarrow \C{H}^{A^I_{k_{n+1}} \alpha_{n+1}}
\end{equation}

where $\alpha_n, \alpha_{n+1}$ are ancillary systems that essentially serve the purpose of an internal memory for the circuit. It is assumed that

\begin{equation}\label{eq:defiso0}
    V_{n+1} = \sum_{\substack{\C{K}_{n-1},\\ k_n, k_{n+1}}} V^{\rightarrow k_{n+1}}_{\C{K}_{n-1}, k_n} \otimes \ket{\C{K}_{n-1} \cup k_n, k_{n+1}} \bra{\C{K}_{n-1}, k_n}
\end{equation}

is an isometry.

In order to describe the full action of the circuit, the QC-QC framework introduces for the $n$-th time step the generic input Hilbert space $\C{H}^{\tilde{A}^I_n}$ and the generic output Hilbert space $\C{H}^{\tilde{A}^O_n}$, which are isomorphic to the agent's input/output spaces, i.e. $\C{H}^{\tilde{A}^{I/O}_n} \cong \C{H}^{A^{I/O}_{k_n}}$ for all possible $k_n$.

The linear maps $V^{\rightarrow k_{n+1}}_{\C{K}_{n-1}, k_n}$ are then embedded in these generic spaces via the isomorphism, obtaining

\begin{equation}
    \tilde{V}^{\rightarrow k_{n+1}}_{\C{K}_{n-1}, k_n}: \C{H}^{\tilde{A}^O_{n} \alpha_n}\rightarrow \C{H}^{\tilde{A}^I_{n+1} \alpha_{n+1}}.
\end{equation}

The global action of the circuit in this time step is given via the controlled application of the above operators

\begin{equation}\label{eq:defiso}
    \tilde{V}_{n+1} = \sum_{\substack{\C{K}_{n-1},\\ k_n, k_{n+1}}} \tilde{V}^{\rightarrow k_{n+1}}_{\C{K}_{n-1}, k_n} \otimes \ket{\C{K}_{n-1} \cup k_n, k_{n+1}} \bra{\C{K}_{n-1}, k_n}
\end{equation}

The sum in the above equations goes over $\C{K}_{n-1} \subset \C{N}$ and $k_n, k_{n+1} \in \C{N}\backslash \C{K}_{n-1}, k_n \neq k_{n+1}$, and this is the assumed convention for all such sums.

The local operations of the agents can also be similarly embedded. Taking them to be pure (the behaviour of non-pure operations can be recovered by considering the behaviour of their Kraus operators), $A_{k_n}: \C{H}^{A^I_{k_n}} \rightarrow \C{H}^{A^O_{k_n}}$, we obtain the counter-parts embedded in the generic spaces 

\begin{equation}
    \tilde{A}^{k_n}: \C{H}^{\tilde{A}^I_{n}}\rightarrow \C{H}^{\tilde{A}^O_{n}}.
\end{equation}

Adding these up just like in the case of the internal operations gives us the total action of the agents during this time step

\begin{equation}
    \tilde{A}_n = \sum_{\C{K}_n, k_{n+1}} \tilde{A}_{k_{n+1}} \otimes \ket{\C{K}_n, k_{n+1}} \bra{\C{K}_n, k_{n+1}}.
\end{equation}

Finally, given a set of linear maps $V^{\rightarrow k_{n+1}}_{\C{K}_{n-1}, k_n}$, the process vector of the corresponding QC-QC can be written as

\begin{equation}
    \ket{w} = \sum_{(k_1,..., k_N)} \dket{V^{\rightarrow k_1}_{\emptyset, \emptyset}} * ... * \dket{V^{\rightarrow k_N}_{\{k_1,...,k_{N-2}\}, k_{N-1}}} * \dket{V^{\rightarrow F}_{\C{N}\backslash k_N, k_N}}.
\end{equation}

A particular example of a QC-QC is the Grenoble process, which we already introduced in \cref{sec:switchtopb}. In that section, we described its behaviour in words. The maps defining it mathematically are \cite{Wechs_2021}

\begin{gather}\label{eq:switchkraus}
\begin{aligned}
    &V^{\rightarrow k_1}_{\emptyset, \emptyset} = \frac{1}{\sqrt{3}} \ket{\psi}^{A^I_{k_1}} \\
    &V^{\rightarrow k_2}_{\emptyset, k_1} = \begin{cases}  \ket{0}^{A^I_{k_2}} \bra{0}^{A^O_{k_1}}, k_2 = k_1 + 1 \pmod 3 \\ \ket{1}^{A^I_{k_2}} \bra{1}^{A^O_{k_1}}, k_2 = k_1 + 2 \pmod 3 \end{cases} \\
    &V^{\rightarrow k_3}_{\{k_1\}, k_2} = \begin{cases}  \ket{0}^{A^I_{k_3}} \ket{0}^{\alpha_3} \bra{0}^{A^O_{k_2}} + \ket{1}^{A^I_{k_3}} \ket{1}^{\alpha_3} \bra{1}^{A^O_{k_2}}, k_2 = k_1 + 1 \pmod 3 \\ \ket{0}^{A^I_{k_3}} \ket{1}^{\alpha_3} \bra{0}^{A^O_{k_2}} + \ket{1}^{A^I_{k_3}} \ket{0}^{\alpha_3} \bra{1}^{A^O_{k_2}}, k_2 = k_1 + 2 \pmod 3 \end{cases}\\
    &V^{\rightarrow F}_{\{k_1, k_2\}, k_3} = \mathbb{1}^{A^O_{k_3} \alpha_3 \rightarrow \alpha_F^{(1)}} \otimes \ket{k_3}^{\alpha^{(2)}_F}.
\end{aligned}
\end{gather}

\begin{remark}\label{remark:fullspace}
In \cite{Wechs_2021}, the authors only demand that $V_{n+1}$ is an isometry on its effective input space, which we can recursively define as the image of $V_n$ restricted to its own effective input space. However, we show here that one can w.l.o.g. assume that $V_{n+1}$ is an isometry on the full space. We first prove by induction that the maps $V_{n+1}$ can always be extended to isometries on the full space and then we prove that these extensions still define the same QC-QC. The effective input space of $V_1$ is just the full space, thus we do not need to extend it. Assume now, that all the maps $V_k$ with $k \leq n$ are isometries on the full space. Now note that $\text{Im}(V_n) \oplus \text{Im}(V_n)^\perp$ is equal to the full input space of $V_{n+1}$. It is then always possible to define an isometry $V'_{n+1}: \text{Im}(V_n)^\perp \rightarrow \C{H}^{A^I_{\C{N}}\beta_{n+1} C_{n+1}}$ which can be written in the form of \cref{eq:defiso0} for some ancillary Hilbert space $\C{H}^{\beta_{n+1}}$. Imposing that $\C{H}^{\alpha_{n+1}}$ and $\C{H}^{\beta_{n+1}}$ be orthogonal, the map $V''_{n+1} = V_{n+1} \oplus V'_{n+1}$ is an isometry due to \cref{lemma:isochar} and is of the form of \cref{eq:defiso0}. The QC-QC defined by these isometries is equivalent to the QC-QC we started out with,

\begin{equation}
    \dket{V''_1} * ... * \dket{V''_{N+1}} = \dket{V''_{N+1}(V''_N(...(V''_1)...)} = \dket{V_{N+1}(V_N(...(V_1)...)}
\end{equation}

where in the first equality we used that $\dket{A}*\dket{B} = \dket{A \circ B}$ (dropping identities) and in the second equality we iteratively used the fact that $V''_{n+1}|_{\text{Im}(V_n)} = V_{n+1}$ and $V''_1 = V_1$.

\end{remark}

\section{A general extension of QC-QCs inspired by photonic experiments}\label{sec:controlintarget}

In this section, we present an alternate way to extend QC-QCs to causal boxes. One way to view the extension here is as as a photonic circuit implementing a QC-QC, similar to how \cref{fig:new_QCQC} is a photonic circuit that implements the Grenoble process. It is, however, more abstract than the extension for the Grenoble process. In \cref{sec:switchtopb}, we modelled specific components making up a photonic circuit, like CNOT and COPY gates and polarising beam splitters. Here, as we will see, we will instead use the more abstract operations $V_{n+1}$ of the QC-QC framework, without looking into how they would implemented in detail in an experimental set-up. 

Considering the experimental set-up in \cref{fig:new_QCQC}, we note that the target, the ancilla and the control are all part of the same overall system as they are all encoded in different degrees of freedom of the photon. Importantly, this means that the ancilla and control are both sent to local agents, but these only act trivially on them. The input/output spaces of the agents will thus be $\C{F}^{\C{T}^{I/O}}_{\bar{A}^{I/O}_i \alpha C}$, but we only allow local operations that leave $\C{H}^\alpha$ and $\C{H}^C$ invariant. This is essentially equivalent to routing these two systems past the local operations and directly to the next internal operation, which is not really different from sending it internally (cf. \cref{fig:routingpast}). We can thus still consider agents whose input and output Hilbert spaces do not include the ancillary and control Hilbert spaces and contract over these spaces when calculating the Choi vector.

\begin{figure}
\centering
\includegraphics[scale=1.5]{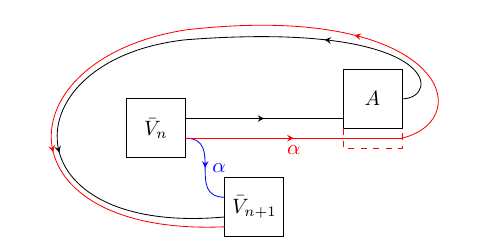}
\caption{Diagrammatic representation of how a causal box can send an ancilla $\alpha$ from one isometry of the sequence representation to the next. The isometry $\bar{V}_n$ can either send it internally to $\bar{V}_{n+1}$ (blue) or send it to the local agent $A$ along with the target system (black line) who then sends it to $\bar{V}_{n+1}$ (red). If $A$ acts trivially on $\alpha$ (indicated by the red, dashed extension of $A$), then the two options are essentially equivalent on the relevant one-message subspace.}
\label{fig:routingpast}
\end{figure}

In analogy to what we did in \cref{sec:switchtopb}, we assume that the circuit acts on each message (photon) separately. For this purpose, we introduce the operation $symm$ which symmetrises tensor products in the same Hilbert space and leaves tensor products over different Hilbert spaces invariant, $symm(\C{H}^X \otimes \C{H}^X) = \C{H}^X \odot \C{H}^X$ and $symm(\C{H}^X \otimes \C{H}^Y) = \C{H}^X \otimes \C{H}^Y$ for $X \neq Y$. More generally, for $N$ Hilbert spaces $\C{H}^{X_n}$ such that $X_n \neq X_m$ if $n \neq m$, we define it as

\begin{equation}
    symm(\bigotimes_{n=1}^N (\ket{\psi^{n, 1}}^{X_n} \otimes \ket{\psi^{n,2}}^{X_n} \otimes ...)) \coloneqq \bigotimes_{n=1}^N (\ket{\psi^{n, 1}}^{X_n} \odot \ket{\psi^{n,2}}^{X_n} \odot ...).
\end{equation}

\begin{example}[Symmetrization]
Consider three systems, $A, B, C$ and the tensor product of states on these systems $\ket{0}^A \ket{1}^A \ket{0}^B \ket{1}^C \ket{1}^C$. The action of $symm$ on this state is

\begin{equation}
    symm(\ket{0}^A \ket{1}^A \ket{0}^B \ket{1}^C \ket{1}^C) = \ket{0}^A \odot \ket{1}^A \otimes \ket{0}^B \otimes \ket{1}^C \odot \ket{1}^C.
\end{equation}

Note that $symm$ gives the same state when applying it to $\ket{1}^A \ket{0}^A \ket{0}^B \ket{1}^C \ket{1}^C$ because $\ket{0}^A \odot \ket{1}^A = \ket{1}^A \odot \ket{0}^A$. In fact, because the order of the tensor factors belonging to different systems does not matter, $\ket{0}^A \ket{1}^B = \ket{1}^B \ket{0}^A$, we can freely exchange any two tensor factors and still obtain the same state after applying $symm$. This is a particularly useful feature. In fact, it is the reason we introduce $symm$ because it allows us to ensure states are properly symmetrised without forcing us to group all messages belonging to the same wire, like we would have to if we used the symmetric tensor product $\odot$.
\end{example}

Given these considerations, we can define the general action of the internal operations in the causal box picture




\begin{gather}
\begin{aligned}
    \bar{V}_1: \C{F}^P \otimes \ket{t=1} &\rightarrow \bigotimes_{k_1} \C{F}(\C{H}^{\bar{A}^I_{k_1} \bar{\alpha}_1 \bar{C}} \otimes \ket{t=2}) \\
    \bar{V}_{n+1}: \bigotimes_{k_n} \C{F}(\C{H}^{\bar{A}^O_{k_n} \bar{\alpha}_n \bar{C}} \otimes \ket{t=2n+1}) &\rightarrow \bigotimes_{k_{n+1}} \C{F}(\C{H}^{\bar{A}^I_{k_{n+1}} \bar{\alpha}_{n+1} \bar{C}} \otimes \ket{t=2n+2}) \\
    \bar{V}_{N+1}: \bigotimes_{k_N} \C{F}(\C{H}^{\bar{A}^O_{k_N} \bar{\alpha}_N \bar{C}} \otimes \ket{t=2N+1}) &\rightarrow \bigotimes_{k_{n+1}} \C{F}(\C{H}^{F \alpha_F} \otimes \ket{t=2N+2}). 
\end{aligned}
\end{gather}

The action of the isometries on single photons is well defined in the QC-QC picture and we simply have to add time stamps and vacuum states to obtain the causal box picture

\begin{gather}
\begin{aligned}\label{eq:photonic_single}
    \bar{V}_{n+1}&(\ket{\psi, t=2n+1}^{\bar{A}^O_{k_n}} \ket{\C{K}_{n-1}, k_n} \bigotimes_{k\neq k_n} \ket{\Omega, t=2n+1}^{\bar{A}^O_{k} C_n}) \\
    &= V_{n+1}(\ket{\psi}^{\bar{A}^O_{k_n}} \ket{\C{K}_{n-1}, k_n}) \ket{t=2n+2} \\
    &= \sum_{k_{n+1}} V^{\rightarrow k_{n+1}}_{\C{K}_{n-1}, k_n}(\ket{\psi}^{A^O_{k_n}}) \ket{t=2n+2} \ket{\C{K}_{n-1} \cup k_n, k_{n+1}} \bigotimes_{k\neq k_n} \ket{\Omega, t=2n+1}^{\bar{A}^O_{k} C_n}.
\end{aligned}
\end{gather}

Additionally, if all wires contain the vacuum, $\bar{V}_{n+1}$ leaves the state invariant besides incrementing the time by 1

\begin{equation}\label{eq:isovac}
    \bar{V}_{n+1}(\bigotimes_{i=1}^N \ket{\Omega,t=2n+1}^{\bar{A}^O_i C_n}) = \bigotimes_{i=1}^N \ket{\Omega,t=2n+2}^{\bar{A}^I_i C_{n+1}}.
\end{equation}

In the multi-message case, we essentially apply the operation to each message individually, analogous to the case of the Grenoble process in \cref{sec:switchtopb}


\begin{gather}
\begin{aligned}\label{eq:pb_ext}
    \bar{V}_1 \bigodot_i &\ket{\psi^i, t=1}^P = symm \Big(\bigotimes_i \sum_{k_1} V^{\rightarrow k_1}_{\emptyset, \emptyset} \ket{\psi^i}^P \ket{t=2} \ket{\emptyset, k_1}\Big) \\
    \bar{V}_{n+1}\Big(\bigotimes_{k_n} &\bigodot_{i=i_{k_n}}^{i_{k_n+1}-1} \ket{\psi^i, t=2n+1}^{\bar{A}^O_{k_n}} \ket{\C{K}^i_{k_{n-1}}, k_n}\Big) \\
    &= symm\Big(\bigotimes_{k_n} \bigotimes_{i=i_{k_n}}^{i_{k_n+1}-1} \bar{V}_{n+1} (\ket{\psi^{i}, t=2n+1}^{\bar{A}^O_{k_n}} \ket{\C{K}^i_{n-1}, k_n})\Big) \\
    &= symm \Big(\bigotimes_{k_n} \bigotimes_{i=i_{k_n}}^{i_{k_n+1}-1} \sum_{k^{k_n}_{n+1}} V^{\rightarrow k^i_{n+1}}_{\C{K}^i_{n-1}, k_n} \ket{\psi^i}^{\bar{A}^O_{k_n}} \ket{t=2n+2} \ket{\C{K}^i_{n}, k_{n+1}^i}\Big) \\
    \bar{V}_{N+1}\Big(\bigotimes_{k_N} &\bigodot_{i=i_{k_N}}^{i_{k_N+1}-1} \ket{\psi^i, t=2N+1}^{A^O_{k_N}} \ket{\C{N} \backslash k_N, k_N}\Big) \\
    &= symm \Big(\bigotimes_{k_N} \bigotimes_{i=i_{k_n}}^{i_{k_n+1}-1} \bar{V}_{N+1} \ket{\psi^i, t=2N+1}^{\bar{A}^O_{k_N}} \ket{\C{N} \backslash k_N, k_N}\Big) \\ &= \bigodot_i V^{\rightarrow F}_{\C{K}_{\C{N} \backslash k_N}, k_N} \ket{\psi^{k_N}}^{\bar{A}^O_{k_N}} \ket{t=2N+2}
\end{aligned}
\end{gather}

where $\C{K}^{k_n}_n \coloneqq \C{K}^{k_n}_{n-1} \cup k_n$ and where $1 = i_1 < i_2 < ... < i_{N+1}$ and we dropped the ancilla to shorten the equations a bit (i.e. $\ket{\psi}^{\bar{A}^I_{k_n}}$ should be understood to include an implicit ancilla, $\ket{\psi}^{\bar{A}^I_{k_n}} \ket{\alpha}^{\alpha_n}$). For the zero-message case, we use \cref{eq:isovac} as the definition, replacing the tensor product over 3 agents with one over $N$ agents.

It should be noted that $symm$ should not be viewed as an operation that is actually carried out by the causal box. It is simply how we denote the fact that messages that are sent from different agents to the same agent end up in the symmetric Fock space of that agent's input space.  Note that we use the symmetric tensor product directly instead of using the more complicated $symm$ in the last line, because all the states are sent to the same agent, the global future, in the last time step.



\begin{restatable}[Sequence representation for QC-QCs]{prop}{symmiso}
\label{prop:symmiso}
The maps $\bar{V}_1,...,\bar{V}_{N+1}$ as defined by \cref{eq:pb_ext} are isometries.
\end{restatable}

The proof is given in \cref{sec:proofapp}. Given \cref{prop:symmiso}, what we constructed is therefore a causal box. This causal box is also an extension of the original QC-QC.

\begin{restatable}[Causal box description of QC-QCs]{prop}{symmequivalence}\label{prop:symmequivalence}
The causal box with sequence representation given by the isometries $\bar{V}_1,..., \bar{V}_{N+1}$ from \cref{eq:pb_ext} is an extension of the QC-QC described by the process vector $\ket{w_{\C{N},F}} = \sum_{(k_1,..., k_N)} \dket{V^{\rightarrow k_1}_{\emptyset, \emptyset}} * ... * \dket{V^{\rightarrow k_N}_{\{k_1,...,k_{N-2}\}, k_{N-1}}} * \dket{V^{\rightarrow F}_{\C{N}\backslash k_N, k_N}}$.
\end{restatable}

\section{Fine-graining}\label{sec:fineproof}

We provide a detailed proof of \cref{theoremfinegraining} based on the proof sketch outlined in \cref{sec:finegraining}, which involves five steps. The first step is to simply note that the action of any QC-QC on local operations in the QC-QC framework can be regarded as a network of CPTP maps as defined and proven in \cite{Vilasini_2022} for general process matrices (of which QC-QCs are a subset). The same holds true for the causal box description. 
Therefore we only need to prove steps 2-5, which we carry out below.

{\bf Step 2 (encoding-decoding scheme):} Consider the CPTP maps corresponding to a given QC-QC $\mathcal{W}$ and the CPTP map describing the causal box $\mathcal{C}$ in the image of our mapping from QC-QC to causal boxes. We will work with the mapping given by the isometric extension, and in this step, construct an encoder $\mathrm{Enc}$ and decoder $\mathrm{Dec}$ such that $\mathcal{W}=\mathrm{Dec}\circ\mathcal{C}\circ\mathrm{Enc}$. This is one of the properties required for $\mathcal{C}$ to be a fine-graining of $\mathcal{W}$.


For simplicity, we begin by considering the case where the ancillas $\alpha_n$ are all trivial. The encoder defined on an arbitrary basis is given by

\begin{equation}
\label{eq: encoder}
    \bigotimes_n \ket{i_n}^{A^O_n} \mapsto \sum_{(k_1,...,k_N)} \lambda^{(k_1,...k_N)}_{i_1,...,i_N} \bigotimes_{n} \ket{i_{k_n}, t=2n+1}^{\bar{A}^O_{k_n}} 
\end{equation}

with $\lambda^{(k_1,...k_N)}_{i_1,...,i_N} \neq 0$ if $V^{\rightarrow k_{n+1}}_{\{k_1,...,k_{n-1}\}, k_n}\ket{i_{k_n}} \neq 0$ for all $n$. To ensure that this is an isometry we also demand $\sum_{(k_1,...,k_N)} |\lambda^{(k_1,...k_N)}_{i_1,...,i_N}|^2 = 1$ (this suffices as the encoder preserves the orthogonality of different basis vectors).

Next, we apply the causal box\footnote{For the sake of concreteness, one can assume that we use our extension from \cref{sec:another}, but in principle any extension which keeps the control and ancillas internally works here}

\begin{equation}\label{eq:applycb}
    \sum_{(k_1,...,k_N)} \lambda^{(k_1,...k_N)}_{i_1,...,i_N} \bigotimes_{n} V^{\rightarrow k_{n+1}}_{\C{K}_{n-1}, k_n} \ket{i_{k_n}}^{A^O_{k_n}} \ket{t=2n+2} + \ket{\text{ortho}}
\end{equation}

where we split the result into two contributions. The first corresponds to states that we would expect from the QC-QC. In the language of \cref{sec:another}, the terms in this contribution correspond to branches where no control mismatch occurred. The contribution $\ket{\text{ortho}}$ corresponds to branches where at least one control mismatch occurred. In particular, $\ket{\text{ortho}}$ is orthogonal to the QC-QC-like contribution on the left. This follows from orthogonality between $\bar{V}^1_{n+1}|_{\C{H}^{O_n}_{\text{corr}}}$ and $\bar{V}^1_{n+1}|_{\C{H}^{O_n}_{\text{mis}}}$. To see this, define $C_{n+1} := \bar{V}^1_{n+1}|_{\C{H}^{O_n}_{\text{corr}}}$ and $M_{n+1} := \bar{V}^1_{n+1}|_{\C{H}^{O_n}_{\text{mis}}}$. Then, the map yielding the QC-QC-like contribution can be written as $C_{N+1} \circ ... \circ C_{n+1} \circ ... \circ C_1$. On the other hand, for the map yielding $\ket{\text{ortho}}$, we can write $\sum C_{N+1} \circ ... \circ M_{n+1} \circ... \circ C_1$, where every term in the sum has at least one of the mismatched operators and the sum goes over all such terms. We then find

\begin{equation}
    C_1^\dagger ... C_{n+1}^\dagger C_{N+1}^\dagger C_{N+1} ... M_{n+1} ... C_1 = C_1^\dagger ... C_{n+1}^\dagger M_{n+1} ... C_1 = 0.
\end{equation}

Additionally, the norm of \cref{eq:applycb} is 1 since the causal box is an isometry.

Next, we define a projection which when acting on \cref{eq:applycb} yields

\begin{equation}\label{eq:cbandproj}
    \sum_{(k_1,...,k_N)} \lambda^{(k_1,...k_N)}_{i_1,...,i_N} \bigotimes_{n} V^{\rightarrow k_{n+1}}_{\C{K}_{n-1}, k_n} \ket{i_{k_n}}^{A^O_{k_n}} \ket{t=2n+2}.
\end{equation}

Note that each term lies in a different orthogonal subspace (i.e. the subspace corresponding to the term's time ordering). We now define the following rescaling map

\begin{equation}\label{eq:scaling}
    \sum_{(k_1,...,k_N)} \lambda^{(k_1,...k_N)}_{i_1,...,i_N} \bigotimes_{n} V^{\rightarrow k_{n+1}}_{\C{K}_{n-1}, k_n} \ket{i_{k_n}}^{A^O_{k_n}} \ket{t=2n+2} \mapsto \sum_{(k_1,...,k_N)} \bigotimes_{n} V^{\rightarrow k_{n+1}}_{\C{K}_{n-1}, k_n} \ket{i_{k_n}}^{A^O_{k_n}}.
\end{equation}

Let us check linearity. We first consider an ordering $(k_1,...,k_N)$ and define vectors $\ket{i^I_{k_{n+1}}} := V^{\rightarrow k_{n+1}}_{\{k_1,...,k_{n-1}\}, k_n} \ket{i_{k_n}}$. We then define the following map

\begin{equation}\label{eq:terms_scaling}
    \bigotimes_{n} \ket{i^I_{k_n}}^{A^I_{k_n}} \ket{t=2n+2} \mapsto \begin{cases}
        0, \text{ if } \ket{i^I_{k_n}}^{A^I_{k_n}} = 0 \text{ for some } n \\
        1/\lambda^{k_1,...,k_N}_{i_1,...,i_N} \bigotimes_{n} \ket{i^I_{k_n}}^{A^I_{k_n}}, \text{ else}
    \end{cases}
\end{equation}

If the vectors $\bigotimes_{n} \ket{i^I_{k_n}}^{A^I_{k_n}}$  which are non-zero are linearly independent, this is well defined and linear. Otherwise, we have to introduce an additional requirement on the $\lambda^{k_1,...,k_N}_{i_1,...,i_N}$. From \cref{eq:terms_scaling}, we can see that if $\bigotimes_{n} \ket{i^I_{k_n}}^{A^I_{k_n}} = \sum_{j_1,...,j_N} a_{j_1,...,j_N} \bigotimes_{n} \ket{j^I_{k_n}}^{A^I_{k_n}}$, then we must have

\begin{equation}
    1/\lambda^{k_1,...,k_N}_{i_1,...,i_N} \bigotimes_n\ket{i^I_{k_n}}^{A^I_{k_n}} = \sum_{j_1,...,j_N} a_{j_1,...,j_N}/\lambda^{k_1,...,k_N}_{j_1,...,j_N} \bigotimes_{n} \ket{j^I_{k_n}}^{A^I_{k_n}}
\end{equation}

Once we impose this condition, \cref{eq:scaling} is also well defined and linear as it is simply the direct sum of maps of the form of \cref{eq:terms_scaling} (here, we make use that all the terms in \cref{eq:cbandproj} lie in orthogonal subspaces).

\begin{remark}[Trivial existence of an encoder]
One may wonder if there always exists an encoder which fulfills the above equation. Fortunately, this is the case. If we set $\lambda^{k_1,...,k_N}_{i_1,...,i_N} = \frac{1}{\sqrt{N}}$, the above equation simplifies to $\ket{i^I_{k_n}}^{A^I_{k_n}} = \sum_{j_1,...,j_N} a_{j_1,...,j_N} \bigotimes_{n} \ket{j^I_{k_n}}^{A^I_{k_n}}$ which is true by assumption.
\end{remark}

Finally, we define the decoder to be the composition of the projector from \cref{eq:cbandproj} and \cref{eq:scaling}. This is linear as it is the composition of two linear maps. Additionally, the RHS of \cref{eq:scaling} is the output of the QC-QC. Thus, the composition of the encoder, the causal box and the decoder is equal to the QC-QC as desired. This also shows that the decoder is an isometry, even though the projector in \cref{eq:cbandproj} necessarily decreases the norm while the map in \cref{eq:scaling} increases the same. This is because if a linear map is composed with an isometry and this composition yields another isometry, then the map itself must be an isometry (on its effective input space).

If the ancillas are non-trivial, we can write everything in terms of Choi vectors. For example, \Cref{eq:cbandproj} can be written as 

\begin{equation}
    \sum_{(k_1,...,k_N)} \lambda^{(k_1,...k_N)}_{i_1,...,i_N} \bigstar_{n} \dket{V^{\rightarrow k_{n+1}}_{\C{K}_{n-1}, k_n}} * \ket{i_{k_n}}^{A^O_{k_n}} \ket{t=2n+2} 
\end{equation}

where $\bigstar_n$ denotes an $n$-fold link product.

The rescaling map then produces

\begin{equation}
    \sum_{(k_1,...,k_N)} \bigstar_{n} \dket{V^{\rightarrow k_{n+1}}_{\C{K}_{n-1}, k_n}} * \ket{i_{k_n}}^{A^O_{k_n}}
\end{equation}

which is once again what we get for the QC-QC.

{\bf Step 3 (preserving signalling relations):}
We recall the definition of signalling. Given a CPTP map $\mathcal{M}$ from a set In$=\{I_1,...,I_n\}$ of input systems to a set Out$=\{O_1,...,O_m\}$ of output systems, we say that a subset $\mathcal{S}_I\subset \mathrm{In}$ of input systems signals to a subset $\mathcal{S}_O\subset \mathrm{Out}$ of output systems in $\mathcal{M}$ if and only if there exists a local operation $\mathcal{N}_{\mathcal{S}_I}$ on $\mathcal{S}_I$ such that 
\begin{equation}
 \tr_{\mathrm{Out}\backslash \mathcal{S}_O}  \circ \mathcal{M} \neq \tr_{\mathrm{Out}\backslash \mathcal{S}_O}  \circ \mathcal{M} \circ \mathcal{N}_{\mathcal{S}_I},
\end{equation}
where tensor factors of the identity channel are kept implicit.\footnote{Explicitly, $\mathcal{N}_{\mathcal{S}_I}$ here corresponds to $\mathcal{N}_{\mathcal{S}_I}\otimes \mathcal{I}_{\mathrm{In}\backslash \mathcal{S}_I}$.} This captures that the output state on the subsystems $\mathcal{S}_O$ can be changed by a suitable choice of local operation $\mathcal{N}_{\mathcal{S}_I}$ on the subsystems $\mathcal{S}_I$. 

A general QC-QC $\mathcal{W}$ is a CPTP map from the input systems $\{A_n^O\}_{n=1}^N$ to the output systems $\{A_n^I\}_{n=1}^N$. In the causal box picture, we will use the shorthand $\{\bar{A}_n^{O,t=2k}\}_{k=1}^N:=\bar{A}^O_n$ and $\{\bar{A}_n^{I,t=2k+1}\}_{k=1}^N:=\bar{A}^I_n$ for this proof. Then the associated causal box $\mathcal{C}$ is a CPTP map from the inputs $\{\bar{A}_n^O\}_{n=1}^N$ to outputs $\{\bar{A}_n^I\}_{n=1}^N$.
Thus, in going from the QC-QC picture to the causal box picture, each in/output system $A^{I/O}_j$ maps to a set of $N$ systems $\bar{A}^{I/O}_j$. For simplicity, we carry out the rest of the proof for the $N=2$ case, and show that if we had $A^O_1$ signals to $A^I_2$ in $\mathcal{W}$, then we would have $\bar{A}^O_1$ signals to $\bar{A}^I_2$ in $\mathcal{C}$. The same proof method readily generalises to the multiparty case and arbitrary signalling relations.

Consider the encoder as defined in \cref{eq: encoder} and recall that each term of the form $\ket{i_{k_n}, t=2n+1}$ is short hand for 

\begin{equation}
\ket{\Omega, t=3}...\ket{\Omega, t=2(n-1)+1}\ket{i_{k_n}, t=2n+1}\ket{\Omega, t=2(n+1)+1}...\ket{\Omega, t=2N+1},
\end{equation}

i.e. having a non-vacuum state $i_{k_n}$ on the output wire of agent $A_{k_n}$ at time $t=2n+1$ and a vacuum state on the same output wire at all other relevant times. Since the encoder is an isometry, we can restrict our attention to pure states when defining our local operation, w.l.o.g. Notice that any operation of the form $\bar{\mathcal{N}}_{k_n}=\otimes_{t=3}^{2N+1} (\mathcal{N}_{k_n}^t\oplus \ket{\Omega}\bra{\Omega})$ on $\bar{A}_{k_n}$ which applies the same map $\mathcal{N}_{k_n}^t:=\mathcal{N}_{k_n}$ at each time $t$ on non-vacuum states and acts as identity on vacuum states, would only change the non-vacuum term $\ket{i_{k_n}, t=2n+1}^{\bar{A}^O_{k_n}}$ of the above expression to $\ket{\mathcal{N}(i_{k_n}), t=2n+1}^{\bar{A}^O_{k_n}}$ while leaving the vacuum terms unchanged.  It is then easy to see that for any local operation $\mathcal{N}_1$ defined on a QC-QC input system $A_1^O$, we can define an operation $\bar{\mathcal{N}}_1$ as above, on the corresponding causal box input systems $\bar{A}_1^O$ such that applying $\mathcal{N}_1$ before the encoder is the same applying $\bar{\mathcal{N}}_1$ after the encoder, call this Property 1.

\begin{equation}
\label{eq: encoder_info_preserving}
 \mathrm{Enc}\Big(\mathcal{N}_1(\ket{\psi}_{A_1^O,...,A_N^O})\Big)=\bar{\mathcal{N}}_1\Big( \mathrm{Enc}(\ket{\psi}_{A_1^O,...,A_N^O})\Big), \quad \forall \ket{\psi}, 
\end{equation}

Next, we also consider the decoder and show that the following property, call this Property 2. If for some states $\rho$ and $\sigma$ on $A^O_1\otimes A^O_2$, we have
\begin{equation}
\label{eq: dec_prop2}
    \tr_{A^I_1}\circ \mathrm{Dec}\circ\mathcal{C}\circ\mathrm{Enc}(\rho)\neq  \tr_{A^I_1}\circ \mathrm{Dec}\circ\mathcal{C}\circ\mathrm{Enc}\circ (\sigma).   
\end{equation}
Then, we must also have

\begin{equation}
    \tr_{\bar{A}^I_1}\circ\mathcal{C}\circ\mathrm{Enc}(\rho)\neq  \tr_{\bar{A}^I_1}\circ\mathcal{C}\circ\mathrm{Enc}\circ (\sigma).   
\end{equation}
This is equivalent to saying that $\mathrm{Dec}$ cannot map states in the image of $\mathcal{C}\circ \mathrm{Enc}$, which are indistinguishable on $\bar{A}^I_2$ to states which are distinguishable on $A^I_2$. 
To show this, recall that our decoder is formed by composing a projector (trace decreasing map) with a rescaling (trace-increasing map) which yields an overall trace preserving map. It is clear from \cref{eq:scaling} that on the image of the previous operation (the projection), the rescaling only changes the constrant factors $\lambda$ (which are independent of the initial states $\rho, \sigma$) while relabelling the causal box system $\bar{A}_2^I$ to the QC-QC system $A_{k_n}^I$ (ignoring the time stamp of the non-vacuum state). Hence the rescaling map cannot make indistinguishable states on $\bar{A}^I_2$ distinguishable on $A^I_2$. For the projection map, notice that
the states on $\bar{A}_2^I$ that are annihilated by the projector are locally orthogonal to the states that are preserved by the projector, since because the projector projects on to the one message subspace of $\bar{A}_2^I$, while discarding the 0 and multiple message subspaces, which are orthogonal to the one-message space of $\bar{A}_2^I$. This means that every pair of states that can be distinguished on $\bar{A}_2^I$ after the projector has been applied, can also be distinguished before the projector has been applied\footnote{Simply by locally applying the projector on the one-message subspace of $\bar{A}_2^I$ and then performing the same distinguishing measurement.}. This establishes that Property 2 is satisfied.

To show that Properties 1 and 2 together give the desired implication, suppose $A^O_1$ signals to $A^I_2$ in $\mathcal{W}$, i.e. there exists $\mathcal{N}_1$ acting on $A^O_1$ such that
\begin{equation}
    \tr_{A^I_1}\circ \mathcal{W}\neq  \tr_{A^I_1}\circ \mathcal{W}\circ \mathcal{N}_1.
\end{equation}
From Step 1, of this overall proof, we know that $\mathcal{W}=\mathrm{Dec}\circ\mathcal{C}\circ\mathrm{Enc}$. Plugging this into the above, we have that there must exist a state $\rho$ on $A^O_1\otimes A^O_2$ such that
\begin{equation}
    \tr_{A^I_1}\circ \mathrm{Dec}\circ\mathcal{C}\circ\mathrm{Enc}(\rho)\neq  \tr_{A^I_1}\circ \mathrm{Dec}\circ\mathcal{C}\circ\mathrm{Enc}\circ \mathcal{N}_1(\rho).
\end{equation}
Notice that this is equivalent to \cref{eq: dec_prop2} for $ \sigma:=\Big(\mathcal{N}_1\otimes \mathcal{I}_{A^O_2}\Big)(\rho)$. 
Then, using Property 2 we have,

\begin{equation}
    \tr_{\bar{A}^I_1}\circ\mathcal{C}\circ\mathrm{Enc}(\rho)\neq  \tr_{\bar{A}^I_1}\circ\mathcal{C}\circ\mathrm{Enc}\circ \mathcal{N}_1(\rho).
\end{equation}

Next we use Property 1 to obtain

\begin{equation}
    \tr_{\bar{A}^I_1}\circ\mathcal{C}\circ\mathrm{Enc}\neq  \tr_{\bar{A}^I_1}\circ\mathcal{C}\circ\bar{\mathcal{N}}_1\circ \mathrm{Enc},
\end{equation}
for some $\bar{\mathcal{N}}_1$ on $\bar{A}^O_1$. From this it follows that 

\begin{equation}
    \tr_{\bar{A}^I_1}\circ\mathcal{C}\neq  \tr_{\bar{A}^I_1}\circ\mathcal{C}\circ\bar{\mathcal{N}}_1,
\end{equation}
which is equivalent to $\bar{A}^O_1$ signals to $\bar{A}^I_2$ in the causal box $\mathcal{C}$.


{\bf Step 4 (extension to full network):} Consider now, the whole network formed by the composition of $\mathcal{W}$ with the local operations of the QC-QC picture, and the corresponding network formed by composition $\mathcal{C}$ with associated local operations in the causal box picture. We need to show that every subnetwork of the former is associated with a corresponding subnetwork of the latter which is a fine-graining of it. In particular, the CPTP maps $\mathcal{W}$ and $\mathcal{C}$ are particular subnetworks where we have already shown that the latter is a fine-graining of the former in Steps 2 and 3. The proof of the current step mirrors a proof given in \cite{Vilasini_2022} for the quantum switch. We repeat this proof here for completeness. However, we do not review the full details of their framework here, and refer the reader to the original paper for further details on the definitions of networks, subnetworks and fine-graining thereof.

Networks are defined by a set of completely positive linear maps and a set of compositions connecting the output of one map to the input of another (or the same) map, while subnetworks of a given network are obtained by removing some of the maps and compositions. The QC-QC network is formed by composing the out/input systems $A_j^{I/O}$ of $\mathcal{W}$ with the corresponding in/output systems of the local operations. 
The subnetworks of this network can be classified into the following types: 
\begin{enumerate}
    \item[(1)] subnetworks formed only by the local operations and no composition with $\mathcal{W}$ 
    \item[(2)] subnetworks formed by composing only the local operations of a subset of agents, say $\{A_1,...,A_l\}$ for $l<N$ with $\mathcal{W}$ by connecting in and outputs $A_j^I$ and $A_j^O$ of the local operations and $\mathcal{W}$ for each $j\in \{1,...,l\}$ 
    \item[(3)] subnetworks formed by composing the local operations $A_i$ of a subset of agents, say $\{A_1,...,A_l\}$ for $l\leq N$ with $\mathcal{W}$ in sequence, e.g, $\mathcal{W}\circ \Big(\otimes_{i=1}^l A_i\Big)$ or $\Big(\otimes_{i=1}^l A_i\Big)\circ \mathcal{W}$.
    \item[(4)] subnetworks formed by a combination of (2) and (3) where some local operations are fully composed with $\mathcal{W}$ (both in and outputs connected) and some are sequentially composed.
\end{enumerate}


Since we have a clear correspondence between the in/output systems of the QC-QC network and those of the causal box network (just barred versions of the QC-QC systems), it is evident what the corresponding subnetworks in the causal box picture would look like. 
subnetworks of type (1) simply correspond to a tensor product of some subset of local operations. For simplicity, take the subnetwork corresponding to a single local operation $A_j$ of the $j^{th}$ agent in the QC-QC picture, the corresponding subnetwork in the causal box picture is the local operation $\bar{A}_j$ of the causal box picture (\cref{sec:statespace}). We can trivially construct an encoding and decoding scheme that recovers $A_j$ from $\bar{A}_j$ and preserves signalling relations, this encoder can act as an identity channel from the QC-QC input $A_j^I$ to the causal box input $\bar{A}_j^I$ for any fixed input time $t=2k$ (while preparing arbitrary states on the remaining times), the decoder can be an identity channel from the causal box output $\bar{A}_j^O$ at the output time $2k+1$ to the QC-QC output $A_j^O$ (while tracing out the systems $\bar{A}_j^O$ associated with all other times). Since $\bar{A}_j$ acts exactly as $A_j$ on non-vacuum states, we have the desired properties. Extending this argument analogously to any set of local operations, we can see that for all subnetworks of the form (1) in the QC-QC picture, the corresponding causal box subnetworks are its fine-grainings.

subnetworks of type (2) are simply reduced processes on the remaining $N-l$ agents, which are themselves QC-QCs when the original process is a QC-QC. Since the arguments of Step 2 and 3 apply to arbitrary QC-QCs, these immediately cover subnetworks of type (2). subnetworks of type (3) are also covered by using a similar construction as Step 2, noting that sequential composition of a QC-QC with a tensor product of a set of local operations, can always be considered as a new QC-QC. Since type (4) is a combination of types (2) and (3), this can also be covered through similar arguments. This proves that the causal box picture protocol defined by our extension is a fine-graining of the corresponding QC-QC picture protocol.

{\bf Step 5 (acyclic causal structure):} We do not provide a full definition of information-theoretic causal structure, as this will not be needed. Rather we only use a minimal property which is common across different definitions of this concept, the idea is that the connectivity of a quantum circuit (or more generally, a quantum information-theoretic network \cite{Vilasini_2022}) tells us about the absence of causal influences. This depends on the decompositions of the maps in the circuit as the same quantum channel $\C{M}$ can have two decompositions into smaller channels leading different connectivity between the in and outputs of $\C{M}$ (see \cite{Vilasini_2022} for further details). Relative to a given decomposition of an information-theoretic network, a system $S_1$ can be a cause of another system $S_2$ only if there is a path of wires from $S_1$ to $S_2$.\footnote{Note that the directionality comes from the fact that inputs precede outputs. For instance in a network formed by composing a map $\C{M}_1: \C{L}(\C{H}^A)\rightarrow \C{L}(\C{H}^B)$ to $\C{M}_2: \C{L}(\C{H}^B)\rightarrow \C{L}(\C{H}^C)$ sequentially via the $B$ system, we have no path of wires connecting the terminal system $C$ to the initial system $A$ and hence $C$ is not a cause of $A$ here.} 

We now show that the causal box extensions we have constructed have an acyclic causal structure. The causal box in the image of our mapping from QC-QCs always admits a sequence representation of a particular form, in terms of a totally ordered set $\C{T}$ (see \cref{fig:grenobleextensionl}, \cref{fig:another} and \cref{fig:another2}). The associated causal box network consists of the causal box (modelling the QC-QC) as depicted in the aforementioned figures, together with the boxes modelling the local operations and the latter always map inputs at time $t$ to outputs at $t+1>t$. Therefore this form of the sequence representation provides a decomposition of the whole causal box network which has the property that there is a wire connecting the output of one map to the input of another only if the time stamp associated with the former is earlier than that of the latter\footnote{W.l.o.g. for this argument, in extensions such as \cref{fig:another} where the control and ancillary wires are internal, we can associate with them the same time stamps as in an extension where they are external and joint with the corresponding target, such as \cref{fig:grenobleextensionl}. This ensures that all compositions, including internal ones, connect output wires with earlier time stamps to input wires with later time stamps.}. This means that there are no backwards in time causal influences, the only causal influences which exist in this decomposition flow forwards in time. This rules out causal cycles and establishes the existence of an explanation of the causal box network in terms of a well-defined and acyclic information-theoretic causal structure.

\section{Proofs of results}\label{sec:proofs}

\subsection{Proofs for the main part}

\compislink*

\begin{proof}
Using the definition of composition \cref{def:comp}, we can write

\begin{gather}
\begin{aligned}
    \C{M}^{C\hookrightarrow B} (\ket{\psi}^A \bra{\phi}^A) &= \sum_{k,l} \bra{k}^C\C{M}(\ket{\psi}^A\ket{k}^B\bra{l}^B\bra{\phi}^A)\ket{l}^C \\
    &= \sum_{k,l} \bra{k}^C(\C{M}_A \otimes \C{M}_B)(\ket{\psi}^A\ket{k}^B\bra{l}^B\bra{\phi}^A)\ket{l}^C \\
    &= \sum_{k, l} \bra{k}^C (M_A \otimes M_B)*(\ket{\psi}^A\ket{k}^B\bra{l}^B\bra{\phi}^A)\ket{l}^C \\
    &= \sum_{k, l} \underbrace{\bra{k}^C M_A \ket{l}^C}_{(M_A)_{kl}} * (\ket{\psi}^A \bra{\phi}^A) \otimes \sum_{i, j} \underbrace{\bra{i}^B M_B \ket{j}^B}_{(M_B)_{ij}} \underbrace{\braket{i|k}^B\braket{l|j}^B}_{\delta_{ik} \delta_{lj}} \\
    &= \sum_{k,l} (M_A)_{kl} \otimes (M_B)_{kl} * (\ket{\psi}^A \bra{\phi}^A) \\
    &\stackrel{\text{\cref{def:linkmat}}}{=} (M_A*M_B) * (\ket{\psi}^A \bra{\phi}^A).
\end{aligned}
\end{gather}

In the third line, we used that $\C{M}_{A/B}(\rho) = M_{A/B}*\rho$ and in the fourth line we used commutativity of the link product and the second sum is the explicit form of the link product $M_B*(\ket{k}^B \bra{l}^B)$. Finally, in the last line, we used the definition of the link product, taking $\C{H}^B$ and $\C{H}^C$ to be the same Hilbert space. We can then also write $\C{M}^{C\hookrightarrow B} (\ket{\psi}^A \bra{\phi}^A) = M^{C\hookrightarrow B} * \ket{\psi}^A \bra{\phi}^A$. As these equations must hold for all $\ket{\psi}^A, \ket{\phi}^A \in \C{H}^A$, we can drop these states and obtain

\begin{equation}
    M^{C\hookrightarrow B} = M_A*M_B.
\end{equation}

This shows that sequential composition can be expressed as the link product, at least in the case of finite-dimensional Hilbert spaces.


\end{proof}

\anotheriso*

\begin{proof}
A direct sum of isometries on orthogonal subspaces is again an isometry if their images are orthogonal (cf. \cref{sec:isometries}). This is the case here, as each isometry maps the $m$-messages space on the agents' outputs to the $m$-messages space on the agents' inputs and states with different numbers of messages are orthogonal to each other. The maps $\bar{V}_{n+1}^m$ with $m\neq 1$ are isometries by assumption. We thus only need to show that $\bar{V}_{n+1}^1$ is an isometry. Note that $\bar{V}_{n+1}^1 = \bar{V}_{n+1}^1|_{\C{H}^{O_n}_{\text{corr}}} \oplus \bar{V}_{n+1}^1|_{\C{H}^{O_n}_{\text{mis}}}$ and $\bar{V}_{n+1}^1|_{\C{H}^{O_n}_{\text{corr/mis}}}: \C{H}^{O_n}_{\text{corr/mis}} \rightarrow \C{H}^{I_{n+1}}_{\text{corr/mis}}$. Again by definition, we have that $\bar{V}_{n+1}^1|_{\C{H}^{O_n}_{\text{mis}}}$ is an isometry. From \cref{eq:anotherone}, we can see that $\bar{V}^1_{n+1}|_{\C{H}^{O_n}_{\text{corr}}} = V_{n+1} \otimes \proj{t=2n+2}{t=2n+1}$, thus

\begin{equation}
    \bar{V}^{1 \dagger}_{n+1}|_{\C{H}^{O_n}_{\text{corr}}} \bar{V}^1_{n+1}|_{\C{H}^{O_n}_{\text{corr}}} = V^\dagger_{n+1} V_{n+1} \otimes \proj{t=2n+1}{t=2n+1} = \mathbb{1} \otimes \proj{t=2n+1}{t=2n+1}
\end{equation}

which is the identity on the input space of $\bar{V}^1_{n+1}|_{\C{H}^{O_n}_{\text{corr}}}$, meaning the map is an isometry. Therefore, $\bar{V}_{n+1}^1$ is also an isometry.

We can construct a particular example of such isometries by considering a map that updates the control by adding the minimum of $\C{N} \backslash \{\C{K}_{n-1} \cup k'_n\}$ to $\C{K}_{n-1}$ as well as adding it to the ancilla while acting as the identity on everything else (upto subsystem and timestamp labelling),

\begin{gather}\label{eq:anotherexample}
\begin{aligned}
\bar{V}_{n+1}^m& ((\bigotimes_{k_n} \bigodot_i \ket{\psi^{k_n}_i}^{\bar{A}^O_{k_n}}) \ket{\C{K}_{n-1}, k'_n} \ket{\alpha}^{\alpha_n}) \\
&= (\bigotimes_{k_n} \bigodot_i \ket{\psi^{k_n}_i}^{\bar{A}^I_{k_n}}) \ket{\C{K}_{n-1} \cup \text{min}(\C{N} \backslash \C{K}_{n}), k'_n} \ket{\alpha, \text{min}(\C{N} \backslash \C{K}_n)}^{\alpha_{n+1}}
\end{aligned}
\end{gather}

where $\C{K}_n = \C{K}_{n-1} \cup k'_n$ and $\inprod{\alpha, k}{\beta, l} = \delta_{k,l} \inprod{\alpha}{\beta}$.\footnote{We can always assume that the ancillas are large enough for this definition to be well defined. This is because we can always enlarge an ancilla with additional states. Since these states never get populated in a normal run of the QC-QC, this does not change the QC-QC. At the same time, we can assume that the internal operations are defined on these additional states (see \cref{remark:fullspace} for details) which is necessary for \cref{eq:anotherone} to be well defined.} For the one-message mismatched subspace, we analogously have the following (where $k_n\neq k'_n$)

\begin{gather}
\begin{aligned}
\bar{V}_{n+1}^1& (\ket{\psi^{k_n}_i}^{\bar{A}^O_{k_n}} \ket{\C{K}_{n-1}, k'_n} \ket{\alpha}^{\alpha_n}) \\
&= \ket{\psi^{k_n}_i}^{\bar{A}^I_{k_n}} \ket{\C{K}_{n-1} \cup \text{min}(\C{N} \backslash \C{K}_n), k'_n} \ket{\alpha, \text{min}(\C{N} \backslash \C{K}_n)}^{\alpha_{n+1}}
\end{aligned}
\end{gather}

It is enough to show that the inner product of the ancilla and control is conserved,

\begin{gather}
\begin{aligned}
    &\inprod{\C{K}_{n-1} \cup \text{min}(\C{N} \backslash \C{K}_n), k'_n}{\C{L}_{n-1} \cup \text{min}(\C{N} \backslash \C{K}_n), l'_n} \inprod{\alpha, \text{min}(\C{N} \backslash \C{K}_n) \cup k'_n\}}{\beta, \text{min}(\C{N} \backslash \C{L}_n)}  \\
    &=\delta_{\C{K}_{n-1} \cup \text{min}(\C{N} \backslash \C{K}_n),\C{L}_{n-1} \cup \text{min}(\C{N} \backslash \C{L}_n)} \delta_{k'_n, l'_n} \inprod{\alpha}{\beta} \delta_{\text{min}(\C{N} \backslash \C{K}_n), \text{min}(\C{N} / \C{L}_n)} \\
    &= \delta_{\C{K}_{n-1}, \C{L}_{n-1}} \delta_{k'_n, l'_n} \inprod{\alpha}{\beta} \\
    &= \inprod{\C{K}_{n-1}, k'_n}{\C{L}_{n-1}, l'_n} \inprod{\alpha}{\beta}.
\end{aligned}
\end{gather}
\end{proof}

\anotherequivalence*

\begin{proof}
    Using \cref{eq:directsum} and orthogonality between states corresponding to different numbers of messages, we obtain

    \begin{equation}
        \dket{\bar{V}_{N+1}}*...*\dket{\bar{V}_1} = \sum_{m=0}^\infty \dket{\bar{V}^m_{N+1}} *...*\dket{\bar{V}^m_1}.
    \end{equation}

    Now note that $\dket{\bar{V}^m_1} * \ket{\psi, t=1}^P$ is 0 unless $m=1$, therefore the only relevant term in the above equation when we compose with agents fulfilling the assumptions from \cref{sec:statespace} is the one corresponding to $m=1$. Thus, using the form of the local operations are extensions of the local operations in the QC-QC picture, we have


    \begin{gather}
    \begin{aligned}
        \bigotimes_{k=1}^N &\dket{\bar{A}_k} \otimes \ket{\psi, t=1}^P * \dket{\bar{V}^1_{N+1}} *...*\dket{\bar{V}^1_1} \\
        =& \bigotimes_{k=1}^N \bigotimes_{n=1}^N (\dket{A_k} + \ket{\Omega}^{\bar{A}^O_{k}} \ket{\Omega}^{\bar{A}^I_{k}}) \otimes \ket{2n+1} \ket{2n+2} *\dket{\bar{V}^1_{N+1}} *...*\dket{\bar{V}^1_1} * \ket{\psi, t=1}^P\\
        =& \sum_{(k_1,...k_N)} \bigotimes_{n=1}^N \dket{A_{k_n}} \otimes \ket{2n+1} \ket{2n+2} * (\dket{V_{N+1}} \otimes \ket{2n+1} \ket{2n+2} * ... * \dket{V_1} \otimes \ket{1} \ket{2}) \ket{\psi, t=1}^P \\
        =& \sum_{(k_1,...,k_N)} \dket{V_{N+1}} * \dket{A_{k_N}} * ... * \dket{A_{k_1}} * \dket{V_1} * \ket{\psi, t=1}^P\\
        =& \sum_{(k_1,...,k_N)} \dket{V_{\C{K}_{N-1}, k_N}^{\rightarrow F}} * \dket{A_{k_N}} * ... * \dket{A_{k_1}} * \dket{V_{\emptyset, \emptyset}^{\rightarrow k_1}} * \ket{\psi, t=1}^P\\ 
        =& \bigotimes_{k=1}^N \dket{A_k} \otimes \ket{\psi}^P * \dket{V_{N+1}} * ... * \dket{V_1}.
    \end{aligned}
    \end{gather}

    In the second equality, we expanded the product and used that only terms with exactly one non-vacuum operator for each time step contribute. Note that the fact that there is no contribution from the terms with a control mismatch can be seen as a consequence of \cref{lemma:accept}. This shows that the causal box is an extension of the QC-QC.

    
\end{proof}

\accept*

\begin{proof}
    We can show by induction that this is always the case. The global past is assumed to send a single message during the first time step and initialises the control correctly, thus the base case is trivial. Assume now the projective measurements $\{\C{P}^{I_k/O_k}_{accept}, 1 - \C{P}^{I_k/O_k}_{accept}\}$ for $k < n+1$ all yield $accept$. The internal operation thus acts effectively as $\bar{V}_{n+1} \circ \C{P}^{O_n}_{accept} = \bar{V}^1_{n+1}|_{\C{H}^{O_n}_{\text{corr}}}$, which outputs a state in $\C{H}^{I_{n+1}}_{\text{corr}}$ by definition (\cref{eq:anotherone}) and thus $\{\C{P}^{I_{n+1}}_{accept}, 1 - \C{P}^{I_{n+1}}_{accept}\}$ yields $accept$. The local operations conserve the number of messages and cannot cause a control mismatch and so $\{\C{P}^{O_{n+1}}_{accept}, 1 - \C{P}^{O_{n+1}}_{accept}\}$ must also give that outcome. 
\end{proof}

\subsection{Proofs for the appendix}\label{sec:proofapp}

\begin{restatable}[Isometry condition for internal operations]{lemma}{krausiso}
\label{lemma:krausiso}
For $\C{K}_{n-1}, \C{L}_{n-1} \subsetneq \C{N}$ and $k_n \in \C{N} \backslash \C{K}_{n-1}, l_n \in \C{N} \backslash \C{L}_{n-1}$ with $\C{K}_{n-1} \cup k_n = \C{L}_{n-1} \cup l_n$, it holds that

\begin{equation}
\sum_{k_{n+1}} V^{\rightarrow k_{n+1} \dagger}_{\C{K}_{n-1}, k_n} V^{\rightarrow k_{n+1}}_{\C{L}_{n-1}, l_n} = \mathbb{1}^{A^O_{k_{n}} \alpha_n} \delta_{\C{K}_{n-1}, \C{L}_{n-1}} \delta_{k_n, l_n}.
\end{equation}
\end{restatable}

\begin{proof}
The lemma is a direct consequence of the maps $V_{n+1}: \C{H}^{A^O \alpha_n C_{n}} \rightarrow \C{H}^{A^I \alpha_{n+1} C_{n+1}}$ in the QC-QC picture being isometries, 

\begin{equation}\label{eq:QCQCiso}
V_{n+1}^\dagger V_{n+1} = \mathbb{1}^{A^O \alpha_n C_{n}} = \sum_{\C{K}_{n-1}, k_n} \mathbb{1}^{A^O_{k_n} \alpha_n} \otimes \ket{\C{K}_{n-1}, k_n} \bra{\C{K}_{n-1}, k_n}.
\end{equation}

Using the definition of $V_{n+1}$ given in \cref{eq:defiso},

\begin{gather}
\begin{aligned}
    V_{n+1}^\dagger V_{n+1} &= \sum_{\substack{\C{K}_{n-1},\\ k_n, k_{n+1}}} \sum_{\substack{\C{L}_{n-1},\\ l_n, l_{n+1}}} V^{\rightarrow k_{n+1} \dagger}_{\C{K}_{n-1}, k_n} V^{\rightarrow l_{n+1}}_{\C{L}_{n-1}, l_n} \otimes \ket{\C{K}_{n-1}, k_n} \braket{\C{K}_{n-1} \cup k_n, k_{n+1}|\C{L}_{n-1} \cup l_n, l_{n+1}} \bra{\C{L}_{n-1}, l_n} \\
    &= \sum_{\C{K}_n} \sum_{k_n, l_n \in \C{K}_n} \sum_{k_{n+1}} V^{\rightarrow k_{n+1} \dagger}_{\C{K}_n \backslash k_n, k_n} V^{\rightarrow k_{n+1}}_{\C{K}_n \backslash l_n, l_n} \otimes \ket{\C{K}_n \backslash k_n, k_n} \bra{\C{K}_n \backslash l_n, l_n}.
\end{aligned}
\end{gather}

Comparing this to \cref{eq:QCQCiso}, we notice that terms with an off-diagonal projector over the control $\ket{\C{K}_n \backslash k_n, k_n} \bra{\C{K}_n \backslash l_n, l_n}$, $k_n \neq l_n$, must vanish. This is only possible if $\sum_{k_{n+1}} V^{\rightarrow k_{n+1} \dagger}_{\C{K}_n \backslash k_n, k_n} V^{\rightarrow k_{n+1}}_{\C{K}_n \backslash l_n, l_n} = 0$ for $k_n \neq l_n$. Meanwhile, terms where $k_n = l_n$ must be the identity which implies 

\begin{equation}
    \sum_{k_{n+1}} V^{\rightarrow k_{n+1} \dagger}_{\C{K}_n \backslash k_n, k_n} V^{\rightarrow k_{n+1}}_{\C{K}_n \backslash k_n, k_n} = \mathbb{1}^{A^O_{k_n} \alpha_n}.
\end{equation}

Combining these, we obtain the statement of the lemma.

\end{proof}

\symmiso*

\begin{proof}
    We use the convention of \cref{eq:symmprod} for the symmetric tensor product. The proof relies on the fact that the inner product is invariant if one replaces vectors of the form 

    \begin{equation}\label{eq:tenpsi}
        \bigotimes_{k} \ket{\psi_k}^{A_k} 
    \end{equation}

    with vectors 

    \begin{equation}\label{eq:symmpsi}
        \bigodot_{k} \ket{\psi_k}^{A_k}. 
    \end{equation}

    This follows from the fact that the different spaces associated to different agents are orthogonal. We can also show this by explicitly calculating the two inner products (using \cref{eq:fock_inner_n})
    
    \begin{gather}
    \begin{aligned}
        \bigodot_k \bra{\psi_k}^{A_k} \bigodot_l \ket{\phi_l}^{A_l} &= \sum_{\pi} \prod_k \bra{\psi_k}^{A_k} \ket{\phi_l}^{A_\pi(l)} \\
        &= \sum_{\pi} \prod_k \braket{\psi_k}{\phi_k} \delta_{k, \pi(k)} \\
        &= \prod_k \braket{\psi_k}{\phi_k} \\
        &= \bigotimes_k \bra{\psi_k}^{A_k} \bigotimes_l \ket{\psi_l}^{A_l}
    \end{aligned}
    \end{gather}

    where we have used the orthogonality of the different Hilbert spaces.
    
    
    Due to \cref{lemma:isochar}, it suffices to consider the inner products of elementary tensors with a fixed but arbitrary number of messages $M$

    \begin{equation}\label{eq:gen_vec}
        \bigotimes_{k_n} \bigodot_{i=i_{k_n}}^{i_{k_n+1}-1} \ket{\psi^i}^{\bar{A}_{k_n}} \ket{\C{K}^i_{n-1}, k_n}
    \end{equation}

    where $1 = i_1 < i_2 < ... < i_{N+1} = M+1$ and we dropped the time stamps to reduce clutter. This is because the subspaces corresponding to different numbers of messages are orthogonal to each other and the operators conserve the number of messages, allowing us to invoke condition 4 of \cref{lemma:isochar}, while the elementary tensors span the whole space, therefore containing in particular a basis (condition 3).

    Replacing the tensor product in \cref{eq:gen_vec} with a symmetric tensor product and using \cref{eq:fock_inner_n}, the inner product of two such tensor products is

    \begin{equation}
        \sum_{\pi \in S^M} \prod_i \braket{\psi^i|\phi^{\pi(i)}} \braket{\C{K}^i_{n-1}, k_n^i|\C{L}^i_{n-1}, l_n^{\pi(i)}}
    \end{equation}

    where we added indexes $i$ to $k_n$ and $l_n$ so that there is only one product over $i$ instead of one over the same and one over the agents' labels. 
    
    On the other hand, applying $\bar{V}_{n+1}$ to \cref{eq:gen_vec}, we obtain

    \begin{equation}
        symm(\bigotimes_{k_n} \bigotimes_{i=i_{k_n}}^{i_{k_n+1}-1} \sum_{k^i_{n+1}} \bar{V}^{\rightarrow k^i_{n+1}}_{\C{K}^i_{n-1}, k_n} \ket{\psi^i}^{\bar{A}_{k_n}} \ket{\C{K}^i_{n-1} \cup k_n, k^i_{n+1}}).
    \end{equation}

    For the inner product, we can once again use tensors of the form

    \begin{equation}
        \bigodot_i \sum_{k^i_{n+1}} \bar{V}^{\rightarrow k^i_{n+1}}_{\C{K}^i_{n-1}, k_n^i} \ket{\psi^i}^{\bar{A}_{k_n}} \ket{\C{K}^i_{n-1} \cup k_n^i, k^i_{n+1}}).
    \end{equation}

    which yields

    \begin{gather}
    \begin{aligned}
    \sum_{\pi \in S^M} & \prod_i \sum_{k^i_{n+1}, l^{\pi(i)}_{n+1}} \bra{\psi^i} \bar{V}^{\rightarrow k^i_{n+1} \dagger}_{\C{K}^i_{n-1}, k_n^i} \bar{V}^{\rightarrow l^{\pi(i)}_{n+1}}_{\C{L}^{\pi(i)}_{n-1}, l_n^\pi(i)} \ket{\phi^{\pi(i)}} \braket{\C{K}^i_{n-1} \cup k_n^i, k_{n+1}^i|\C{L}^{\pi(i)}_{n-1} \cup l_n^{\pi(i)}, l_{n+1}^{\pi(i)}} \\
    &= \sum_{\pi \in S^M} \prod_i \sum_{k^i_{n+1}} \bra{\psi^i} \bar{V}^{\rightarrow k^i_{n+1} \dagger}_{\C{K}^i_{n-1}, k_n^i} \bar{V}^{\rightarrow k^i_{n+1}}_{\C{L}^{\pi(i)}_{n-1}, l_n^\pi(i)} \ket{\phi^{\pi(i)}} \braket{\C{K}^i_{n-1} \cup k_n^i, k_{n+1}^i|\C{L}^{\pi(i)}_{n-1} \cup l_n^{\pi(i)}, k^i_{n+1}} \\
    &= \sum_{\pi \in S^M} \prod_i \sum_{k^i_{n+1}} \bra{\psi^i} \bar{V}^{\rightarrow k^i_{n+1} \dagger}_{\C{K}^i_{n-1}, k_n^i} \bar{V}^{\rightarrow k^i_{n+1}}_{\C{L}^{\pi(i)}_{n-1}, l_n^\pi(i)} \ket{\phi^{\pi(i)}} \delta_{\C{K}^i_{n-1} \cup k_n^i, \C{L}^{\pi(i)}_{n-1}\ \cup l_n^{\pi(i)}} \\
    &\stackrel{\text{\cref{lemma:krausiso}}}{=} \sum_{\pi \in S^M} \prod_i \braket{\psi^i|\phi^{\pi(i)}} \delta_{\C{K}^i_{n-1}, \C{L}^{\pi(i)}_{n-1}} \delta_{k_n^i, l_n^\pi(i)} \delta_{\C{K}^i_{n-1} \cup k_n^i, \C{L}^{\pi(i)}_{n-1}\ \cup l_n^{\pi(i)}} \\
    &= \sum_{\pi \in S^M} \prod_i \braket{\psi^i|\phi^{\pi(i)}} \braket{\C{K}^i_{n-1}, k_n^i|\C{L}^i_{n-1}, l_n^{\pi(i)}}
    \end{aligned}
    \end{gather}

   where in the second equality we used that if the target spaces are different (i.e. $k^i_{n+1} \neq l^{\pi(i)}_{n+1}$), the states are orthogonal.




   
\end{proof}

\symmequivalence*

\begin{proof}
W.l.o.g., we assume the global past to be trivial. If we have a QC-QC with non-trivial global past, compose this QC-QC with some arbitrary state $\rho^P \in \C{L}(\C{H}^P)$ to obtain a QC-QC with trivial global past and compose the causal box we obtain from the QC-QC with $\rho^P \otimes \proj{t=1}{t=1}$. Note that the behaviour of the original QC-QC and causal box are thus completely determined by QC-QCs and causal boxes with trivial past since $\rho^P$ is arbitrary. Thus, it is enough to prove the statement in the case of trivial global past. Further instead of considering the Choi vector, we can consider its projection to the space

\begin{gather}
\begin{aligned}\label{eq:effHeff}
    \text{Span}[\{\bigotimes_{n=1}^N \ket{i_{k_n}, 2n}^{\bar{A}^I_{k_n}} \otimes \ket{\Omega, (2n)^c}^{\bar{A}^I_{k_n}} \otimes  \ket{j_{k_n}, 2n+1}^{\bar{A}^O_{k_n}} \otimes \ket{\Omega, (2n+1)^c}^{\bar{A}^O_{k_n}}
    \otimes \ket{i_F, 2N+2}^{\bar{F}}\}_{\substack{i_{k_1} \neq \Omega; \\i_{k_n}, j_{k_n}: \\ i_{k_n} = \Omega \\ \implies j_{k_n} = \Omega}}.
\end{aligned}
\end{gather}

This space is obtained by imposing the following constraints: firstly, at any time there is a message on at most a single wire. Secondly, at $t=2$, there cannot be vacuum on all input wires. Thirdly, if the input wire of an agent during one time step carries the vacuum, so does their output wire during the next time step. The first and second constraints follow from the fact that the causal boxes as we defined them conserve the number of messages. Additionally, since we assumed the global past to be trivial, we have that $V_1$ is simply a state (potentially in superposition on different wires), i.e. the number of messages is exactly 1. The third constraint is justified by the fact that we assume that the local agents only output a message at $t=2n+1$ if they received a message at $t=2n$. Thus, terms that do not fulfill this constraint simply vanish when composing with agents.

Let us now consider a basis state from \cref{eq:effHeff} and see what happens when we compose it with the pure local operations from \cref{def:local_equivalence}

\begin{gather}\label{eq:basiscomp}
\begin{aligned}
    \bigotimes_{k=1}^N &\dket{\bar{A}_k} *\bigotimes_{n=1}^N \ket{i_{k_n}, t=2n}^{\bar{A}^I_{k_n}} \ket{\Omega, (2n)^c}^{\bar{A}^I_{k_n}} \ket{j_{k_n}, t=2n+1}^{\bar{A}^O_{k_n}} \ket{\Omega, (2n+1)^c}^{\bar{A}^O_{k_n}} \ket{i_F, t=2N+2}^{\bar{F}} \\
    =& (\bigotimes_{n=1}^N \bigotimes_{t \in \C{T}^I} (\dket{A_n} + \ket{\Omega}^{\bar{A}^I_{k_n}} \ket{\Omega}^{\bar{A}^O_{k_n}}) \otimes \ket{t+1} \ket{t}) * \\
    &\bigotimes_{n=1}^N \ket{i_{k_n}, t=2n}^{\bar{A}^I_{k_n}} \ket{\Omega, (2n)^c}^{\bar{A}^I_{k_n}}  \ket{j_{k_n}, t=2n+1}^{\bar{A}^O_{k_n}} \ket{\Omega, (2n+1)^c}^{\bar{A}^O_{k_n}}
    \ket{i_F, t=2N+2}^{\bar{F}} \\
    =& \prod_{n=1}^N (\dket{A_{k_n}} + \ket{\Omega}^{\bar{A}^I_{k_n}} \ket{\Omega}^{\bar{A}^O_{k_n}})\\
    &* (\ket{i_{k_n}}^{\bar{A}^I_{k_n}} \ket{j_{k_n}}^{\bar{A}^O_{k_n}}) ((\ket{\Omega}^{\bar{A}^I_{k_n}} \ket{\Omega}^{\bar{A}^O_{k_n}}) * (\ket{\Omega}^{\bar{A}^I_{k_n}} \ket{\Omega}^{\bar{A}^O_{k_n}}))^3 \ket{i_F, t=2N+2} \\
    =& \begin{cases}
    \prod_{n=1}^N \dket{A_{k_n}} * (\ket{i_{k_n}}^{A^I_{k_n}} \ket{j_{k_n}}^{A^O_{k_n}}) \ket{i_F, t=2N+2}, &\text{ if } i_{k_n}, j_{k_n} \neq \Omega, \forall n \\ 
    0, &\text{ else} 
    \end{cases}
\end{aligned}
\end{gather}

In the first equality, we used the definition of extensions of local operations. In the second equality, we contracted the time stamps and also used that $\dket{A_{k_n}} * (\ket{\Omega}^{\bar{A}^I_{k_n}} \otimes \ket{\Omega}^{\bar{A}^O_{k_n}}) = 0$ for all $k_n$. Finally, we used that both the local and internal operations map non-vacuum states to non-vacuum states. The contraction thus vanishes unless all $i_{k_n}, j_{k_n}$ are either the vacuum or the non-vacuum. However, as we stated earlier, $i_{k_1} \neq \Omega$ and so only the latter is actually possible.

Additionally, note that there is an obvious isomorphism between the global future state in the above equation and the global future in the QC-QC, namely $\ket{i_F, t=2N+2} \cong \ket{i_F}$. We will thus drop the time stamp and the vacuum states of the global future from here on out.

The coefficient of the basis state in \cref{eq:basiscomp} can be found using the definition of the Choi vector. The Choi vector of some operator $T: \C{H}^X \rightarrow \C{H}^Y$ is $T \dket{\mathbb{1}}^{XX} = \sum_i \ket{i}^X T \ket{i}^X = \sum_{ij} \ket{i}^X \ket{j}^Y \bra{j}^Y T \ket{i}^X = \sum_{ij} \bra{j} T \ket{i} \ket{i}^X \ket{j}^Y$. Applied to our basis state and with the help of \cref{eq:photonic_single}, we find the coefficient to be

\begin{equation}\label{eq:coeff}
    \bra{i_F} V^{\rightarrow F}_{\C{N}\backslash k_N, k_N} \ket{j_{k_N}} * \bigstar_{n=1}^{N-1} \bra{i_{k_{n+1}}} V^{\rightarrow k_{n+1}}_{\C{K}_{n-1}, k_n} \ket{j_{k_n}} * \bra{i_{k_1}} V^{\rightarrow k_1}_{\emptyset, \emptyset}
\end{equation}

where we used the notation $\bigstar_{n=1}^N A_n = A_1 * A_2 * ... * A_N$. 

With \cref{eq:basiscomp} and \cref{eq:coeff}, we can now calculate the contraction $\bigotimes_{k=1}^N \dket{\bar{A}_k} * \dket{\bar{V}}$. 

\begin{gather}
\begin{aligned}
    \bigotimes_{k=1}^N &\dket{\bar{A}_k} * \dket{\bar{V}} \\
    \cong & \sum_{(k_1,...,k_N)} \sum_{\substack{i_1,..., i_N, \\j_1,...., j_N, i_F}} \bra{i_F} V^{\rightarrow F}_{\C{N}\backslash k_N, k_N} \ket{j_{k_N}} * \bigstar_{n=1}^{N-1} \bra{i_{k_{n+1}}} V^{\rightarrow k_{n+1}}_{\C{K}_{n-1}, k_n} \ket{j_{k_n}} * \bra{i_{k_1}} V^{\rightarrow k_1}_{\emptyset, \emptyset} \\
    &\prod^{N}_{n=1} \dket{A_{k_n}} * (\ket{i_{k_n}}^{A^I_{k_n}} \ket{j_{k_n}}^{A^O_{k_n}}) \ket{i_F} \\
    =& \sum_{(k_1,...,k_N)} \bigstar_{n=1}^{N} (\dket{A_{k_n}} * \sum_{i_{k_{n+1}}, j_{k_n}} \bra{i_{k_{n+1}}} V^{k_{n+1}}_{\C{K}_{n-1}, k_n} \ket{j_{k_n}} (\ket{i_{k_n}}^{A^I_{k_n}} \ket{j_{k_n}}^{A^O_{k_n}})) \\
    =& \sum_{(k_1,...,k_N)} \bigstar_{n=1}^{N} \dket{A_{k_n}} * \dket{V^{k_{n+1}}_{\C{K}_{n-1}, k_n}} \\
    =&  \bigotimes_{k=1}^N \dket{A_k} * \ket{w}.
\end{aligned}
\end{gather}

In the above equation, one should interpret $k_{N+1} = F$ and $\ket{j_{k_0}} = 1$. In the third equality, we used again that $\dket{T} = \sum_{ij} \bra{i} T \ket{j} \ket{i}^X \ket{j}^Y$. In the last equality, we used commutativity of the link product and the definition of the process vector of the QC-QC. 

This shows that we have found a causal box extension of the QC-QC. 
\end{proof}

\end{document}